%% file: ms.tex
\pdfoutput=1
\documentclass[11pt]{article}
\input{preamble}
\hypersetup{pdftitle={A Mean Field Model of Consumption Inertia and Keeping Up with the Joneses},pdfauthor={Junkee Jeon}}

\begin{document}

\hypersetup{pageanchor=false}
\begin{titlepage}
\begin{center}
{\LARGE\bfseries A Mean Field Model of Consumption Inertia\\ and Keeping Up with the Joneses\par}
\vspace{1.6em}
{\large Junkee Jeon\par}
\vspace{0.4em}
{\normalsize Department of Applied Mathematics, Kyung Hee University, Republic of Korea\\
\href{mailto:junkeejeon@khu.ac.kr}{junkeejeon@khu.ac.kr}\par}
\vspace{0.8em}
{\normalsize October 2026\par}
\end{center}
\vspace{1em}
\noindent\textbf{Abstract.}
\input{abstract}

\vspace{1em}

\vspace{0.6em}
\end{titlepage}

\setcounter{page}{1}\hypersetup{pageanchor=true}
\input{sec1_introduction}
\input{sec2_model}
\input{sec3_implications}
\input{sec4_conclusion}

\appendix
\input{appendixA}

\counterwithin*{equation}{section}
\renewcommand{\theequation}{\thesection.\arabic{equation}}
\renewcommand{\theHequation}{\thesection.\arabic{equation}}
\counterwithin*{theorem}{section}\renewcommand{\thetheorem}{\thesection.\arabic{theorem}}\renewcommand{\theHtheorem}{\thesection.\arabic{theorem}}
\counterwithin*{proposition}{section}\renewcommand{\theproposition}{\thesection.\arabic{proposition}}\renewcommand{\theHproposition}{\thesection.\arabic{proposition}}
\counterwithin*{lemma}{section}\renewcommand{\thelemma}{\thesection.\arabic{lemma}}\renewcommand{\theHlemma}{\thesection.\arabic{lemma}}
\counterwithin*{corollary}{section}\renewcommand{\thecorollary}{\thesection.\arabic{corollary}}\renewcommand{\theHcorollary}{\thesection.\arabic{corollary}}
\counterwithin*{definition}{section}\renewcommand{\thedefinition}{\thesection.\arabic{definition}}\renewcommand{\theHdefinition}{\thesection.\arabic{definition}}
\counterwithin*{remark}{section}\renewcommand{\theremark}{\thesection.\arabic{remark}}\renewcommand{\theHremark}{\thesection.\arabic{remark}}
\counterwithin*{figure}{section}\renewcommand{\thefigure}{\thesection.\arabic{figure}}\renewcommand{\theHfigure}{\thesection.\arabic{figure}}
\counterwithin*{table}{section}\renewcommand{\thetable}{\thesection.\arabic{table}}\renewcommand{\theHtable}{\thesection.\arabic{table}}

\input{ecompanion}

\input{supplement_body}

\input{references}
\end{document}

%% file: preamble.tex
\usepackage[T1]{fontenc}
\usepackage{amsmath}

\usepackage{txfonts}

\usepackage{amsthm,mathtools,bm}
\usepackage[margin=1in]{geometry}
\usepackage{setspace}
\usepackage{graphicx,enumitem}
\usepackage[authoryear,round]{natbib}
\setcitestyle{citesep={,},aysep={}}
\usepackage[font=small,labelfont=bf,labelsep=period]{caption}
\usepackage{placeins}
\usepackage{xcolor}
\definecolor{linkink}{HTML}{1F3F5B}
\usepackage{xr-hyper}
\usepackage[colorlinks=true,linkcolor=linkink,citecolor=linkink,urlcolor=linkink]{hyperref}

\theoremstyle{plain}
\newtheorem{theorem}{Theorem}
\newtheorem{proposition}{Proposition}
\newtheorem{lemma}{Lemma}
\newtheorem{corollary}{Corollary}
\theoremstyle{definition}
\newtheorem{definition}{Definition}
\newtheorem{remark}{Remark}

\newcommand{\E}{\mathbb{E}}
\newcommand{\R}{\mathbb{R}}
\newcommand{\F}{\mathcal{F}}
\newcommand{\FB}{\mathbb{F}^B}
\newcommand{\dd}{\mathrm{d}}
\newcommand{\ind}{\mathbf{1}}
\newcommand{\hg}{\hat\gamma}
\newcommand{\Xf}{\mathcal{X}}
\newcommand{\za}{z_\alpha}
\newcommand{\zb}{z_\beta}
\newcommand{\la}{l^{\alpha}}
\newcommand{\lb}{l^{\beta}}
\newcommand{\Rr}{\mathcal{R}}
\newcommand{\J}{\mathcal{J}}
\newcommand{\vb}{\bar\varphi}
\newif\iftablesatend
\tablesatendfalse

%% file: abstract.tex
Households compare their living standard with others' and dislike changing
it. We study consumption and investment when felicity depends on consumption
relative to a benchmark, the population's geometric average consumption,
and cuts in felicity weigh more than increases. Asset returns are exogenous.
For finitely many types and arbitrary initial wealth and inherited
standards, a mean field equilibrium exists under a weak interaction
condition, and for fixed budget multipliers the benchmark is the fixed point
of a contraction. In a complete market with common risk, a mean field
equilibrium at empirical type frequencies is an exact Nash equilibrium of
the finite game, for any preferences over relative consumption. We bound
the gain from deviating when plans computed at another composition are
used. We compute equilibria in which
groups with different wealth, inherited standards, and loss aversion revise
spending at different dates. With identical households and common risk, the
unique symmetric equilibrium is the policy of a single loss averse household
with modified curvature. In a baseline example, ignoring the inherited
standard costs up to 17\% of wealth, while small errors about the
composition of the reference group cost less than 0.5\%.

%% file: sec1_introduction.tex
\section{Introduction}\label{sec:intro}

\citet{Duesenberry} proposed that a household's consumption depends both on
what other households consume and on the highest standard of living the
household has reached in the past. The first idea lives on in models of
keeping up with the Joneses \citep{Abel,Gali,LU} and in work on how status
concerns shape portfolios \citep{DKK,Roussanov}. The second is captured by
reference dependence and loss aversion \citep{KT,KR}, which help to explain
why spending reacts slowly to news and how this sluggishness interacts with
risk \citep{CS,Pagel}. This paper puts the two ideas into one model of
consumption and investment in which the reference standard is set by the
spending of the population itself.

Household data support both ideas. Spending responds to news
infrequently, responses to windfalls are either zero or substantial
\citep{FKZ}, and consumption responds more strongly to
predictable declines in income than to predictable increases \citep{Shea}.
At the same time, a household's spending rises with that of its neighbors
and coworkers \citep{KKSK,DFP}. When the standard of living is judged
relative to others, inertia has three consequences. A household that does
not revise still changes its spending when the
benchmark moves, the benchmark in turn depends on when the other households
revise, and two households with the same wealth but different inherited
standards spend and invest differently. None of them arises in a frictionless model
of relative consumption, and the second cannot be studied one household at
a time.

We model inertia with the intertemporal loss aversion (ILA) preferences of
\citet{CJK}, applied to consumption relative to a power of the benchmark,
the geometric average consumption of the reference group. Cuts in relative
felicity weigh more in utility than increases of the same size. As a
result, the household keeps its relative standard as long as the ratio of the
marginal value of wealth to marginal utility stays inside a band, and at the
ends of the band it revises by the smallest amount that brings the ratio
back. In the explicit equilibria, unlike under a fixed cost of adjustment,
anticipated spending paths have no jumps after an initial revision, and
unlike under habit formation, there are periods of inaction. The band and
the cost of consumption are given by
explicit formulas, which yield closed form equilibria when revisions are
synchronized and explicit constants in the existence proof. When cuts and
increases weigh equally, the adjustment costs reduce to a constant and a
rescaling of utility, relative consumption may change at every instant, and
the inherited standard does not affect choices (Remark~\ref{rem:frictionless}).

A model of the whole population is needed because the households themselves
produce the benchmark. With inaction, the benchmark depends on where every
household sits inside its band, so the state of the game grows with the
number of households. A mean field formulation, in which every household
takes the benchmark as given, replaces this state by a consistency condition
for the benchmark. We ask how a mean field equilibrium can be
computed when households differ in wealth, in the standards they inherit,
and in the dates at which they revise spending, and, since actual reference
groups such as neighborhoods and workplaces are small, whether its policies
remain equilibrium policies in a group of finitely many households, in which
each household moves the benchmark. Asset returns are taken as given, so the
effects we describe come from relative concerns and adjustment costs.

With finitely many types, common and private risk, and arbitrary positive
initial wealth and inherited standards, Theorem~\ref{thm:weakexist} shows
that a mean field equilibrium exists when relative concerns are weak in an
explicit sense and two further conditions keep each individual problem well
posed. Households without intertemporal loss aversion can be added as
further types (Remark~\ref{rem:frictionless}). For prices within constant
multiples of a geometric price, the logarithm of optimal relative
consumption is Lipschitz in the logarithm of the price, with a constant
equal to the reciprocal of the curvature of utility, so that for fixed
budget multipliers the benchmark map is a contraction near a geometric
benchmark. Brouwer's
theorem gives multipliers at which all budgets bind, and under a stronger
condition they are unique.
Section~\ref{sec:computing} treats the inner benchmark problem and
the outer budget equations separately. With common risk and a common
discount rate, the revision rules are computed once in index coordinates,
and wealth enters only through the budgets. The lattice has one dimension
for each type, so its size grows quickly with the number of types but not
with the number of households. In Section~\ref{sec:computeexample} three groups revise at different dates
until common shocks bring their indices together, and on all
simulated paths of two examples with types that differ in loss aversion,
the more loss averse type revises only together with the other once they
have revised jointly.

With common risk in a complete market, Theorem~\ref{thm:exact} shows that a
mean field equilibrium at an empirical type distribution is an exact Nash
equilibrium of the finite game with that composition and of all its
replicas. Because a deviating household also moves the geometric benchmark, any
relative consumption path that it can finance in the finite game is
affordable for a price taker with the same wealth, and the policy of the
price taker is optimal, so no deviation is profitable.
The argument uses only that all risk is common and traded, that the
geometric benchmark includes the household itself, that relative concerns
are below one, and that preferences depend on relative consumption
(Remark~\ref{rem:scope}), so households may differ in preferences, wealth,
and inherited standards and need not revise at the same time. When households cannot trade the private risks of others, the
budget comparison fails, and the reflected plans of the stationary
equilibrium cannot be financed from any finite wealth in a finite population
of identical households with a positive relative concern.

Populations rarely have the exact composition for which policies were
computed. When type frequencies approach a target distribution,
Theorem~\ref{thm:populationlimit} shows that the finite equilibria obtained
from Theorem~\ref{thm:weakexist} at the empirical frequencies converge along
subsequences to a mean field equilibrium. Theorem~\ref{thm:certificate}
bounds the gain from deviating when every plan is a best response to some
benchmark other than the realized one, for instance when plans computed at
a target composition are used in the actual population. For logarithmic
utility, a common discount rate, and a common relative concern below one
half, the target policies remain an approximate Nash equilibrium after the
type counts are rounded, with an error of order one over the population size
(Corollary~\ref{cor:rounded}). In an example with power utility, the bound
for keeping the target policies is below 0.004\% of wealth with ten
households (Section~\ref{sec:advice}).

Explicit equilibria check these computations and have economic content of
their own. With common risk and identical households, the unique symmetric
equilibrium for every population size is the policy of a single loss averse
household whose risk aversion is modified by the relative concern
(Theorem~\ref{thm:symmetricunique}). If risk aversion exceeds one, a
stronger relative concern lowers effective risk aversion and raises the
risky share at revision dates. When households that differ in curvature and relative
concern revise together (Theorem~\ref{thm:common}), a larger share of
households with low effective risk aversion raises the risky share at
revision dates of households with risk aversion above one and a positive
relative concern, although their own preferences are unchanged
(Proposition~\ref{prop:composition}). With purely idiosyncratic risk, the
benchmark grows at a constant rate in a stationary equilibrium
(Theorem~\ref{thm:stat}), and a household that keeps its relative
consumption fixed lets its spending grow at a fraction of that rate, so
spending growth between revisions need not reflect a decision to raise the
standard of living.
Benchmark growth lowers the effective interest rate, which in our example
turns negative within the range of estimated relative concerns, and
Lemma~\ref{lem:band} therefore covers every real interest rate.

Adjustment costs, initial wealth, and inherited standards move the timing
of revisions and the level of the benchmark but not its long run growth rate
(Section~\ref{sec:growth}). At fixed wealth, however, a policy that ignores
the inherited standard at time zero loses up to 17\% of wealth in our
numerical example, and its largest loss ranges from 9\% to 29\% as the
parameters vary. A small error about the composition of the reference group
costs less than 0.5\%, but a household that believes all its peers are like
itself loses 16\%, mostly because its revision band is misplaced
(Section~\ref{sec:advice}). Section~\ref{sec:evidence} relates these
implications to household evidence. Estimates of the weight that households
place on the consumption of their peers lie between 0.3 and 0.44, which our
existence result covers when risk aversion is at least one.

\subsection{Related Literature}

The frictionless starting point is the consumption and portfolio problem of
\citet{Merton}, solved by the martingale method of \citet{KLS} and
\citet{CoxHuang} \citep[see also][]{KS98}. Costly adjustment of durable
goods \citep{GL}, consumption commitments \citep{CS}, proportional
transaction costs \citep{DN}, and ratcheting constraints \citep{Dybvig,JKS}
all make spending persistent and risk taking dependent on the state.
Reference points formed from past consumption appear in \citet{DLPY},
\citet{LYZ}, and \citet{GHR}, and \citet{BLN} study loss aversion with an
endogenously updated reference level. In portfolio selection, \citet{HS}
show how an endogenous reference point changes the apparent degree of loss
aversion, and \citet{SL} find that investors who fail to anticipate updates
of their reference point invest in a way that is inconsistent over time.
\citet{GR} let aspirations be formed from the incomes of others and show how
they shape investment incentives and the long run income distribution. Our
individual problem builds on \citet{CJK}, and \citet{CJK25} examine the
preference representation further. Here the reference standard is relative
and determined in equilibrium, and loss aversion applies to changes in a
household's own relative standard rather than to a shortfall from others.

Habit formation explains persistent consumption and several features of
asset prices \citep{Const,CC,DK}, and \citet{HHK} and \citet{BR} study
preferences in which consumption at one date is a close substitute for
consumption at nearby dates. In these models a reference stock enters flow
utility. In our model the adjustment cost is attached to changes in relative
felicity, so a household can leave its relative consumption unchanged while
its actual spending follows the benchmark.

Equilibria of investment games under relative performance criteria are
characterized by \citet{ET} and \citet{LZ}, and \citet{BM} study fund
managers with relative wealth concerns. \citet{LS} add consumption,
\citet{DRS} consider recursive utility, \citet{FZ} develop a general mean
field portfolio game, and \citet{HZ} compare competitive and homophilous
interaction. \citet{Zhou} studies how to design relative performance
criteria in mean field contribution games, \citet{BWY} analyze a mean field
game with external habit formation, and \citet{LiZh} combine habit
formation with relative performance criteria. Consumption in these
models, when present, is adjusted continuously. With adjustment costs,
households stay inactive for long periods, their risky share departs from
its frictionless value, and the equilibrium construction must account for
the cost of variation and for the inherited standard.

The consistency principle of mean field games goes back to \citet{HMC} and
\citet{LL}. Common noise is treated by \citet{CDL} and \citet{CD18}, and
singular controls by \citet{FH}, \citet{GX}, \citet{CDF}, and \citet{DFFN}.
\citet{Lacker16,Lacker20}, \citet{Fischer}, and \citet{CF} study the
convergence of Nash equilibria of finite games to mean field equilibria. In
operations research, \citet{WBV}, \citet{AJW}, and \citet{LW} develop mean
field equilibrium as an approximation to dynamic games with many players and
study its existence, uniqueness, and comparative statics.
\citet{ABF} study irreversible investment as a singular control problem
and prove existence and uniqueness of a stationary mean field equilibrium in
an economy with two regimes.
In the relative performance games of \citet{LZ} and \citet{LS}, the
equilibria of the finite game and of the mean field game are both available
in closed form and can be compared directly. Theorem~\ref{thm:exact}
compares budgets instead, so it needs no solution of the finite game and
applies with path dependent preferences, adjustment costs, inherited
standards, and revisions at different dates, but it requires that all risk
be common and traded. \citet{AC} develop finite difference methods for mean field games. Our
lattice scheme is a Markov chain approximation in the sense of \citet{KD},
and we solve the coupled obstacle problems with Howard's algorithm, whose
convergence for obstacle problems is studied by \citet{BMZ}. Our verification and duality arguments for the individual
problem are related to \citet{BR}, \citet{RS}, and \citet{BK}.

Section~\ref{sec:model} presents the model, the equilibrium results, and
the computational procedure. Section~\ref{sec:implications} derives
economic implications, computes an asynchronous equilibrium, compares
policies at fixed wealth, and relates the results to household evidence.
Section~\ref{sec:conclusion} discusses the scope of the results and open
problems.
Appendix~\ref{app:proofs} proves Theorems~\ref{thm:common} to~\ref{thm:stat}
and~\ref{thm:exact}, Corollaries~\ref{cor:common} and~\ref{cor:stat}, and
Propositions~\ref{prop:composition} to~\ref{prop:comparativestat}, and it
describes the other appendices, which contain the remaining proofs and the
computations and collect further results.

%% file: sec2_model.tex
\section{Model and Equilibrium}\label{sec:model}

\subsection{Market, Preferences, and Equilibrium}\label{sec:setup}

Let $(\Omega,\F,\mathbb F,\mathbb P)$ be a filtered probability space
satisfying the usual conditions and carrying two independent standard
Brownian motions $B$ and $W$. The Brownian motion $B$ is common to all
households, and $W$ is specific to the representative household, whose filtration is the
augmented filtration generated by its initial type and $(B,W)$. The type is
independent of $(B,W)$. We write $\FB=(\F^B_t)_{t\ge0}$ for the augmented
filtration of $B$. When only one source of risk is present, the other
Brownian motion is omitted.

A household can invest in a bank account with interest rate $r>0$ and in
two risky technologies with volatilities $\sigma$ and $\sigma_p$, driven by
$B$ and by its own $W$. Their Sharpe ratios $\kappa_a\ge0$ and
$\kappa_p\ge0$ may depend on the type, since a household does not trade the
technologies of other types, and $\kappa^2=\kappa_a^2+\kappa_p^2>0$.
Given its type, the household trades in a complete market with state price
density
\begin{equation}\label{eq:spd}
 \xi_t=\exp\Big(-\big(r+\tfrac12\kappa^2\big)t-\kappa_aB_t-\kappa_pW_t\Big).
\end{equation}
A consumption plan $c$ can be financed from initial wealth $X_0$ if and only
if $\E\int_0^\infty\xi_tc_t\,\dd t\le X_0$ \citep{KS98}.

A benchmark is a positive $\FB$ adapted process $Y$. Given a relative concern
$\theta\in[0,1)$, relative consumption is $h_t=c_tY_t^{-\theta}$. Let
$u(h)=h^{1-\gamma}/(1-\gamma)$ for $\gamma>0$, $\gamma\ne1$, and $u=\log$
for $\gamma=1$. Following the ILA model of \citet{CJK}, the household evaluates relative consumption by
\begin{equation}\label{eq:U}
 U(h)=\E\int_0^\infty e^{-\delta t}\Big(u(h_t)\,\dd t-\alpha\,\dd v_t^{+}
       -\beta\,\dd v_t^{-}\Big),\qquad v_t=u(h_t),
\end{equation}
where $v^{+}$ and $v^{-}$ are the increasing and decreasing parts of $v$ on
$[0-,t]$ and $h_{0-}>0$ is the inherited relative standard, which we call the
relative reference. Integrals with
respect to $\dd v^{\pm}$ are taken over $[0,\infty)$ and include a jump at
time zero. The parameters satisfy $\delta>0$, $\beta\ge0$, and
$-\beta<\alpha<1/\delta$, and
\[
 L=\frac{1+\delta\beta}{1-\delta\alpha}>1
\]
is the coefficient of intertemporal loss aversion. The
admissible set $\mathcal A$ consists of the positive adapted right
continuous processes $h$ of finite variation on compacts, with the given
$h_{0-}$, such that
$\E\int_0^\infty e^{-\delta t}(|v_t|\,\dd t+\dd v_t^{+}+\dd v_t^{-})<\infty$.
Given $Y$ and $X_0$, the household solves
\begin{equation}\label{eq:primal}
 V(X_0,Y)=\sup\Big\{U(h):\ h\in\mathcal A,\
 \E\int_0^\infty \xi_tY_t^{\theta}h_t\,\dd t\le X_0\Big\}.
\end{equation}
Types and initial data are $\F_0$ measurable, and the expectations in
\eqref{eq:U} and \eqref{eq:primal} are conditional on $\F_0$, so budget
multipliers may be $\F_0$ measurable random variables.

Adjustment costs are charged on changes in the relative standard. A
household that scales its spending with $Y^\theta$ keeps $h$ constant and
pays nothing, whereas a household that holds its spending fixed while the
benchmark rises sees its relative standard fall. If the benchmark has a
Brownian component, admissibility forces actual consumption to move with
$Y^\theta$ between revisions, since $h$ must have finite variation. The
benchmark enters \eqref{eq:primal} only through the deflator $\xi Y^\theta$.

\begin{definition}\label{def:mfe}
The type $(\gamma,\delta,\alpha,\beta,\theta,\kappa_a,\kappa_p)$, the wealth
$X_0$, and the relative reference $h_{0-}$ of the representative household
are $\F_0$ measurable and independent of $(B,W)$. A \emph{mean field
equilibrium} is a pair $(Y,h^\ast)$ such that $h^\ast$ attains
\eqref{eq:primal} given $Y$ and
\begin{equation}\label{eq:consistency}
 \log Y_t=\E\big[\log c^\ast_t\,\big|\,\F^B_t\big],\qquad
 c^\ast_t=h^\ast_tY_t^{\theta},\qquad t\ge0 .
\end{equation}
\end{definition}

\begin{definition}\label{def:nash}
In the game with $N$ households the benchmark is
$\bar Y_t=\prod_{j=1}^N(c^j_t)^{1/N}$, and household $i$ evaluates
$h^i=c^i\bar Y^{-\theta_i}$ by \eqref{eq:U} with its own parameters.
A profile $(c^1,\dots,c^N)$ of plans that can be financed, with
$h^i\in\mathcal A$ for all $i$, is a \emph{Nash equilibrium} if no household
can raise its utility by switching to another plan that it can finance and
whose relative consumption lies in $\mathcal A$, while the consumption
processes of the other households are held fixed.
\end{definition}

The equilibrium notion is open loop. With common risk only, the Brownian
market is complete for every household, and a plan can be financed if and
only if it satisfies the budget inequality. We write
\begin{equation}\label{eq:ghat}
 \hg=\gamma+\theta(1-\gamma)>0
\end{equation}
for the \emph{effective risk aversion}.

\subsection{Individual Policies}\label{sec:individual}

Fix an expenditure market $(r,\kappa)$, where we now allow any $r\in\R$
because the markets that arise after absorbing the benchmark into prices can
have zero or negative interest rates. For $0<\za<\zb$ consider
\begin{gather}
 \tfrac12\kappa^2z^2H''(z)+(\delta-r)zH'(z)-\delta H(z)+1-z=0,
 \qquad \za\le z\le\zb,\label{eq:Hode}\\
 H(\za)=\alpha,\quad H'(\za)=0,\qquad H(\zb)=-\beta,\quad H'(\zb)=0,
 \label{eq:Hbc}
\end{gather}
and for $\rho>0$ define
\begin{equation}\label{eq:K}
 K(\rho)=r+\frac{\delta-r}{\rho}+\frac{(\rho-1)\kappa^2}{2\rho^2}.
\end{equation}
Let $\mathcal B$ be a standard Brownian motion, $\nu=\delta-r-\kappa^2/2$, and
$z_0\in[\za,\zb]$. The reflected log index solves
\begin{equation}\label{eq:skor}
 x_t=\log z_0+\nu t-\kappa\mathcal B_t+\la_t-\lb_t\in[\log\za,\log\zb],
 \qquad Z_t=e^{x_t},
\end{equation}
where $\la$ and $\lb$ are continuous nondecreasing processes starting at
zero that increase only when $x$ is at the lower and at the upper endpoint,
respectively. We write $\ell=\la-\lb$ for the net regulator. A low index
triggers an increase in relative consumption and a high index triggers a
decrease.

\begin{lemma}\label{lem:band}
Let $r\in\R$, $\kappa,\delta>0$, $\beta\ge0$, and $-\beta<\alpha<1/\delta$.
\begin{enumerate}[label=(\roman*),nosep,leftmargin=2em]
\item There is a unique pair $0<\za<\zb$ and a unique $C^2$ function $H$
satisfying \eqref{eq:Hode} and \eqref{eq:Hbc}. Moreover $H'<0$ on
$(\za,\zb)$ and $\za<1-\delta\alpha<1+\delta\beta<\zb$. The endpoints
depend on $(r,\kappa,\delta,\alpha,\beta)$ but not on utility curvature,
and both are strictly increasing in $r$.
\item Let $\rho>0$ and $\xi_t=\exp(-(r+\kappa^2/2)t-\kappa\mathcal B_t)$.
The cost $\E\int_0^\infty\xi_te^{\ell_t/\rho}\,\dd t$ is finite if and only
if $K(\rho)>0$. In that case there is a positive, strictly decreasing $C^2$
function $\Xf_\rho$ on $[\za,\zb]$ such that
$\E[\int_t^\infty\xi_se^{\ell_s/\rho}\dd s\,|\,\F_t]
=\xi_te^{\ell_t/\rho}\Xf_\rho(Z_t)$,
\begin{equation}\label{eq:XH}
 \Xf_\rho(z)+\rho z\Xf_\rho'(z)=-H'(z),
\end{equation}
and $e^{-d/\rho}\le K(\rho)\Xf_\rho\le e^{d/\rho}$ with $d=\log(\zb/\za)$.
\end{enumerate}
\end{lemma}

Here $\Xf_\rho$ is the ratio of wealth to consumption, and $H$ and
$\Xf_\rho$ have explicit formulas. Curvature does not enter the band in part (i), but it enters the
finiteness condition in part (ii) through $\rho$.

Suppose now that the benchmark is geometric, $\log Y_t=\log Y_0+gt+sB_t$
with constants $g$ and $s$. Then
$\xi_tY_t^\theta=Y_0^\theta\hat\xi_t$, where
$\hat\xi_t=\exp(-(\hat r+\hat\kappa^2/2)t-\hat\kappa\hat{\mathcal B}_t)$,
$\hat{\mathcal B}=(\hat\kappa_aB+\hat\kappa_pW)/\hat\kappa$, and
\begin{equation}\label{eq:hat}
 \hat\kappa_a=\kappa_a-\theta s,\qquad \hat\kappa_p=\kappa_p,\qquad
 \hat\kappa^2=\hat\kappa_a^2+\hat\kappa_p^2,\qquad
 \hat r=r-\theta(g-\kappa_as)-\tfrac12\theta^2s^2 .
\end{equation}
We call $(\hat r,\hat\kappa)$ the \emph{artificial market}, and objects
computed in it carry hats: $\hat z_\alpha$, $\hat z_\beta$, $\hat H$,
$\hat\Xf_\rho$, and $\hat K(\rho)$.

\begin{theorem}\label{thm:br}
Let the benchmark be geometric, and assume $\hat\kappa>0$ and
$\hat K(\gamma)>0$, where $\hat r$ may be any real number. Put
$\hat X_0=X_0/Y_0^\theta$. If
$\hat X_0/h_{0-}\in[\hat\Xf_\gamma(\hat z_\beta),\hat\Xf_\gamma(\hat z_\alpha)]$,
let $\hat z_0$ solve $\hat\Xf_\gamma(\hat z_0)=\hat X_0/h_{0-}$ and set
$h_0=h_{0-}$. If the ratio lies above this interval let
$\hat z_0=\hat z_\alpha$, and if it lies below let $\hat z_0=\hat z_\beta$,
with $h_0=\hat X_0/\hat\Xf_\gamma(\hat z_0)$ in both cases. Let
$(x,\la,\lb)$ solve \eqref{eq:skor} in the artificial market, driven by
$\hat{\mathcal B}$ and started at $\hat z_0$, and let $\hat Z=e^x$. Then
\[
 h^\ast_t=h_0\exp(\ell_t/\gamma)
\]
is the unique solution of \eqref{eq:primal}. Optimal wealth is
$X_t=Y_t^\theta h^\ast_t\hat\Xf_\gamma(\hat Z_t)$, and the loadings of
$\dd X_t/X_t$ on $(\dd B_t,\dd W_t)$ are
\begin{equation}\label{eq:port}
 \frac{1}{\Rr_t}(\kappa_a,\kappa_p)+\Big(1-\frac1{\Rr_t}\Big)(\theta s,0),
 \qquad \frac1{\Rr_t}=-\frac{\hat Z_t\hat\Xf_\gamma'(\hat Z_t)}
 {\hat\Xf_\gamma(\hat Z_t)} .
\end{equation}
\end{theorem}

In \eqref{eq:port} the first term is the Merton exposure scaled by the
revealed risk tolerance $1/\Rr_t$, and the second term is the exposure to
common risk that wealth needs because spending moves with the benchmark.
Applying \eqref{eq:XH} in the artificial market gives $\Rr_t=\gamma$ at the
endpoints of the band and $\Rr_t>\gamma$ inside it, where $\hat H'<0$. Since
the loadings equal $(\theta s,0)+\Rr_t^{-1}(\hat\kappa_a,\hat\kappa_p)$, they
lie closer to the benchmark loading $(\theta s,0)$ between revisions than at
revision dates. Existence and uniqueness extend to deflators between two
constant multiples of a geometric deflator with finite cost, together with
a logarithmic stability estimate on which the existence results rely.

\subsection{Common Risk: Explicit Equilibria}\label{sec:common}

In this subsection $\kappa_p=0$, so $\kappa=\kappa_a>0$ and all households
face the same state price density. Households share $(\delta,\alpha,\beta)$
and may differ in $(\gamma_i,\theta_i)$, wealth, and inherited consumption.
Let $\za$, $\zb$, $H$, and $\Xf_\rho$ be the objects of Lemma~\ref{lem:band}
for the market $(r,\kappa)$. In the game with $N$ households define
\begin{equation}\label{eq:AN}
 \vb_N=\frac{\frac1N\sum_{j}\gamma_j^{-1}}
             {1+\frac1N\sum_{j}\theta_j(\gamma_j^{-1}-1)},\qquad
 \varphi_i=\frac1{\gamma_i}+\theta_i\Big(1-\frac1{\gamma_i}\Big)\vb_N .
\end{equation}
Averaging the second identity shows that $\vb_N$ is the average of the
$\varphi_i$. In the mean field game the averages are replaced by
expectations, and we write $\vb$ and $\varphi$. Inherited absolute
consumption $c^i_{0-}>0$ induces relative references
$h^i_{0-}=c^i_{0-}\bar c_{0-}^{\,-\theta_i}$ with
$\log\bar c_{0-}=N^{-1}\sum_j\log c^j_{0-}$, or the corresponding
expectation, assumed finite, in the mean field case.

\begin{theorem}\label{thm:common}
Assume $\theta_i\vb_N\le1$ and $K(1/\varphi_i)>0$ for all $i$. Let
$z_0\in[\za,\zb]$, let $(x,\la,\lb)$ solve \eqref{eq:skor} with
$\mathcal B=B$, and suppose the initial data are synchronized in the sense
that $X^i_0=c^i_{0-}\Xf_{1/\varphi_i}(z_0)$ for all $i$. Then
\begin{equation}\label{eq:commonplan}
 c^i_t=c^i_{0-}\exp\big(\varphi_i\ell_t\big),\qquad i=1,\dots,N,
\end{equation}
is a Nash equilibrium of the game with $N$ households, and each $c^i$ is
the unique best response to the others. If the types take values in a
compact set, the same formula with $(\vb,\varphi)$ defines a mean field
equilibrium. Without synchronized initial data, the same conclusions hold
with an upward revision at time zero if $z_0=\za$, $c^i_{0-}$ is replaced in
\eqref{eq:commonplan} by $c^i_0=X^i_0/\Xf_{1/\varphi_i}(\za)$, and
$\log(c^i_0/c^i_{0-})\ge\theta_iN^{-1}\sum_j\log(c^j_0/c^j_{0-})$ for every
$i$. They also hold with a downward revision at time zero if $z_0=\zb$,
$c^i_0=X^i_0/\Xf_{1/\varphi_i}(\zb)$, and the inequality is reversed. In
the mean field case the average
is replaced by the expectation, and $\log X^i_0$ is assumed integrable.
\end{theorem}

In this equilibrium all households revise at the same dates and in the
same direction, and only the size of the revision differs across types. The
condition $\theta_i\vb_N\le1$ ensures that the relative consumption of each
type moves in the direction of $\ell$ or stays constant.

\begin{corollary}\label{cor:common}
Let the households be identical. Then $\varphi_i=1/\hg$ for every $N$ and in
the mean field game. In equilibrium every household consumes and invests as
a single ILA household with risk aversion $\hg$, discount rate
$\delta$, and adjustment costs $(\alpha,\beta)$, hence with the same $L$.
In particular:
\begin{enumerate}[label=(\roman*),nosep,leftmargin=2em]
\item consumption is revised when $\lambda e^{\delta t}\xi_tc_t^{\hg}$ reaches
$\za$ or $\zb$, where $\lambda>0$ is a constant fixed by the budget, and
these thresholds depend neither on $\theta$ nor on $\gamma$,
\item the band for the ratio of consumption to wealth is
$[1/\Xf_{\hg}(\za),1/\Xf_{\hg}(\zb)]$,
\item the share of wealth in the risky technology equals $\kappa/(\sigma\hg)$
at revision dates, which is the equilibrium share under common noise in
\citet{LZ} and \citet{DRS}, and it is strictly smaller inside the band.
\end{enumerate}
\end{corollary}

Theorem~\ref{thm:common} requires synchronized initial data or initial
revisions of all relative standards in a common direction, and for
identical households these cases cover every positive wealth and inherited
consumption. The next result adds uniqueness among symmetric equilibria.

\begin{theorem}\label{thm:symmetricunique}
Suppose all households face only common risk and have identical parameters,
wealth $X_0>0$, and inherited consumption $c_{0-}>0$, and let $K(\hg)>0$.
For every $N\ge1$ and for the mean field game there is exactly one symmetric
equilibrium, up to indistinguishability. It is the optimal policy of a
single ILA household with risk aversion $\hg$, including an
initial revision when $X_0/c_{0-}$ lies outside
$[\Xf_{\hg}(\zb),\Xf_{\hg}(\za)]$.
\end{theorem}

For a symmetric profile, utility equals $1-\theta$ times the utility of a
single household with risk aversion $\hg$, so the symmetric problem reduces
to the problem of one household. Uniqueness then follows from a deviation
constructed with H\"older's inequality, and no a priori bound on the index or
the multiplier of a competing equilibrium is needed. The same reduction
shows that the equilibrium maximizes utility over symmetric profiles that
can be financed, so within this class relative concerns do not lead to
overconsumption. We do not know whether asymmetric equilibria exist.

\subsection{Idiosyncratic Risk: Stationary Equilibrium}\label{sec:stationary}

In this subsection $\kappa_a=0$ for every type, so that $\kappa=\kappa_p>0$
and the benchmark is deterministic. Define
\begin{equation}\label{eq:gstar}
 g^\ast=\frac{\E[(r+\kappa^2/2-\delta)/\gamma]}{1+\E[\theta(1/\gamma-1)]}
\end{equation}
and, for each type, the artificial market $\hat r=r-\theta g^\ast$,
$\hat\kappa=\kappa$, with
$\hat\zeta=2(\delta-\hat r)/\kappa^2-1$ and the corresponding $\hat z_\alpha$,
$\hat z_\beta$, and $\hat\Xf_\gamma$.

\begin{theorem}\label{thm:stat}
Assume that the type takes values in a compact subset of the parameter
domain, that $\hat K(\gamma)>0$ almost surely, and that $\log h_{0-}$ is
integrable. Let
\[
 Y_0=\exp\Big(\frac{\E[\log h_{0-}]}{1-\E[\theta]}\Big)
 \qquad\text{and}\qquad X_0=Y_0^\theta h_{0-}\hat\Xf_\gamma(\hat z_0),
\]
where, conditionally on the type, $\log\hat z_0$ has the density
proportional to $e^{\hat\zeta y}$ on $[\log\hat z_\alpha,\log\hat z_\beta]$.
Then $Y_t=Y_0e^{g^\ast t}$, together with the plans of
Theorem~\ref{thm:br} for $g=g^\ast$ and $s=0$, is a mean field equilibrium.
Conversely, if $Y_t=Y_0e^{gt}$ is the benchmark of a mean field equilibrium
whose artificial markets with $\hat r=r-\theta g$ satisfy $\hat K(\gamma)>0$,
then $g=g^\ast$.
\end{theorem}

\begin{corollary}\label{cor:stat}
In the equilibrium of Theorem~\ref{thm:stat}:
\begin{enumerate}[label=(\roman*),nosep,leftmargin=2em]
\item $g^\ast$ does not depend on the distribution of $(\alpha,\beta)$, hence,
for given discount rates, not on that of $L$. For identical households
$g^\ast=(r+\kappa^2/2-\delta)/\hg$.
\item Between revisions the consumption of a household grows at rate
$\theta g^\ast$, while its relative consumption is constant.
\item Let $\hat p$ be the stationary density of $\log\hat Z$. Conditionally on
the type, $\E\la_t=\tfrac12\kappa^2\hat p(\log\hat z_\alpha)\,t$ and
$\E\lb_t=\tfrac12\kappa^2\hat p(\log\hat z_\beta)\,t$. For identical
households
$(1-\theta)g^\ast=\kappa^2\big(\hat p(\log\hat z_\alpha)
-\hat p(\log\hat z_\beta)\big)/(2\gamma)$.
\item The share of wealth in the private technology is $\kappa/(\sigma_p\Rr_t)$
with $\Rr_t$ from \eqref{eq:port}. It can depend on $\theta$ through $\hat r$
whenever $g^\ast\ne0$, although households invest in independent technologies.
\end{enumerate}
\end{corollary}

Stationarity requires the initial index to be drawn from its invariant law
given the type, so Theorem~\ref{thm:stat} covers only initial wealth drawn in
this particular way. The next subsection removes this restriction.

\subsection{Arbitrary Initial Wealth and Inherited Standards}\label{sec:exist}

Let the type $i\in\{1,\ldots,n\}$ have probability $\pi_i>0$, independently of
the Brownian motions, and put
\[
 \nu_i=\delta_i-r-\tfrac12(\kappa_{a,i}^2+\kappa_{p,i}^2),\qquad
 \Theta_1=\sum_i\pi_i\theta_i,\qquad
 \Theta_2=\sum_i\pi_i\frac{\theta_i}{\gamma_i},\qquad
 D=1-\Theta_1+\Theta_2 .
\]
Define the trend coefficients
\begin{equation}\label{eq:generalcenter}
 g=-\frac1D\sum_i\pi_i\frac{\nu_i}{\gamma_i},\qquad
 s=\frac1D\sum_i\pi_i\frac{\kappa_{a,i}}{\gamma_i},
\end{equation}
and the artificial markets
$\hat\kappa_{a,i}=\kappa_{a,i}-\theta_is$, $\hat\kappa_{p,i}=\kappa_{p,i}$,
and $\hat r_i=r-\theta_i(g-\kappa_{a,i}s)-\frac12\theta_i^2s^2$, as in
\eqref{eq:hat}. Let $[e^{a_i},e^{b_i}]$ be the band of
Lemma~\ref{lem:band} for the artificial market of type $i$ and its
$(\delta_i,\alpha_i,\beta_i)$, and put $d_i=b_i-a_i$ and
$m_i=(a_i+b_i)/2$.

\begin{theorem}\label{thm:weakexist}
Suppose every type satisfies the parameter restrictions of
Section~\ref{sec:setup}, $\hat\kappa_i>0$, $\hat K_i(\gamma_i)>0$, and
\begin{equation}\label{eq:weakcondition}
 \Theta_1+\Theta_2<1,\qquad\text{equivalently}\qquad
 \sum_i\pi_i\theta_i\big(1+\gamma_i^{-1}\big)<1 .
\end{equation}
Then for arbitrary $X_i>0$ and $h_{i,0-}>0$ there exists a mean field
equilibrium, with common and private risk, heterogeneous preferences, and
nonstationary initial data. The optimal relative consumption of each type
has a bounded index and is admissible on the infinite horizon. There are
multipliers $\lambda_i=e^{\mu_i}$ such that, with
\begin{equation}\label{eq:existconstants}
 k(\mu)=\frac1D\sum_i\pi_i\frac{m_i-\mu_i}{\gamma_i},\qquad
 R=\frac{\sum_i\pi_id_i/(2\gamma_i)}{1-\Theta_1-\Theta_2},
\end{equation}
the equilibrium benchmark satisfies
\begin{equation}\label{eq:existbound}
 \big|\log Y_t-k(\mu)-gt-sB_t\big|\le R\qquad\text{for all }t\ge0,
 \text{ almost surely.}
\end{equation}
For fixed multipliers the benchmark is unique among adapted c\`adl\`ag
benchmarks whose logarithm stays within a bounded distance of $gt+sB_t$,
and best response iteration converges to it at the geometric rate
$q=\Theta_2/(1-\Theta_1)<1$.
\end{theorem}

\begin{corollary}\label{cor:weakgrowth}
For homogeneous preference and market parameters, condition
\eqref{eq:weakcondition} reads $\theta<\gamma/(\gamma+1)$. Subject also to
$\hat K(\gamma)>0$, every finite distribution of positive initial wealth and
inherited standards admits a mean field equilibrium, with common risk,
private risk, or both. In every equilibrium constructed in
Theorem~\ref{thm:weakexist}, $\log Y_t/t\to g=(r+\kappa^2/2-\delta)/\hg$ almost
surely and $s=\kappa_a/\hg$. For logarithmic utility $\hat K(1)=\delta$, so
$\theta<1/2$ suffices.
\end{corollary}

By this estimate, a change of $\varepsilon$ in the log price moves log
relative consumption by at most $\varepsilon/\gamma$. With the
multipliers fixed, the benchmark map is therefore a contraction with factor
$q$ on the ball of radius $R$ around $k(\mu)+gt+sB_t$. Bounds on expenditure
whose constants do not depend on the multipliers confine the multipliers to a
compact set, and Brouwer's theorem gives multipliers for which every budget
binds.

\subsection{Finite Populations}\label{sec:finite}

In this subsection all risk is common, and all households trade in the
complete market with $\xi_t=\exp(-(r+\kappa^2/2)t-\kappa B_t)$, $\kappa>0$.
A finite type distribution can then be read as an actual population, with
weights equal to the observed fractions of the types. We write $U_a$ for
the utility \eqref{eq:U} with the parameters of type $a$.

\begin{theorem}\label{thm:exact}
Consider a mean field equilibrium with types $a=1,\ldots,n$, weights
$\pi_a>0$, and consumption plans $c^a=h^aY^{\theta_a}$, in which type $a$ has
initial wealth $X_a$ and relative reference $h_{a,0-}>0$. Suppose
$\pi_a=N_a/N$ for positive integers $N_a$ with $\sum_aN_a=N$. In the game
with $N$ households, give $N_a$ households the primitives of type $a$ and
the plan $c^a$. This profile is a Nash equilibrium. For every integer
$J\ge1$ the same plans form a Nash equilibrium of the replica game with
$JN_a$ households of type $a$. The benchmark and the consumption, financing
wealth, and utility of every type are the same in all these games.
\end{theorem}

The proof compares budgets. Suppose a household of type $a$ deviates to
$\tilde c$ while the others keep their plans. The benchmark becomes
$Y(\tilde c/c^a)^{1/N}$, and the relative consumption of the deviating
household becomes $h^a(\tilde c/c^a)^{1-\theta_a/N}$. A price taker facing
$Y$ can buy this relative consumption at cost
$\E\int\xi_t(c^a_t)^{\theta_a/N}\tilde c_t^{\,1-\theta_a/N}\dd t$, and
H\"older's inequality bounds this cost by $X_a$ because both $c^a$ and
$\tilde c$ cost at most $X_a$. Hence every relative consumption path that the
household can afford in the finite game is also affordable for a price taker
with the same wealth. Because the benchmark is a geometric average, the
deviation enters as a power, and the argument involves no passage to a
limit.

\begin{remark}\label{rem:scope}
The proof of Theorem~\ref{thm:exact} uses four properties of the model.
All risk is common and the market is complete for it, so the plans of one
type coincide across households and their geometric average equals the mean
field benchmark. The benchmark is the geometric average of the consumption
of all households including the one that deviates. Relative concerns lie in
$[0,1)$, which puts the exponent in H\"older's inequality in $(0,1]$.
Preferences depend on consumption only through the path of relative
consumption. The proof does not use the form of \eqref{eq:U}, the
reflected structure of the policies, or uniqueness of best responses.
Theorem~\ref{thm:exact} therefore holds when $U_a$ is any
functional of relative consumption paths, provided the admissible class is
the same in the mean field problem and in the game. An example is time
separable power utility of $cY^{-\theta}$, which gives the preferences of
\citet{Gali}. For an arithmetic average the deviation does not enter the
benchmark as a power, and the argument does not apply directly.
\end{remark}

Theorem~\ref{thm:weakexist} without private risk, applied to the empirical
type distribution, supplies equilibria that satisfy the hypotheses of
Theorem~\ref{thm:exact}, so a
Nash equilibrium exists in every population that meets the conditions of
Theorem~\ref{thm:weakexist} at its own composition.
A direct contraction argument in the finite game covers a different range
of small populations. If
households cannot trade the private risks of others, the budget comparison fails.
For identical households with $\theta>0$, the reflected plans of the
stationary equilibrium, combined with the benchmark of a finite
population, then cannot be financed from any finite wealth.

Theorem~\ref{thm:exact} keeps the composition fixed. To compare
populations whose type frequencies approach a target distribution we need
continuity of individual choices and of the budget normalization, and the
next theorem provides it without assuming uniqueness.

\begin{theorem}\label{thm:populationlimit}
Consider $n$ fixed types in the complete market with state price density
$\xi$ and arbitrary $X_i>0$ and $h_{i,0-}>0$. Let $\pi_i>0$ with
$\sum_i\pi_i=1$, and suppose the hypotheses of Theorem~\ref{thm:weakexist}
hold at $\pi$ without private risk. For positive integers $N_i^{(N)}$ with
$\sum_iN_i^{(N)}=N$ and $\pi^{(N)}_i=N^{(N)}_i/N\to\pi_i$, choose for each
large $N$ an equilibrium of Theorem~\ref{thm:weakexist} for the weights
$\pi^{(N)}$ and assign its plan for type $i$ to $N_i^{(N)}$ households.
These profiles are Nash equilibria. Every sequence of them has a
subsequence along which, for a mean field equilibrium with weights $\pi$,
every type $i$, and every $T<\infty$,
\[
 \E\!\int_0^T\!|\log Y^{(N)}_t-\log Y_t|\,\dd t\to0,\quad
 \E\!\int_0^T\!|\log c^{(N)}_{i,t}-\log c_{i,t}|\,\dd t\to0,\quad
 \E\!\int_0^\infty\!\xi_t|c^{(N)}_{i,t}-c_{i,t}|\,\dd t\to0 .
\]
Utilities converge along the same subsequence, and the financing wealth
processes converge uniformly in probability on compact time intervals. If
the limiting equilibrium is unique within the class of
Theorem~\ref{thm:weakexist}, the whole sequence converges.
\end{theorem}

Theorem~\ref{thm:populationlimit} recomputes the policies at each
empirical composition. With logarithmic utility one can instead keep the
policies computed for the target composition.

\begin{corollary}\label{cor:rounded}
In the complete market with state price density $\xi$, suppose all types
have logarithmic utility, discount rate $\delta$, and relative concern
$\theta\in[0,1/2)$, while adjustment costs, wealth, and relative references
may differ. Fix a mean field equilibrium of Theorem~\ref{thm:weakexist} with
weights $\pi_a>0$, and give its plan $c^a$, unchanged, to each of $N_a$
households, where $\sum_aN_a=N$ and $d_N=\sum_a|N_a/N-\pi_a|\le1$. The
profile can be financed and is an $\varepsilon_N$ Nash equilibrium with
$\varepsilon_N\le Cd_N$ for a constant $C$ independent of $N$. Its benchmark
satisfies $\sup_{t\ge0}|\log(Y^N_t/Y_t)|\le M_\Delta d_N$ almost surely for a
deterministic $M_\Delta$, and the utility of each type differs from its mean
field value by at most $Cd_N$. If $|N_a/N-\pi_a|\le1/N$ for every $a$, all
these errors are of order $1/N$.
\end{corollary}

The proof gives an explicit value of $\varepsilon_N$. With logarithmic utility and common
$\delta$ and $\theta$, the log relative consumptions of different types stay
within a bounded distance of each other, so a rounding error moves the
benchmark in proportion to it, uniformly in time.

In practice policies are computed for a benchmark that is not exactly the
one the population generates, for instance at a target composition. The next
result bounds what a household can gain by deviating from such a profile. It
requires that each plan be an exact best response to some benchmark $Y^0$,
uses only quantities computable from the plans, $Y^0$, and the multipliers,
and does not cover the discretization error of a numerical scheme. Let
$G_\gamma(M)=M+(\gamma-1)^{-1}-\gamma(\gamma-1)^{-1}M^{(\gamma-1)/\gamma}$
for $\gamma\ne1$ and $G_1(M)=M-1-\log M$. Both are nonnegative and vanish
together with their first derivative at $M=1$.

\begin{theorem}\label{thm:certificate}
In the complete market of this subsection, let $N_a$ households of type $a$
use a plan $c^a$ with $w_a=\E\int_0^\infty\xi_tc^a_t\dd t\le X_a$, put
$Y^N=\prod_a(c^a)^{N_a/N}$, and suppose that $h^{a,N}=c^a(Y^N)^{-\theta_a}$
is admissible with reference $h_{a,0-}$. Let $Y^0$ be a positive adapted
process such that for every type the plan $h^{a,0}=c^a(Y^0)^{-\theta_a}$,
with reference $h_{a,0-}$, is admissible, has a bounded index for the price
$\lambda_ae^{\delta_at}\xi_t(Y^0_t)^{\theta_a}$ with a constant
$\lambda_a>0$, and satisfies the optimality conditions of
Appendix~\ref{app:proofs} for this price. Put $\Delta=\log(Y^N/Y^0)$ and
suppose $\E\int_0^\infty\xi_tc^a_te^{\theta_a\Delta_t}\dd t<\infty$. If the
other households keep their plans, no household of type $a$ can raise its
utility by more than
\[
 \varepsilon_a=\lambda_a\E\!\int_0^\infty\!\xi_tc^a_t\,
 G_{\gamma_a}\big(e^{\theta_a\Delta_t}\big)\dd t
 +U_a(h^{a,0})-U_a(h^{a,N})
 +\lambda_a\Big(X_a-\E\!\int_0^\infty\!\xi_tc^a_te^{\theta_a\Delta_t}\dd t\Big)
\]
by switching to any plan that it can finance and whose relative consumption
is admissible.
\end{theorem}

We call $\varepsilon_a$ the certificate of type $a$.
The three terms measure the error in the price of relative consumption,
the utility lost because the household's own relative consumption moves with
the difference between the benchmark used to compute its plan and the
realized one, and the part of the budget that the plan leaves unused or overdraws at
realized prices. The first term is of second order in $\Delta$. If the plans
form a mean field equilibrium with weights $N_a/N$, bounded indices, and
binding budgets, as those of Theorem~\ref{thm:weakexist} do, one may take
$Y^0=Y^N$, and then $\varepsilon_a=0$, which recovers
Theorem~\ref{thm:exact} for these equilibria. Section~\ref{sec:advice}
evaluates the three terms for plans computed at a target composition.

\begin{remark}\label{rem:frictionless}
Types without intertemporal loss aversion, with $L_i=1$, that is
$\beta_i=-\alpha_i\ge0$, can be included. Their utility is
$(1-\delta_i\alpha_i)\E\int_0^\infty e^{-\delta_it}u(h_t)\dd t$ plus a
constant, as in \citet{Gali}. If these households may choose any positive
plan with finite discounted felicity, in the mean field problem and in the
game, their optimal relative consumption has a Brownian component and the
constant index $1-\delta_i\alpha_i$. Their band degenerates to
$a_i=b_i=\log(1-\delta_i\alpha_i)$, and Theorem~\ref{thm:weakexist} holds under
its other hypotheses, including $\hat K_i(\gamma_i)>0$, with $d_i=0$.
Theorems~\ref{thm:exact} and~\ref{thm:populationlimit} also hold, and
Theorem~\ref{thm:certificate} holds with a constant index in place of its
optimality conditions. Since $g$ in
\eqref{eq:generalcenter} does not involve $\alpha_i$ or $\beta_i$, long run
benchmark growth is the same as if these types had some $L_i>1$ with their
other parameters unchanged. With $\theta>0$ and both kinds of types present,
Corollary~\ref{cor:rounded} is not covered, because a rounding that changes
the total weight of the types with $L_i=1$ gives $Y^N/Y$ a Brownian
component, and the relative consumption of the other types then has
infinite variation.
\end{remark}

\subsection{Computing Equilibria}\label{sec:computing}

Theorems~\ref{thm:weakexist} to~\ref{thm:certificate} suggest a procedure
for computing equilibrium policies and using them in a finite population.
For
fixed multipliers, Theorem~\ref{thm:weakexist} identifies the benchmark as
the unique fixed point of a contraction with factor
$q=\Theta_2/(1-\Theta_1)$ in a ball of explicit radius $R$. The iteration of
the benchmark map therefore converges geometrically, and any computed
benchmark can be checked against \eqref{eq:existbound}. The multipliers
must then make every budget bind. Brouwer's theorem shows that such
multipliers exist but gives no convergent algorithm, and a contraction for
them is available only under a stronger condition, so we solve the outer budget equations with a root finder and report its
residuals.
Theorems~\ref{thm:exact} and~\ref{thm:certificate} concern exact plans. A
numerical solution approximates them, and its discretization error has to
be assessed separately (Section~\ref{sec:scope}).

With common risk and a common discount rate, the computation can be
organized in index coordinates. Write $x^a=\log z^a$ for the log index of
type $a$. Between revisions every $x^a$ moves by $\nu\,\dd t-\kappa\,\dd B_t$,
so the vector of indices moves along the diagonal. A revision that raises
$\log h^b$ by one unit moves $x^a$ by
$\gamma_a\ind_{a=b}+\theta_a\pi_b/(1-\Theta_1)$, a direction fixed by the
primitives. By the optimality conditions of Appendix~\ref{app:proofs}, the
plans are best responses to
the benchmark they generate when the adjoint of each type, a conditional
expectation of a discounted function of its own future index, lies between
two obstacles and reaches them when the type revises. If the revision rules
are functions of the current vector of indices and generate a well defined
reflected process, this vector is a Markov process, the adjoints are
functions of it, and the conditions can be checked state by state. Wealth and the multipliers enter only through the initial
point. We compute these rules once on a lattice, where the adjoint
of each type solves a double obstacle problem coupled to the revisions of
the other types, and we solve the coupled problems with Howard's
algorithm. The lattice has one dimension for each type, so its size
grows quickly with the number of types but not with the number of
households of each type. For given rules, the budget of each
type is a linear functional of its consumption, and the outer step solves
$n$ equations in the $n$ multipliers. Section~\ref{sec:computeexample}
applies the procedure to three groups whose revisions are not synchronized.

%% file: sec3_implications.tex
\FloatBarrier
\section{Economic Implications}\label{sec:implications}

The curves in Sections~\ref{sec:spending} to~\ref{sec:loss} are computed
from the formulas for the band and the cost function behind
Lemma~\ref{lem:band}, which remain valid at zero and negative effective
rates. The two simulated paths in these subsections
use projected Euler reflection with time step 0.002 years \citep{Lep} and
therefore contain discretization error.

\subsection{Actual Spending and the Relative Standard}\label{sec:spending}

In the stationary equilibrium with idiosyncratic risk the benchmark grows at
rate $g^\ast$, so a household that does not revise its relative consumption
still raises its spending at rate $\theta g^\ast$. An observer who sees
spending rise smoothly may thus be looking at a household that is only
maintaining its standard of living. Figure~\ref{fig:relative} plots actual
and relative consumption along one simulated path. The two series drift
apart during inaction, and relative consumption changes only while the index
is at an endpoint of its band.

\begin{figure}[!t]\centering
\includegraphics[width=\textwidth]{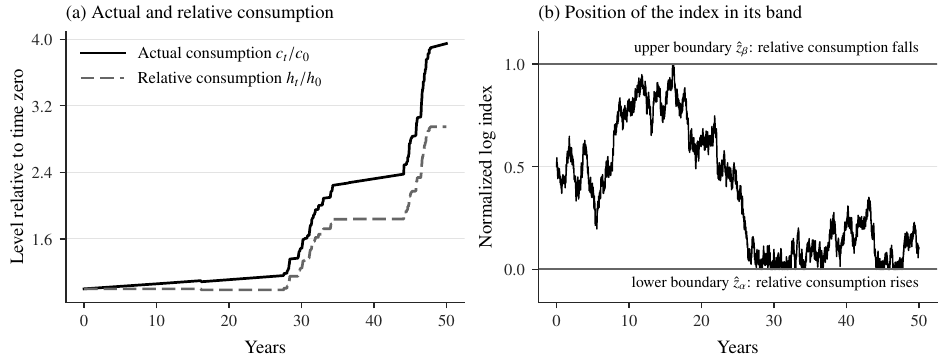}
\caption{Actual Spending and the Relative Standard.
\emph{Notes.} Stationary equilibrium with idiosyncratic risk,
$\gamma=3.5$, $\delta=0.015$, $r=0.0086$, $\kappa=0.3462$, $\theta=0.3$,
$\alpha=0$, and $L=1.5$, so that $g^\ast=1.95\%$. The household starts at the
midpoint of the band of its log index with compatible wealth, and the
population follows the invariant law of Theorem~\ref{thm:stat}. The path
uses projected Euler reflection with step 0.002 years.}
\label{fig:relative}
\end{figure}

\subsection{Relative Concerns and Risk Taking under Common Shocks}\label{sec:commonimpl}

With identical households and common risk, one reflected index triggers all
revisions, and the relative concern changes only the exponent $1/\hg$ that
maps the cumulative regulator into consumption. For $\gamma>1$ a larger
$\theta$ lowers $\hg$, so a given regulator increment produces a larger
change in consumption. For $\gamma<1$ the effect goes the other way, and
with logarithmic utility it disappears. The index thresholds are determined
by adjustment costs, discounting, and investment opportunities and are
therefore independent of $\theta$, whereas the band for the ratio of
consumption to wealth varies with $\theta$ unless $\gamma=1$.

Figure~\ref{fig:common} compares three separate economies of identical
households. Along the common Brownian path in panel (a), a stronger
relative concern produces larger revisions. In panel (b), the risky share
equals $\kappa/(\sigma\hg)$ at the endpoints of the band and is lower inside
it. On the range of ratios of consumption to wealth covered by all three
curves, the curve with the strongest relative concern lies highest, and each
curve still dips well below its endpoint value inside its band.

\begin{figure}[!t]\centering
\includegraphics[width=\textwidth]{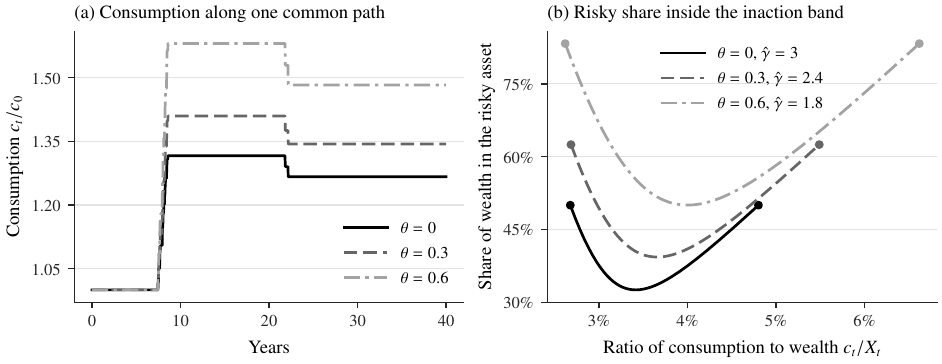}
\caption{Relative Concerns under Common Risk.
\emph{Notes.} Identical households with $r=0.02$, $\delta=0.05$,
$\kappa=0.30$, $\sigma=0.20$, $\gamma=3$, $\theta\in\{0,0.3,0.6\}$,
$\alpha=0$, and $L=2$. Panel (a) uses one Brownian path and the same initial index for the three economies,
with initial wealth matched to that index, and normalizes consumption by its
initial level. Panel (b) is computed from the cost function, and dots mark
the endpoints of the band.}
\label{fig:common}
\end{figure}

\subsection{Population Size and Population Composition}\label{sec:composition}

By Theorem~\ref{thm:exact}, replicating every type in the same proportion
leaves the equilibrium plans in place, and the benchmark and the consumption,
wealth, and utility of each type stay the same, even when preferences and
wealth differ and households revise at different dates. Changes in
composition move the equilibrium, as the next result shows.

\begin{proposition}\label{prop:composition}
Consider the synchronized equilibria of Theorem~\ref{thm:common} with two
types, common $(r,\kappa,\sigma,\delta,\alpha,\beta)$, and population share
$w\in(0,1)$ of type 1. Let $\hg_i=\gamma_i+\theta_i(1-\gamma_i)$,
\[
 \vb(w)=\frac{w/\gamma_1+(1-w)/\gamma_2}{w\hg_1/\gamma_1+(1-w)\hg_2/\gamma_2},
 \qquad
 \varphi_i(w)=\frac1{\gamma_i}+\theta_i\Big(1-\frac1{\gamma_i}\Big)\vb(w),
 \qquad Q=r+\frac{\kappa^2}{2}-\delta .
\]
On any interval of shares on which $\theta_i\vb(w)\le1$ and
$K(1/\varphi_i(w))>0$ for both types, $\log Y_t/t\to Q\vb(w)$ almost surely,
and the risky share of type $i$ at either endpoint of the band is
$\kappa\varphi_i(w)/\sigma$. Moreover
\[
 \vb'(w)=\frac{\hg_2-\hg_1}{\gamma_1\gamma_2\,\big(w\hg_1/\gamma_1+(1-w)\hg_2/\gamma_2\big)^2},
 \qquad
 \frac{\dd}{\dd w}\frac{\kappa\varphi_i(w)}{\sigma}
 =\frac{\kappa}{\sigma}\,\theta_i\Big(1-\frac1{\gamma_i}\Big)\vb'(w).
\]
Hence, if $\hg_1<\hg_2$, a larger share of type 1 raises the asymptotic
growth rate of the benchmark when $Q>0$, and it raises the risky share at
revision dates of every type with $\gamma_i>1$ and $\theta_i>0$, although the
preferences and opportunities of that type are unchanged.
\end{proposition}

As the share of households with low effective risk aversion grows, the
benchmark reacts more strongly to common shocks, and a household with
$\gamma_i>1$ and $\theta_i>0$ responds by holding more of the risky asset at
revision dates. Figure~\ref{fig:composition} shows both effects for two
types. Proposition~\ref{prop:composition} concerns the risky share at the
endpoints of the band and compares the compatible initial allocations of
Theorem~\ref{thm:common}, so no welfare conclusion follows.

\begin{figure}[!t]\centering
\includegraphics[width=\textwidth]{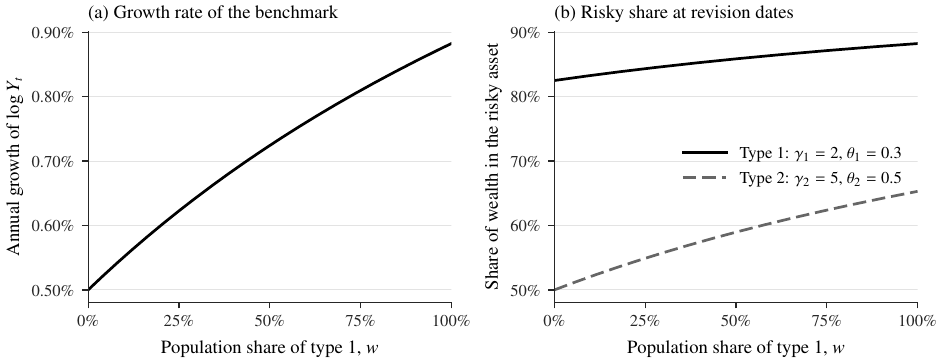}
\caption{Population Composition and Investment Spillovers.
\emph{Notes.} Two synchronized types with $(\gamma_1,\theta_1)=(2,0.3)$ and
$(\gamma_2,\theta_2)=(5,0.5)$, and common $r=0.02$, $\delta=0.05$,
$\kappa=0.30$, $\sigma=0.20$, $\alpha=0$, and $L=2$. Initial wealth is
compatible with a common initial index for each composition. Both panels are
computed from exact formulas, and replicating a fixed composition changes
neither curve.}
\label{fig:composition}
\end{figure}

\subsection{Inherited Standards and Initial Risk Taking}\label{sec:inherited}

In the symmetric equilibrium of Theorem~\ref{thm:symmetricunique}, let
$b_-=1/\Xf_{\hg}(\za)$ and $b_+=1/\Xf_{\hg}(\zb)$, so $b_-<b_+$. For initial
wealth $X_0$ and inherited consumption $c_{0-}$, Theorem~\ref{thm:symmetricunique}
and the initial rule of Theorem~\ref{thm:br} for a single household with risk
aversion $\hg$ imply
\[
 \frac{c_0}{X_0}=\min\Big\{b_+,\max\Big\{b_-,\frac{c_{0-}}{X_0}\Big\}\Big\}.
\]
Inside this interval the household keeps its inherited spending. An
inherited standard below the interval leads to an immediate increase, and one
above it to an immediate cut. After either revision the risky share equals
$\kappa/(\sigma\hg)$, whereas an interior initial ratio leads to a strictly
lower share. Figure~\ref{fig:inherited} fixes initial wealth and varies the
inherited standard across economies with identical households. Households
with the same wealth choose different spending and portfolios, and risk
taking is lower for intermediate inherited standards than at either
extreme.

\begin{figure}[!t]\centering
\includegraphics[width=\textwidth]{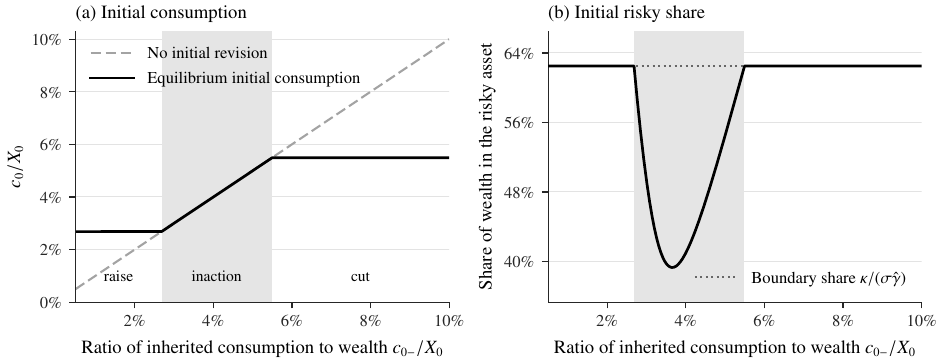}
\caption{Inherited Standards at Fixed Initial Wealth.
\emph{Notes.} Identical households with common risk, $r=0.02$,
$\delta=0.05$, $\kappa=0.30$, $\sigma=0.20$, $\gamma=3$, $\theta=0.3$,
$\alpha=0$, and $L=2$, so that $\hg=2.4$. Initial wealth is normalized to one.
The shaded region requires no initial revision, and outside it the initial
revision of Theorem~\ref{thm:symmetricunique} applies. Both panels are
computed from exact formulas.}
\label{fig:inherited}
\end{figure}

\subsection{Benchmark Growth and Current Spending}\label{sec:growth}

The following aggregation identity splits the log benchmark into a
deterministic trend, a common shock, and a term that depends on where
households are inside their bands.

\begin{proposition}\label{prop:agg}
Let $(Y,h^\ast)$ be a mean field equilibrium, let $\lambda>0$ be $\F_0$
measurable, and define
$z^\ast_t=\lambda e^{\delta t}\xi_tY_t^\theta(h^\ast_t)^\gamma$. Assume that
the type takes values in a compact subset of the parameter domain and that
the expectations below are finite. Then
\begin{equation}\label{eq:agg}
 \Big(1+\E\Big[\theta\Big(\frac1\gamma-1\Big)\Big]\Big)\log Y_t
 =-\E\Big[\frac{\log\lambda}{\gamma}\Big]
 +\E\Big[\frac{r+\frac12\kappa^2-\delta}{\gamma}\Big]t
 +\E\Big[\frac{\kappa_a}{\gamma}\Big]B_t
 +\E\Big[\frac{\log z^\ast_t}{\gamma}\,\Big|\,\F^B_t\Big].
\end{equation}
If the last term is $o(t)$ almost surely, in particular if $z^\ast$ takes
values in a fixed compact subset of $(0,\infty)$, then $\log Y_t/t\to g^\ast$
almost surely, where $g^\ast$ is given by \eqref{eq:gstar} with
$\kappa^2=\kappa_a^2+\kappa_p^2$.
\end{proposition}

Given the distribution of risk exposures, discount rates, curvature, and
relative concerns, the adjustment costs have no effect on $g^\ast$. They
affect spending and welfare through the bounded index term in
\eqref{eq:agg}, the wealth required to finance a given standard, and the
path of consumption.

The initial cross section does not affect long run growth either. In every
equilibrium of Theorem~\ref{thm:weakexist} the benchmark stays within
distance $R$ of $k(\mu)+gt+sB_t$ by \eqref{eq:existbound}, and $g$ in
\eqref{eq:generalcenter} depends only on preferences, the market, and the
type weights, so $\log Y_t/t\to g$ almost surely. Initial wealth and
inherited standards move the timing of revisions and, through the
multipliers, the level of the benchmark, whereas the composition of
preferences moves its growth rate, as \eqref{eq:generalcenter} and
Proposition~\ref{prop:composition} show. In the model of \citet{GR}, where
aspirations formed from the incomes of others drive investment, initial
histories matter for the long run distribution. Here asset returns are
exogenous, and the benchmark grows at a rate fixed by preferences and the
market.

\begin{proposition}\label{prop:comparativestat}
For identical households in the stationary equilibrium of
Theorem~\ref{thm:stat}, let $Q=r+\kappa^2/2-\delta$. Then
\[
 \frac{\partial g^\ast}{\partial\theta}=\frac{(\gamma-1)Q}{\hg^{\,2}},\qquad
 \frac{\partial\hat r}{\partial\theta}=-\frac{\gamma Q}{\hg^{\,2}} .
\]
Throughout the range of $\theta$ on which $\hat K(\gamma)>0$, both index
thresholds strictly decrease in $\theta$ if $Q>0$, strictly increase if
$Q<0$, and are constant if $Q=0$.
\end{proposition}

Suppose $\gamma>1$ and $Q>0$. A stronger desire to keep pace then raises
stationary benchmark growth and lowers the effective rate
$\hat r=r-\theta g^\ast$, because one unit of relative consumption
requires expenditure that grows faster. With $\gamma<1$ the sign of
$\partial g^\ast/\partial\theta$ reverses, so relative concerns need not raise
growth. In Figure~\ref{fig:stationary},
faster benchmark growth comes with a lower band for the ratio of consumption
to wealth. Proposition~\ref{prop:comparativestat} gives the direction in
which the index thresholds move in general, but the downward shift of this
band is only a numerical finding for these parameters. Panel (c) shows the
expected regulator increments divided by $\gamma$, that is, the upward and
downward contributions to $\log h$ per year. By
Corollary~\ref{cor:stat}(iii), their difference plus spending growth during
inaction equals benchmark growth.

\begin{figure}[!t]\centering
\includegraphics[width=\textwidth]{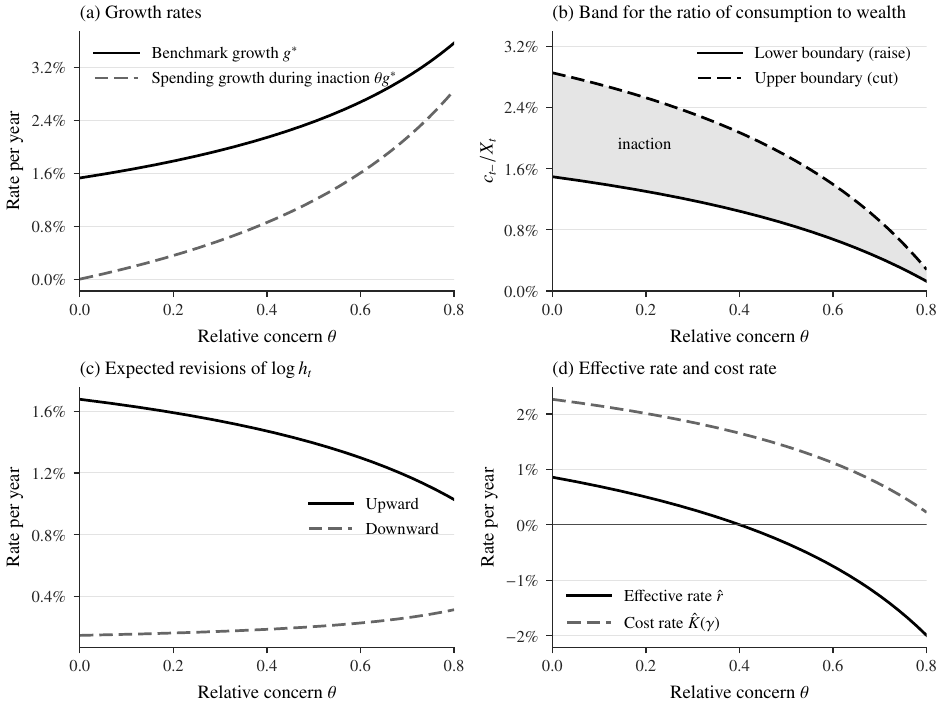}
\caption{Growth, Spending, and Revisions under Idiosyncratic Risk.
\emph{Notes.} Identical households with $\gamma=3.5$, $\delta=0.015$,
$r=0.0086$, $\kappa=0.3462$, $\alpha=0$, and $L=1.5$. Each value of $\theta$
corresponds to the stationary equilibrium of Theorem~\ref{thm:stat} with
compatible initial wealth. All curves are computed from exact formulas. In
panel (d), $\hat r$ crosses zero at $\theta\approx0.401$, and $\hat K(\gamma)$
remains positive on the plotted range.}
\label{fig:stationary}
\end{figure}

A negative effective rate causes no difficulty for the band, which exists
for every real rate, but a sufficiently low rate makes the regulated
consumption stream impossible to finance. The stream has finite cost only if
$\hat K(\gamma)>0$, and in this example the condition fails at
$\theta\approx0.836$. As this limit is approached, the wealth needed per
unit of consumption grows like the reciprocal of $\hat K(\gamma)$.

\subsection{Loss Aversion and Risk Taking inside the Band}\label{sec:loss}

With the relative concern fixed, a higher $L$ widens the band for
the ratio of consumption to wealth and lowers the risky share inside it, so
a low risky share inside the inaction band is consistent with moderate
curvature once the cost of changing consumption is taken into account.

\subsection{Asynchronous Revisions with Three Groups}\label{sec:computeexample}

The closed forms of Theorems~\ref{thm:common} and~\ref{thm:symmetricunique}
require that all households revise at the same dates. We now compute an
equilibrium in which they do not. Three equally large groups have the
preferences of Figure~\ref{fig:common}, with $\gamma=3$ and $\theta=0.3$,
in the market of that figure. Their inherited consumption is $0.90$, $1.00$,
and $1.04$, and their wealth is $38$, $31$, and $27$ times their inherited
consumption. The weak interaction condition holds with
$\Theta_1+\Theta_2=0.4$, the contraction factor of
Theorem~\ref{thm:weakexist} is $q=1/7$, and the radius in
\eqref{eq:existbound} is $R=0.59$. We use the procedure of
Section~\ref{sec:computing}, and Table~\ref{tab:async} reports its accuracy.

The first group inherits a standard that is low relative to its wealth and
raises its relative consumption at time zero, while the other two groups
keep their inherited standards. It raises its spending until its log index
reaches $-0.874$, the end of its inaction region given the positions of the
other two groups, whereas synchronized households raise their spending at
$\log\za=-0.910$. The band of Theorem~\ref{thm:common} therefore does not
describe the asynchronous phase. Because the three log indices differ, a
common shock brings one group to a revision boundary while the others
remain inside their inaction regions. Figure~\ref{fig:async} follows one
path. The first group raises its spending repeatedly during the first three
years, the second group starts to raise its spending when its index reaches
that of the first, and the third group joins after about seven years. From
then on the three groups revise together, as in Theorem~\ref{thm:common},
and they hold the same risky share although their wealth differs. Before
that date the groups differ in their risky shares, and a household that
does not revise still changes its spending because the benchmark moves when
the other groups revise. In $20{,}000$ simulated paths the three groups
synchronize on every path, within ten years with probability $0.23$ and on
average after $23$ years. In the exact equilibrium, groups with common
preferences synchronize in finite time almost surely. The
benchmark stays within the bound \eqref{eq:existbound} along every path.

By Theorem~\ref{thm:exact}, the plans of the exact equilibrium that this
computation approximates form a Nash equilibrium of the game with $J$
households in each group for every $J\ge1$. The computed plans approximate
these plans with the accuracy reported in Table~\ref{tab:async}, and
Section~\ref{sec:scope} discusses what this leaves open.
We compute two further equilibria on the full lattice, with two types that differ in loss aversion and wealth, and with
two types that differ in curvature, relative concern, loss aversion, and
inherited standards. Figure~\ref{fig:twotypes} follows one path of the
first example, in which both types have $\gamma=3$ and $\theta=0.3$, type 1
has $L=1.5$, and type 2 has $L=2.5$, so that the band of type 1 lies
strictly inside the band of type 2. Type 1 raises its standard repeatedly
during the first seven years, and the spending of type 2 rises with the
benchmark although type 2 does not revise. At about $7.5$ years the index
of type 2 reaches its lower boundary while type 1 is revising, and both
raise their standards together. Afterwards type 1 cuts its standard three
times, and the benchmark and the spending of type 2 fall with it. On all
$20{,}000$ simulated paths of both examples, once a joint revision has
occurred the more loss averse type revises only at dates at which the
other type revises as well, whereas about a third of the revisions of the
less loss averse type are not joint. This is a numerical observation, with
the following explanation. The two log indices receive
the same common increments, and in the computed equilibria a revision of
the type with the narrower band alone moves the index of the other type
away from the boundary it is approaching, so the less loss averse type
reaches a boundary first.
\begin{figure}[!t]\centering
\includegraphics[width=\textwidth]{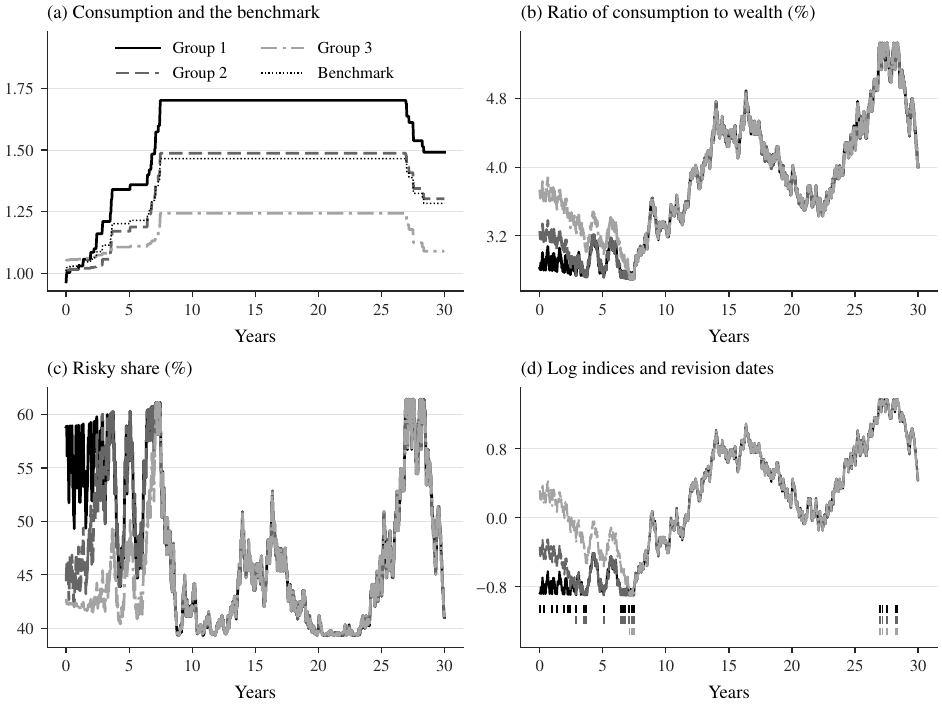}
\caption{Asynchronous Revisions with Three Groups.
\emph{Notes.} Common risk with $r=0.02$, $\delta=0.05$, $\kappa=0.30$,
$\sigma=0.20$, $\gamma=3$, $\theta=0.3$, $\alpha=0$, and $L=2$. Three equally
large groups with inherited consumption $0.90$, $1.00$, and $1.04$ and wealth
equal to $38$, $31$, and $27$ times inherited consumption (groups 1, 2, and 3
in black, dark gray, and light gray). One simulated path of the lattice
scheme with time step $1/64$. Panel (a) uses units in which inherited
consumption is $0.90$, $1.00$, and $1.04$, and ticks at the bottom of panel
(d) mark the revision dates of groups 1, 2, and 3, from top to bottom.}
\label{fig:async}
\end{figure}

\iftablesatend\else
\input{table1}
\fi

\begin{figure}[!t]\centering
\includegraphics[width=\textwidth]{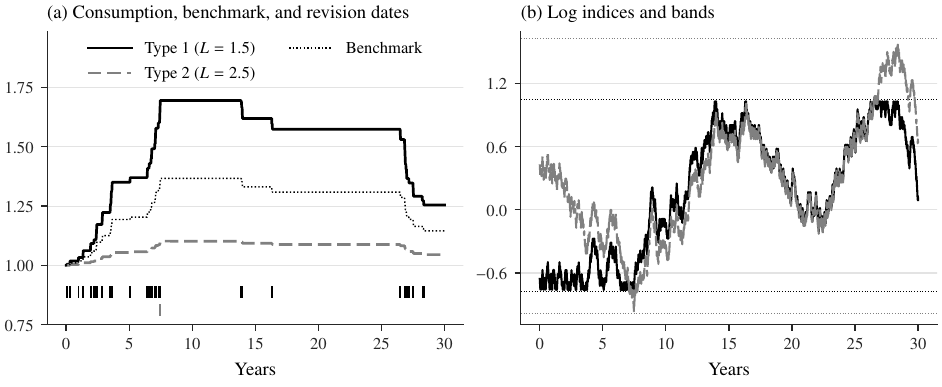}
\caption{Two Types with Different Loss Aversion.
\emph{Notes.} Common risk with $r=0.02$, $\delta=0.05$, $\kappa=0.30$,
$\sigma=0.20$, $\gamma=3$, $\theta=0.3$, and $\alpha=0$. Type 1 has $L=1.5$ and
wealth $32$, and type 2 has $L=2.5$ and wealth $27$, both with inherited
consumption $1$. One simulated path of the lattice scheme with time step
$1/64$. Ticks at the bottom of panel (a) mark the revision dates of types 1
and 2, from top to bottom, and the dotted lines in panel (b) mark the ends
of the band of each type.}
\label{fig:twotypes}
\end{figure}

\subsection{Value of Policies at Fixed Wealth}\label{sec:advice}

Figure~\ref{fig:composition} compares economies whose initial wealth is
compatible with each composition, so it does not show what a given
household gains from a better policy. To measure this gain, we fix wealth
and inherited consumption and
compare policies for the same household. Every
comparison policy is financed by the household's wealth and has relative
consumption of finite variation, as admissibility requires. The
frictionless policy of \citet{Merton} does not have this property, because
its consumption is a diffusion, and comparing with it would make the
adjustment cost infinite. Losses are reported as the additional wealth, in
percent of initial wealth, that a household following the alternative
policy would need to obtain the utility of the optimal policy, with its
inherited standard and the benchmark held fixed. Utilities are
solutions of linear differential equations, which are available in closed
form for the reflected index and are computed by finite differences for the
misjudged composition. The parameters are those of the earlier figures and
are not estimated, so the losses show how large the effects can be in the
model and are not measurements.

\emph{Inherited standards.} Consider a household in the stationary
equilibrium of Figure~\ref{fig:stationary} with $\theta=0.3$. It faces
the deterministic benchmark and therefore solves the problem of
Theorem~\ref{thm:br}. We compare the optimal policy with two policies that
ignore the inherited standard at time zero and follow the optimal rule
afterwards. The first starts at the level that would be optimal if the
initial standard could be chosen freely, which with $\alpha=0$ is the lower
end of the band for the ratio of consumption to wealth. The second starts at the ratio
$\hat K(\gamma)$ of the frictionless problem. Figure~\ref{fig:advice}(a)
shows that the first policy loses nothing when the household inherits a
standard below the band, because the optimal policy then also raises
spending at once to the lower end. When the inherited standard lies inside
or above the band, the first policy forces a costly cut, and the loss
reaches $17\%$ of wealth. For other values of $\theta$, $L$, and $\gamma$,
the largest loss ranges from $9\%$ to $29\%$. The second policy loses up to $4.4\%$, and it
loses nothing only when the inherited ratio happens to equal
$\hat K(\gamma)$.

\emph{Misjudged composition.} Let the population consist of two types with
$\gamma=2$ and $\gamma=5$ and common relative concern $0.3$, in equal
shares, with the other parameters of Figure~\ref{fig:stationary}. A
household of the second type that believes the share of the first type to
be $w'$ uses the band and the portfolio rule that are optimal for the
benchmark growth implied by $w'$. Its spending still moves with the
realized benchmark between revisions. Figure~\ref{fig:advice}(b) shows the
loss. Misjudging the share by $0.1$ costs $0.2\%$ to $0.3\%$ of wealth, and
less than $0.5\%$ for the other parameter values. Believing that the reference
group consists only of households like itself costs $16\%$, and the cost of
this error rises steeply with $\theta$ and with the curvature of the household. Most of the loss comes from
the misplaced band. With the
believed band and the correct portfolio rule the loss in the last case is
$12\%$, and with the correct band and the believed portfolio rule it is
$1\%$.

\emph{Recomputing for the actual composition.} Finally consider the two
types of Figure~\ref{fig:composition} with a target share of $1/3$ for the
first type, wealth $40$, and inherited consumption $0.80$ and $0.85$. Both
types then raise their spending at time zero, and the equilibria at the
target and at every composition in Table~\ref{tab:advice} are those of
Theorem~\ref{thm:common} with $z_0=\za$. These plans are known in closed
form, so the certificate below involves no discretization error. In a
population of $N$ households
the closest integer counts give the actual composition, and
Table~\ref{tab:advice} compares two options. The first is to keep the
plans computed at the target. The certificate of
Theorem~\ref{thm:certificate}, with $Y^0$ the target benchmark, then bounds
the gain of any household from any deviation. The relative references in
the game depend on the actual composition, but the target plans still
revise upward at time zero from these references, so they satisfy the
optimality conditions that the theorem requires. The second is to switch
all households to the equilibrium recomputed at the actual composition, and
the table reports the resulting change in utility. The certificate is below
$0.004\%$ of wealth for ten households and decreases with the square of
the composition error, whereas recomputing changes utilities by amounts of
first order. The second order behavior arises because the second and third
terms of Theorem~\ref{thm:certificate} cancel to first order, which need
not occur in general, and recomputing mainly shifts utility between the
types.

\begin{figure}[!t]\centering
\includegraphics[width=\textwidth]{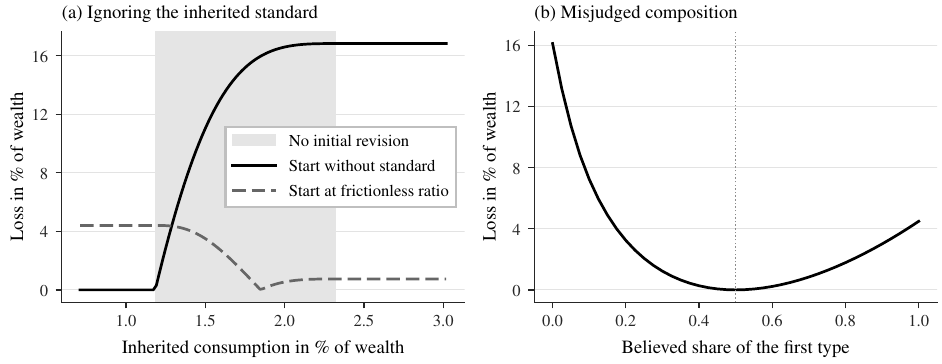}
\caption{Value of Policies at Fixed Wealth.
\emph{Notes.} Panel (a): stationary equilibrium of
Figure~\ref{fig:stationary} with $\theta=0.3$. The first policy starts at
the lower end of the band for the ratio of consumption to wealth, and the
second at the frictionless ratio $\hat K(\gamma)$. Both follow the optimal
rule afterwards. The shaded region requires no initial revision. Panel
(b): two types with $\gamma=2$ and $\gamma=5$ in equal shares and
$\theta=0.3$, and a household of the second type starting at the middle of
its band. The dotted line marks the true share.}
\label{fig:advice}
\end{figure}

\iftablesatend\else
\input{table2}
\fi

\subsection{Relation to Household Evidence}\label{sec:evidence}

We do not estimate the model, but several of its implications can be
compared with household data. Apart from the size of the relative concern,
the comparisons are qualitative.

\paragraph{Size of the relative concern.}
Between revisions $\log c_t=\log h_t+\theta\log Y_t$ with $h_t$ constant, so
$\theta$ is also the elasticity of a household's spending with respect to the
benchmark over periods without revisions. \citet{MM} estimate Euler
equations in which utility depends on consumption relative to a power of the
geometric mean of the current consumption of a peer group, which is the
benchmark used here, and obtain values between 0.30 and 0.44 in the
specifications in which the parameter is significant. In Danish
administrative data, \citet{DFP} estimate an elasticity of own consumption
with respect to the consumption of workplace peers of about 0.3. Values of
this size satisfy the weak interaction condition \eqref{eq:weakcondition}
for every $\gamma\ge1$, because for homogeneous households it reads
$\theta<\gamma/(\gamma+1)$ (Corollary~\ref{cor:weakgrowth}). They also lie
below the bound $1/2$ in Corollary~\ref{cor:rounded}. Our baseline
value is $\theta=0.3$. These estimates motivate but do not identify $\theta$
in our model. For the parameters of Figure~\ref{fig:stationary} the effective
rate $\hat r$ changes sign at $\theta\approx0.401$, inside the range of these
estimates, so the band for zero and negative rates in Lemma~\ref{lem:band} is
needed at empirically relevant values.

\paragraph{Infrequent revisions.}
After an unexpected change in wealth, a household solves its problem again
with its current standard as the inherited one, so the initial rule of
Theorem~\ref{thm:br} applies. If the ratio of consumption to wealth stays
inside the band, the relative standard does not change. If a gain carries
the ratio below the band, the household raises its standard at once, by
just enough to bring the ratio back to the lower end of the band. Small
gains therefore cause few revisions, and both the fraction of households that
revise and the average change in spending per unit of gain increase with the
size of the gain. \citet{FKZ} find this pattern. In their survey data, spending
responses to unanticipated gains are either zero or substantial and rise
with the size of the gain mainly through the extensive margin, and the
authors account for them with a utility cost of changing consumption.
Because the index of Theorem~\ref{thm:br} is reflected at the band,
relative consumption does not jump after time zero unless wealth changes
unexpectedly, whereas with a fixed cost of adjustment consumption also
jumps along anticipated paths.
\citet{Shea} finds that consumption responds more strongly to predictable
income declines than to predictable increases and notes that this asymmetry
is consistent with loss aversion, which ILA builds in through $L>1$.

\paragraph{Spending of others.}
Between revisions a household spends $h_tY_t^\theta$, so its spending rises
with the benchmark even when it does not revise, and holding a relative
standard constant costs more when the benchmark grows faster. In the
stationary equilibrium, faster benchmark growth lowers the effective rate
$\hat r=r-\theta g^\ast$ at which households finance spending, and as the cost
rate $\hat K(\gamma)$ approaches zero, the wealth needed per unit of spending
grows without bound. Evidence that
a household's spending and borrowing rise with the spending or income of its
neighbors is in line with this mechanism. Nonparticipating neighbors of
winners in the Dutch Postcode Lottery have higher car consumption than other
nonparticipants \citep{KKSK}, and in Finland the car purchases of nearest
neighbors raise a household's own purchases \citep{GKI}. Larger lottery
prizes increase subsequent borrowing and bankruptcies among the winners'
neighbors \citep{AMS}. \citet{BeMo} find that nonrich households in the
United States consume a larger share of their income, and spend more on
visible goods, when top incomes are higher.

\paragraph{Peer effects in risk taking.}
Proposition~\ref{prop:composition} and the portfolio rule \eqref{eq:port}
give two channels through which others affect the risk taking of a household
with $\gamma>1$ and $\theta>0$, and neither involves learning. First, a
reference group with more households of low effective risk aversion makes
the benchmark react more strongly to common shocks, and the household then
holds a larger risky share at revision dates although its own preferences
and opportunities are unchanged. Second, because spending moves with
$Y^\theta$, the loading of wealth on common risk contains the term
$(1-1/\Rr_t)\theta s$, which is positive when $s>0$ since $\Rr_t\ge\gamma>1$.
The household therefore holds some of the risk that drives the benchmark even
if this risk carries no premium for it, and between revisions its loadings
move toward $(\theta s,0)$, the loadings of $Y^\theta$. Peer effects in stock
market participation and in investment decisions are well documented
\citep{HKS,BISW,OT}. In a field experiment with a Brazilian brokerage,
\citet{BEFY} separate learning from the utility of owning the asset that a
peer owns and find that both channels matter. \citet{DKK} derive a
preference for the assets held by one's community from relative wealth
concerns, and \citet{HJWZ} attribute heavy trading in local stocks to the
desire to keep up with neighbors.
\FloatBarrier

%% file: table1.tex
\begin{table}[!t]\centering\small
\caption{Accuracy of the Computed Asynchronous Equilibrium.}\label{tab:async}
\setlength{\tabcolsep}{4pt}
\begin{tabular}{@{}lcccc@{}}
\hline
Time step & $1/16$ & $1/32$ & $1/64$ & $1/128$\\
\hline
\multicolumn{5}{@{}l}{\emph{Inner equilibrium in index coordinates}}\\
Lattice points & $6{,}992$ & $18{,}656$ & $49{,}830$ & $137{,}361$\\
Howard iterations & 8 & 7 & 8 & 9\\
Largest complementarity residual & $8.5\times10^{-15}$ & $9.2\times10^{-15}$ & $2.5\times10^{-14}$ & $1.8\times10^{-14}$\\
Synchronized states: error in adjoint & $9.3\times10^{-4}$ & $1.6\times10^{-4}$ & $1.6\times10^{-4}$ & $4.0\times10^{-5}$\\
Synchronized states: error in cost & $2.5\times10^{-3}$ & $2.1\times10^{-3}$ & $1.1\times10^{-3}$ & $1.0\times10^{-3}$\\
Synchronized states: error in risky share & $7.0\times10^{-3}$ & $3.5\times10^{-3}$ & $1.5\times10^{-3}$ & $1.9\times10^{-3}$\\
\multicolumn{5}{@{}l}{\emph{Outer budget equations}}\\
Largest budget residual & $<10^{-15}$ & $<10^{-15}$ & $<10^{-15}$ & $<10^{-15}$\\
Monte Carlo cost minus wealth, in \% & $-0.15$ & $0.25$ & $0.20$ & $0.05$\\
\multicolumn{5}{@{}l}{\emph{Results}}\\
Initial consumption to wealth, in \% & $2.83,\,3.25,\,3.74$ & $2.82,\,3.25,\,3.74$ & $2.82,\,3.25,\,3.74$ & $2.81,\,3.25,\,3.73$\\
Initial risky share, in \% & $57.7,\,44.3,\,42.7$ & $58.3,\,44.3,\,42.7$ & $58.8,\,44.4,\,42.7$ & $59.2,\,44.4,\,42.8$\\
Mean time to synchronization, in years & $24.0$ & $23.4$ & $23.2$ & $23.5$\\
Synchronization within ten years & $0.21$ & $0.22$ & $0.23$ & $0.22$\\
\hline
\end{tabular}

\vspace{4pt}
\begin{minipage}{\textwidth}\footnotesize
\emph{Notes.} Errors on synchronized states compare the lattice solution
with the closed forms of Theorem~\ref{thm:common} on the diagonal of the
index lattice. The adjoint error is scaled by the largest absolute
adjoint, the cost error is relative, and the risky share error is absolute.
The Monte Carlo row reports the largest deviation over the three groups of
the simulated cost of consumption from wealth, from $20{,}000$ paths over
$200$ years, with standard errors of at most $0.68\%$. The simulations
use the interpolated policy at states off the lattice. Initial ratios and risky shares are
listed for groups 1, 2, and 3. Synchronization means that the three log
indices coincide, and the last row gives the probability that this happens
within ten years.
\end{minipage}
\end{table}

%% file: table2.tex
\begin{table}[!t]\centering\small
\caption{Plans Computed at a Target Composition.}\label{tab:advice}
\begin{tabular}{@{}rcccccc@{}}
\hline
 & & & \multicolumn{2}{c}{Certificate $\varepsilon_a$, \% of wealth} &
 \multicolumn{2}{c}{Effect of recomputing, \% of wealth}\\
$N$ & $N_1$ & $d_N$ & Type 1 & Type 2 & Type 1 & Type 2\\
\hline
10 & 3 & 0.067 & $1.4\times10^{-4}$ & $3.7\times10^{-3}$ & $\phantom{-}2.7\times10^{-2}$ & $-4.0\times10^{-3}$\\
20 & 7 & 0.033 & $3.4\times10^{-5}$ & $9.0\times10^{-4}$ & $-1.2\times10^{-2}$ & $\phantom{-}4.8\times10^{-3}$\\
50 & 17 & 0.013 & $5.4\times10^{-6}$ & $1.4\times10^{-4}$ & $-4.9\times10^{-3}$ & $\phantom{-}1.7\times10^{-3}$\\
100 & 33 & 0.0067 & $1.4\times10^{-6}$ & $3.6\times10^{-5}$ & $\phantom{-}2.5\times10^{-3}$ & $-7.5\times10^{-4}$\\
1000 & 333 & 0.00067 & $1.4\times10^{-8}$ & $3.6\times10^{-7}$ & $\phantom{-}2.5\times10^{-4}$ & $-7.8\times10^{-5}$\\
\hline
\end{tabular}

\vspace{4pt}
\begin{minipage}{\textwidth}\footnotesize
\emph{Notes.} Types $(\gamma,\theta)=(2,0.3)$ and $(5,0.5)$ with the market and
adjustment costs of Figure~\ref{fig:composition}, target share $1/3$ of
type 1, wealth $40$, and inherited consumption $0.80$ and $0.85$. The
count $N_1$ is the closest integer to $N/3$ and $d_N=\sum_a|N_a/N-\pi_a|$. The certificate
bounds the gain of a household from any deviation when every household keeps
the target plan. The last two columns report utility in the equilibrium
recomputed at the actual composition minus utility under the target plans.
Utilities are converted into wealth by dividing by the multiplier of each
type, the marginal utility of wealth. Unlike the compensating wealth of
Figure~\ref{fig:advice}, this conversion is linear, which is accurate for
differences of this size.
\end{minipage}
\end{table}

%% file: sec4_conclusion.tex
\section{Discussion}\label{sec:conclusion}

We have studied consumption and investment when households compare their
consumption with that of others and find it costly to revise their relative
standard. For finitely many types, a weak interaction condition, together
with conditions that keep each individual problem well posed, guarantees a
mean field equilibrium for arbitrary initial wealth and inherited standards,
and relative concerns in the range of micro estimates satisfy this
condition when risk aversion is at least one. In a complete market with common risk, any mean field
equilibrium with finitely many types is an exact Nash equilibrium of every
finite game whose type frequencies equal its weights, and a computable
certificate bounds the gain from deviating when plans computed for another
benchmark are used. With identical households and common risk, the relative
concern enters the unique symmetric equilibrium only through the curvature
of an otherwise standard loss averse household. Adjustment costs and
the initial cross section leave the long run growth of the benchmark
unchanged.

\subsection{Scope of the Results}\label{sec:scope}

\emph{Uniqueness.} For fixed multipliers the benchmark is unique among benchmarks whose logarithm stays within a
bounded distance of $gt+sB_t$ (Theorem~\ref{thm:weakexist}), at a given price each household has exactly
one multiplier at which its budget binds, and identical households facing common
risk have exactly one symmetric equilibrium for every population size
(Theorem~\ref{thm:symmetricunique}). As in \citet{ABF}, uniqueness of
the multipliers that make all budgets bind requires a lower bound on the
response of expenditure to the own multiplier, which inaction makes hard to
obtain. Households facing the same price forget their inherited standards
at their first revision in the relevant direction, and this yields
uniqueness among such benchmarks under an explicit strengthening of
\eqref{eq:weakcondition}. The results on finite populations do not
depend on uniqueness: Theorem~\ref{thm:exact} applies to every mean field
equilibrium with finitely many types, Theorem~\ref{thm:populationlimit}
gives convergence along subsequences and of the whole sequence when the
limit is unique, and the certificate of Theorem~\ref{thm:certificate}
applies whichever equilibrium the plans come from.

\emph{Computed equilibria.} For one type the lattice scheme
reduces to a monotone Markov chain approximation in one dimension of the
kind analyzed by \citet{KD}. With several types a revision of one type also
moves the indices of all other types. We know no comparison
principle for the resulting system of coupled obstacle problems, and
convergence as the time step tends to zero is not proved.
Section~\ref{sec:computeexample} checks the computation against the closed forms on synchronized states, across four
time steps, using Monte Carlo budget diagnostics for the interpolated
policies, and against the
bound \eqref{eq:existbound}. Theorems~\ref{thm:exact} and~\ref{thm:certificate} assume
exact best responses, so for the computed plans of
Section~\ref{sec:computeexample} they apply to the equilibrium that these
plans approximate rather than to the plans themselves. The example has
groups with identical preferences, which lets us store the adjoints on a
sorted part of the lattice, and the two examples with two types of
different preferences are computed on the full lattice. With three or more
types and different preferences the lattice grows quickly.

\emph{Common risk and complete markets.} Theorem~\ref{thm:exact} uses that
a household can finance every plan whose value at the state price density
does not exceed its wealth. When households cannot trade the private risks
of others, the reflected plans of the stationary equilibrium cannot be
carried over to a finite population of identical households with
$\theta>0$, although the finite game may still
have equilibria. The mean field results of
Theorems~\ref{thm:stat} and~\ref{thm:weakexist} include private risk.

\emph{Modeling choices.} The benchmark is a geometric average, as in the
Euler equation estimates of \citet{MM}. This form lets a deviation enter the
benchmark as a power, which Theorem~\ref{thm:exact} uses. Asset returns are exogenous, as in the relative
performance games of \citet{ET}, \citet{LZ}, and \citet{LS}, so equilibrium
determines consumption, portfolios, and the benchmark but not prices. The
existence proof applies Brouwer's theorem to finitely many multipliers,
which matters only for limits with type distributions of infinite support.
The weak interaction
condition \eqref{eq:weakcondition} is sufficient but not necessary.

\subsection{Open Problems}\label{sec:open}

Whether the equilibrium is unique without a strengthening of \eqref{eq:weakcondition} remains open, and so is existence beyond
the weak interaction condition and for continuous distributions of
preference parameters. On the computational side, a proof that the lattice
scheme converges, combined with an a posteriori bound that adds the
discretization error to the certificate of Theorem~\ref{thm:certificate},
would make computed policies verified approximate equilibria. Numerical schemes for mean field games \citep{AC} and approximate
aggregation \citep{KrS} could be applied to larger type spaces, but their
accuracy would need careful assessment \citep{DH}. In finite populations in which households cannot trade the
private risks of others, an approximation result would require policies
that can be financed, and the reflected construction does not provide them.
Endogenous asset prices, finite horizons, and bequests would change both the
individual problem and the consistency condition.
Finally, panel data on spending and portfolios could identify the relative concern and the asymmetry of
adjustment costs jointly and would discipline the quantitative statements
of Section~\ref{sec:implications}.

%% file: appendixA.tex
\FloatBarrier
\section{Proofs of the Main Equilibrium Results}\label{app:proofs}

Appendices~\ref{ec:individual} to~\ref{ec:further} contain the remaining
proofs and the computational details, and Appendix~\ref{sup:intro} collects further results. Appendix~\ref{ec:individual} treats the individual problem,
including the proofs of Lemma~\ref{lem:band} and Theorem~\ref{thm:br},
deflators between constant multiples of a geometric deflator
(Theorem~\ref{ec:thm:comparable}), and a bound on the gain from any
deviation. Appendix~\ref{ec:existence} proves Theorems~\ref{thm:weakexist},
\ref{thm:populationlimit}, and~\ref{thm:certificate} and
Corollaries~\ref{cor:weakgrowth} and~\ref{cor:rounded} and treats finite
games, private risk, and types without intertemporal loss aversion.
Appendix~\ref{ec:compute} describes the lattice scheme, its diagnostics, and
further examples, and Appendix~\ref{ec:further} proves uniqueness with all
budgets imposed under a strengthening of \eqref{eq:weakcondition} and
synchronization of revisions for groups with common preferences.
Throughout Appendices~\ref{ec:individual} to~\ref{sup:intro}, $\varsigma=\kappa^2/2$.

The proofs below use four results from Appendix~\ref{ec:individual}: the sufficient
optimality conditions of Theorem~\ref{ec:thm:foc} for the Lagrangian
\begin{equation}\label{eq:lagrangian}
 \J_P(h)=\E\int_0^\infty e^{-\delta t}\Big(\big(u(h_t)-P_t\psi_p(h_t)\big)\dd t
        -\alpha\,\dd v_t^{+}-\beta\,\dd v_t^{-}\Big),\qquad
 \psi_p(h)=h^{1/p},\ p\in(0,1],
\end{equation}
stated in terms of the index $z_t=P_t\psi_p'(h_t)/u'(h_t)$ and its adjoint
$\Psi_t=\E[\int_t^\infty e^{-\delta(s-t)}(1-z_s)\dd s\,|\,\F_t]$, which
must stay in $[-\beta,\alpha]$, equal $\alpha$ whenever $h$ increases, and
equal $-\beta$ whenever $h$ decreases, the
representation $\Psi_t=H(Z_t)$ for the reflected index (Lemma~\ref{ec:lem:fk}),
the cost formula (Lemma~\ref{ec:lem:cost}), and the integrability of
regulators (Lemma~\ref{ec:lem:variation}). Throughout, $\delta t+\log\xi_t=\nu
t-\kappa B_t$ when $\kappa_p=0$.

\subsection{Common Risk}

\begin{proof}[Proof of Theorem~\ref{thm:common}]
Fix $i$ and let $p_i=1-\theta_i/N$ and $\Pi_i=\prod_{j\ne i}(c^j)^{1/N}$. Then
$h^i=(c^i)^{p_i}\Pi_i^{-\theta_i}$, and the expenditure of household $i$ is
$c^i=\varpi\,\psi_{p_i}(h^i)$ with $\varpi=\Pi_i^{\theta_i/p_i}$. Its problem is
that of Theorem~\ref{ec:thm:foc} with $P_t=\lambda e^{\delta t}\xi_t\varpi_t$.
Since $\vb_N=N^{-1}\sum_j\varphi_j$, along the candidate $\log\Pi_i=\mathrm{const}+(\vb_N-\varphi_i/N)\ell$ and
$\log h^i=\mathrm{const}+(\varphi_i-\theta_i\vb_N)\ell$, so the index
satisfies
\[
 \log z^i_t=\mathrm{const}+\delta t+\log\xi_t
 +\Big[\frac{\theta_i}{p_i}\Big(\vb_N-\frac{\varphi_i}N\Big)
 +\Big(\gamma_i+\frac1{p_i}-1\Big)(\varphi_i-\theta_i\vb_N)\Big]\ell_t .
\]
The bracket equals $\gamma_i\varphi_i-(\gamma_i-1)\theta_i\vb_N$, which is
$1$ by \eqref{eq:AN}. Choosing $\lambda$ so that $z^i_0=z_0$ gives
$\log z^i_t=\log z_0+\nu t-\kappa B_t+\ell_t=x_t$ for every $i$, and
Lemma~\ref{ec:lem:fk} gives $\Psi^i=H(e^{x})\in[-\beta,\alpha]$. Since
$\varphi_i-\theta_i\vb_N=(1-\theta_i\vb_N)/\gamma_i\ge0$, relative
consumption increases only when $\dd\la>0$, where $\Psi^i=\alpha$, and
decreases only when $\dd\lb>0$, where $\Psi^i=-\beta$. Hence the conditions
of Theorem~\ref{ec:thm:foc} hold.
Since $\varphi_i>0$, Lemma~\ref{ec:lem:cost} with $\rho=1/\varphi_i$ shows
that $\E\int\xi_tc^i_t\dd t=c^i_{0-}\Xf_{1/\varphi_i}(z_0)=X^i_0$, so the
budget binds. Because $(1-\gamma_i)(\varphi_i-\theta_i\vb_N)=\varphi_i-1$,
utility along the plan is a constant times $e^{(\varphi_i-1)\ell}$, and
Lemma~\ref{ec:lem:variation} with $K(1/\varphi_i)>0$ gives admissibility
(logarithmic estimates of the same lemma apply when $\gamma_i=1$).
Theorem~\ref{ec:thm:foc} shows that $c^i$ is the unique best response. The
mean field case is the same computation with $p=1$ and $\Pi_i$ replaced by
$Y$, where $\log Y=\log\bar c_{0-}+\vb\ell$ by consistency. Compactness of
the type parameters and integrability of the inherited log consumption make
\eqref{eq:consistency} well defined. For the version with an upward revision
at time zero, the same computation applies after time zero with $z_0=\za$,
and the budget binds by the choice of $c^i_0$. The condition on the
logarithmic increments states that the relative consumption of no
household jumps downward at time zero. Upward jumps occur where
$\Psi^i_0=H(\za)=\alpha$, which is the support condition for the initial
jump. Averaging the condition shows that no household's consumption falls
at time zero. The downward version is symmetric, with $z_0=\zb$ and
$\Psi^i_0=H(\zb)=-\beta$ at every downward jump.
\end{proof}

\begin{proof}[Proof of Corollary~\ref{cor:common}]
For identical types \eqref{eq:AN} gives $\vb_N=\varphi_i=1/\hg$, and
$\hg-\theta=\gamma(1-\theta)>0$ gives $\theta\vb_N\le1$.
Lemma~\ref{ec:lem:cost} yields $X_t=c_t\Xf_{\hg}(Z_t)$ and the interval in
(ii). Since $c$ has finite variation, It\^o's formula gives the risky share
$-(\kappa/\sigma)Z\Xf'_{\hg}(Z)/\Xf_{\hg}(Z)$. At either endpoint
\eqref{eq:XH} and $H'=0$ make it $\kappa/(\sigma\hg)$, and in the interior
$H'<0$ makes it strictly smaller. The index in (i) is the index of the
proof of Theorem~\ref{thm:common}. The share $\kappa/(\sigma\hg)$ equals the
equilibrium share under common noise of \citet{LZ} and \citet{DRS} for identical
agents, since $1/\hg=\gamma^{-1}/(1+\theta(\gamma^{-1}-1))$.
\end{proof}

\begin{proof}[Proof of Theorem~\ref{thm:symmetricunique}]
For a symmetric profile $c^i=c$ the relative consumption is $c^{1-\theta}$,
and for power and logarithmic utility
\begin{equation}\label{eq:symmetricutility}
 u_\gamma(c^{1-\theta})=(1-\theta)u_{\hg}(c).
\end{equation}
The same factor relates the two signed variation measures, including their
initial jumps. A symmetric profile is therefore admissible exactly
when the corresponding plan of the single household problem with risk
aversion $\hg$ is, and its utility is that of this plan multiplied by
$1-\theta$. Let $\tilde c$ be the unique optimum of that problem, which exists by Lemma~\ref{lem:band} and Theorem~\ref{thm:br}. It is an equilibrium,
because the computation in the proof of Theorem~\ref{thm:common} with
$\varphi_i=1/\hg$ applies, including to the initial jump, whose direction is
that of the jump in relative consumption.

Conversely, let $c$ be a symmetric equilibrium. For finite $N$ put
$p=1-\theta/N$ and $q=(1-\theta)/p\in(0,1]$. Hold the consumption $c$ of the
others fixed and give one household the plan
$\check c_t=c_t^{1-q}\tilde c_t^{\,q}$. Its relative consumption after the
deviation is $\check c_t^{\,p}c_t^{-\theta(N-1)/N}=\tilde c_t^{\,1-\theta}$, the reference at
$0-$ agrees, and this relative plan is admissible. H\"older's inequality
for the measure $\xi_t(\omega)\,\dd t\,\mathbb P(\dd\omega)$ gives
\[
 \E\int_0^\infty\xi_t\check c_t\,\dd t\le
 \Big(\E\int_0^\infty\xi_tc_t\,\dd t\Big)^{1-q}
 \Big(\E\int_0^\infty\xi_t\tilde c_t\,\dd t\Big)^{q}\le X_0,
\]
and since the market is complete, the deviation can be financed. If $c$ differed
from $\tilde c$, uniqueness of the optimum for a single household and
\eqref{eq:symmetricutility} would make the deviation strictly profitable.
For the mean field game use $q=1-\theta$: with the benchmark fixed at $c$,
the plan $c^\theta\tilde c^{\,1-\theta}$ has relative consumption
$\tilde c^{\,1-\theta}$, and the same argument applies.
\end{proof}

\begin{proof}[Proof of Proposition~\ref{prop:composition}]
The aggregate exponent in \eqref{eq:AN} is $\vb(w)$, because
$1+\theta_i(\gamma_i^{-1}-1)=\hg_i/\gamma_i$, and
$w\varphi_1(w)+(1-w)\varphi_2(w)=\vb(w)$. The reflected index gives
$\ell_t=x_t-x_0-\nu t+\kappa B_t$ with $\nu=-Q$. Since $x$ is bounded and
$B_t/t\to0$, $\ell_t/t\to Q$ almost surely, and \eqref{eq:commonplan} gives
$\log Y_t=\log Y_0+\vb(w)\ell_t$, which proves the growth statement. With
$\rho_i=1/\varphi_i(w)$, the wealth of type $i$ is
$X^i_t=c^i_t\Xf_{\rho_i}(Z_t)$, its risky share is
$-(\kappa/\sigma)Z_t\Xf'_{\rho_i}(Z_t)/\Xf_{\rho_i}(Z_t)$, and at either
endpoint \eqref{eq:XH} and $H'=0$ make it $\kappa\varphi_i(w)/\sigma$. Direct
differentiation of $\vb$ gives the stated derivative, whose numerator is
$(\hg_2-\hg_1)/(\gamma_1\gamma_2)$, and differentiation of $\varphi_i(w)$
completes the proof.
\end{proof}

\subsection{Idiosyncratic Risk and Aggregation}

\begin{proof}[Proof of Proposition~\ref{prop:agg}]
By definition,
$\log c^\ast_t=\theta\log Y_t+\gamma^{-1}\big(\log z^\ast_t-\log\lambda
-\delta t-\log\xi_t-\theta\log Y_t\big)$. Take conditional expectations
given $\F^B_t$ and use \eqref{eq:consistency}, \eqref{eq:spd}, the
independence of the type from $(B,W)$, and $\E[W_t\,|\,\F^B_t]=0$. The
coefficient on the left is positive because $1-\theta+\theta/\gamma>0$, and
the limit follows from $B_t/t\to0$.
\end{proof}

\begin{proof}[Proof of Theorem~\ref{thm:stat}]
By Theorem~\ref{thm:br} with $s=0$, the plan
$h^\ast_t=h_{0-}\exp(\ell_t/\gamma)$ is the unique best response, where
$x_t=\log\hat z_0+\hat\nu t-\kappa W_t+\ell_t$ and
$\hat\nu=\delta-\hat r-\kappa^2/2$. No initial revision occurs because
$X_0/(Y_0^\theta h_{0-})$ lies in the interval of Theorem~\ref{thm:br}. The density
proportional to $e^{\hat\zeta y}$ is stationary for this reflected Brownian
motion with drift \citep[Chapter 5]{H}, so $\E[x_t-x_0\,|\,\text{type}]=0$
and $\E[\ell_t\,|\,\text{type}]=-\hat\nu t$. Hence
\[
 \E[\log c^\ast_t]=\E[\theta](\log Y_0+g^\ast t)+\E[\log h_{0-}]
 +\E\Big[\frac{r+\frac12\kappa^2-\delta-\theta g^\ast}{\gamma}\Big]t .
\]
The constant equals $\log Y_0$ by the choice of $Y_0$ and the slope equals
$g^\ast$ by \eqref{eq:gstar}, which is \eqref{eq:consistency}. For the
converse, the best responses are those of Theorem~\ref{thm:br}, their indices
remain in $[\hat z_\alpha,\hat z_\beta]$, these intervals lie in a fixed
compact set because the types do, and Proposition~\ref{prop:agg} gives
$g=\lim_t\log Y_t/t=g^\ast$.
\end{proof}

\begin{proof}[Proof of Corollary~\ref{cor:stat}]
Part (i) follows from \eqref{eq:gstar}. Between increases of the regulators
$h$ is constant, so $c=hY^\theta$ grows at rate $\theta g^\ast$. For (iii),
work conditionally on the type, write $a=\log\hat z_\alpha$,
$b=\log\hat z_\beta$, and $\varsigma=\kappa^2/2$. The invariant density
satisfies $\varsigma\hat p'=\hat\nu\hat p$, so integration by parts gives
$\int_a^b(\varsigma f''+\hat\nu f')\hat p\,\dd x=\varsigma[f'\hat p]_a^b$ for
$f\in C^2$. Stationarity and It\^o's formula give
$0=\varsigma[f'\hat p]_a^b\,t+f'(a)\E\la_t-f'(b)\E\lb_t$, and choosing the
two endpoint derivatives independently yields
$\E\la_t=\varsigma\hat p(a)t$ and $\E\lb_t=\varsigma\hat p(b)t$. For
identical households consistency gives $(1-\theta)g^\ast t=\E\ell_t/\gamma$,
which is the stated identity. Setting $s=0$ in \eqref{eq:port} proves (iv).
\end{proof}

\begin{proof}[Proof of Proposition~\ref{prop:comparativestat}]
Differentiate $g^\ast=Q/\hg$ and $\hat r=r-\theta Q/\hg$, using
$\partial\hg/\partial\theta=1-\gamma$. The thresholds are strictly increasing
in the rate by Lemma~\ref{lem:band}(i), and $\hat\kappa$, $\delta$,
$\alpha$, and $\beta$ do not depend on $\theta$.
\end{proof}

\subsection{Exact Implementation}

\begin{proof}[Proof of Theorem~\ref{thm:exact}]
The geometric mean of the assigned consumption is
$\prod_a(c^a_t)^{N_a/N}=\exp(\sum_a\pi_a\log c^a_t)=Y_t$, every assigned plan
can be financed from the wealth of its type, and its relative consumption is
the admissible plan $h^a$. Fix a household of type $a$, hold the consumption
of the others fixed, and write $c^\ast=c^a$, $h^\ast=h^a$, $\theta=\theta_a$,
and $X=X_a$. Let $\tilde c$ be a deviation that can be financed and whose
relative consumption in the finite game is admissible, and put
$p=1-\theta/N\in(0,1]$. The induced benchmark and relative consumption are
\[
 \tilde Y_t=Y_t(\tilde c_t/c^\ast_t)^{1/N},\qquad
 \tilde h_t=\tilde c_t\tilde Y_t^{-\theta}=h^\ast_t(\tilde c_t/c^\ast_t)^p .
\]
The reference at $0-$ is still $h_{a,0-}$, so the utility of the deviation
in the finite game is $U_a(\tilde h)$, including its initial adjustment
cost. For $p<1$, H\"older's inequality for the measure
$\xi_t(\omega)\,\dd t\,\mathbb P(\dd\omega)$ gives
\[
 \E\int_0^\infty\xi_tY_t^\theta\tilde h_t\,\dd t
 =\E\int_0^\infty\xi_t(c^\ast_t)^{1-p}\tilde c_t^{\,p}\,\dd t
 \le\Big(\E\int_0^\infty\xi_tc^\ast_t\,\dd t\Big)^{1-p}
    \Big(\E\int_0^\infty\xi_t\tilde c_t\,\dd t\Big)^{p}\le X,
\]
and the bound is immediate for $p=1$. With random initial primitives the
same argument applies conditionally on the type. Thus $\tilde h$ is
admissible and affordable for a price taker facing $Y$, and optimality in
the mean field equilibrium gives $U_a(\tilde h)\le U_a(h^\ast)$. Replacing
$(N,N_a)$ by $(JN,JN_a)$ preserves the benchmark, and the argument applies
with $p=1-\theta/(JN)$. Since the plans are unchanged, so are the financing
wealth $X^a_t=\xi_t^{-1}\E[\int_t^\infty\xi_sc^a_s\dd s\,|\,\F_t]$, the
portfolios, and the utilities.
\end{proof}

%% file: ecompanion.tex
\section{Individual Optimization}\label{ec:individual}

\subsection{Optimality Conditions}\label{ec:foc}

The results of this subsection hold for any positive price process. They also cover
the finite game, in which expenditure is a convex power of relative
consumption. Let $P$ be a positive progressively measurable process, let
$p\in(0,1]$ and $\psi_p(h)=h^{1/p}$, and consider the Lagrangian $\J_P$ of
\eqref{eq:lagrangian} on $\mathcal A$. It is well defined in
$[-\infty,\infty)$. For $h^\ast\in\mathcal A$ put
\begin{equation}\label{ec:eq:index}
 z^\ast_t=\frac{P_t\,\psi_p'(h^\ast_t)}{u'(h^\ast_t)},\qquad
 \Psi^\ast_t=\E\Big[\int_t^\infty e^{-\delta(s-t)}\big(1-z^\ast_s\big)\dd s
 \,\Big|\,\F_t\Big],
\end{equation}
where $\Psi^\ast$ denotes the right continuous version.

\begin{theorem}\label{ec:thm:foc}
Let $h^\ast\in\mathcal A$ have a bounded index $z^\ast$ and
$\J_P(h^\ast)>-\infty$. Suppose that almost surely
\begin{equation}\label{ec:eq:foc}
 -\beta\le\Psi^\ast_t\le\alpha\ \text{ for all }t\ge0,\qquad
 \int_{[0,\infty)}e^{-\delta t}\big(\alpha-\Psi^\ast_t\big)\dd v^{\ast+}_t=0,\qquad
 \int_{[0,\infty)}e^{-\delta t}\big(\Psi^\ast_t+\beta\big)\dd v^{\ast-}_t=0 .
\end{equation}
Then $\J_P(h)\le\J_P(h^\ast)$ for every $h\in\mathcal A$, with equality only
if $h=h^\ast$. If moreover $P_t=\lambda e^{\delta t}\xi_t\varpi_t$ for an
$\F_0$ measurable $\lambda>0$ and a positive process $\varpi$, and
$\E\int_0^\infty\xi_t\varpi_t\psi_p(h^\ast_t)\dd t=X_0$, then $h^\ast$ is the
unique maximizer of $U$ over
$\{h\in\mathcal A:\E\int_0^\infty\xi_t\varpi_t\psi_p(h_t)\dd t\le X_0\}$.
\end{theorem}

\begin{proof}
Write $\Phi_p=\psi_p\circ u^{-1}$. Since
$\Phi_p'(v)=\psi_p'(h)/u'(h)=p^{-1}h^{\gamma+1/p-1}$ with $h=u^{-1}(v)$ is
increasing in $v$, the map $v\mapsto v-P_t\Phi_p(v)$ is strictly concave. Let
$h\in\mathcal A$ with $\J_P(h)>-\infty$, and write $v=u(h)$ and
$v^\ast=u(h^\ast)$. Concavity gives
\[
 \big(v_t-P_t\Phi_p(v_t)\big)-\big(v^\ast_t-P_t\Phi_p(v^\ast_t)\big)
 \le\big(1-z^\ast_t\big)\big(v_t-v^\ast_t\big).
\]
The process $A=v-v^\ast$ is right continuous, adapted, of finite variation,
and $A_{0-}=0$ because both plans start from the same reference. Since
$z^\ast$ is bounded and $\E\int_{[0,\infty)}e^{-\delta s}\dd|A|_s<\infty$ by
admissibility, Fubini's theorem yields
\[
 \E\int_0^\infty e^{-\delta t}(1-z^\ast_t)A_t\,\dd t
 =\E\int_{[0,\infty)}\Big(\int_s^\infty e^{-\delta t}(1-z^\ast_t)\dd t\Big)\dd A_s
 =\E\int_{[0,\infty)}e^{-\delta s}\Psi^\ast_s\,\dd A_s .
\]
The last equality replaces the integrand by its optional projection, which
is allowed because $\dd A$ is an optional random measure of integrable
weighted variation \citep[Chapter VI]{DM}. Using
$\dd A=\dd v^{+}-\dd v^{-}-\dd v^{\ast+}+\dd v^{\ast-}$ we obtain
\begin{align*}
 \J_P(h)-\J_P(h^\ast)\le{}&\E\int_{[0,\infty)}e^{-\delta s}
 \Big((\Psi^\ast_s-\alpha)\dd v^{+}_s-(\Psi^\ast_s+\beta)\dd v^{-}_s\Big)\\
 &-\E\int_{[0,\infty)}e^{-\delta s}
 \Big((\Psi^\ast_s-\alpha)\dd v^{\ast+}_s-(\Psi^\ast_s+\beta)\dd v^{\ast-}_s\Big).
\end{align*}
The first expectation is nonpositive by the inequalities in
\eqref{ec:eq:foc}, and the second vanishes by the two support conditions. If
$\J_P(h)=\J_P(h^\ast)$, the concavity inequality is an equality for almost
every $(t,\omega)$, so $v=v^\ast$ almost everywhere and then everywhere by
right continuity. For the last claim let $h$ satisfy the budget. Then
$U(h)\le U(h)+\lambda\big(X_0-\E\int\xi_t\varpi_t\psi_p(h_t)\dd t\big)
=\J_P(h)+\lambda X_0\le\J_P(h^\ast)+\lambda X_0=U(h^\ast)$.
\end{proof}

The Lagrangian is strictly concave in utility coordinates, although utility
need not be concave as a functional of consumption. Let $I_\gamma=u((0,\infty))$.
A path is \emph{locally interior} if its values and left limits on every
bounded time interval lie in a compact subset of $I_\gamma$. All reflected
plans in this paper have this property.

\begin{theorem}\label{ec:thm:necessary}
Suppose $h^\ast$ maximizes $\J_P$, $\J_P(h^\ast)>-\infty$, $z^\ast$ is bounded,
and $v^\ast=u(h^\ast)$ is locally interior. Then \eqref{ec:eq:foc} holds.
Within this regularity class, \eqref{ec:eq:foc} is necessary and sufficient
for optimality.
\end{theorem}

\begin{proof}
Write $v=v^\ast$, $z=z^\ast$, $\Psi=\Psi^\ast$, and
$k_{\alpha\beta}(a)=\alpha a^++\beta a^-$. The condition $\alpha+\beta>0$
makes $k_{\alpha\beta}$ convex, positively homogeneous, and subadditive,
even when $\alpha<0$. Its integral against $\dd v$ is the adjustment cost.
Put $M_{\alpha\beta}=\max\{|\alpha|,\beta\}$.

Fix a bounded stopping time $\tau$, an event $A\in\F_\tau$, and a
deterministic $T>\|\tau\|_\infty$. Let $\vartheta_n\le T$ be the stopping times at
which $v$ leaves increasing compact subintervals of $I_\gamma$, also
localizing its variation. Local interiority gives $\vartheta_n\uparrow T$ almost
surely. One may first restrict to the events where $\tau<\vartheta_n$ and then
exhaust them. For small $\varepsilon>0$, adding
$\varepsilon\ind_A\ind_{[\tau,\vartheta_n)}$ to $v$ is admissible. Before
$\vartheta_n$, $P=z/\Phi_p'(v)$ is bounded after localization, so the perturbed
expenditure is finite, and after $\vartheta_n$ the path is unchanged.
Subadditivity bounds the extra adjustment cost by
$\varepsilon\ind_A(\alpha e^{-\delta\tau}+\beta e^{-\delta\vartheta_n})$.
We differentiate the running payoff at $\varepsilon=0$, which localization
justifies, use optimality, and let $n\to\infty$ and then $T\to\infty$.
Boundedness of $z$ then gives
$\E[\ind_Ae^{-\delta\tau}\Psi_\tau]\le\alpha\E[\ind_Ae^{-\delta\tau}]$. The
negative perturbation gives the lower bound $-\beta$. Using all bounded
stopping times and the c\`adl\`ag optional version of $\Psi$ yields
$-\beta\le\Psi\le\alpha$ up to indistinguishability.

It remains to prove the support conditions. On
$[0,\vartheta_n)$ use the admissible path $v_{0-}+(1-\varepsilon)(v-v_{0-})$, and at
$\vartheta_n$ return to $v$. Before the return every signed increment, including
the initial jump, is multiplied by $1-\varepsilon$, and at the return
subadditivity bounds the additional cost by
$\varepsilon M_{\alpha\beta}e^{-\delta\vartheta_n}|v_{\vartheta_n-}-v_{0-}|$. Optimality
and differentiation imply
\[
 \E\int_{[0,\vartheta_n)}e^{-\delta t}
 \big((\alpha-\Psi_t)\dd v_t^++(\beta+\Psi_t)\dd v_t^-\big)
 \le(M_{\alpha\beta}+\|\Psi\|_\infty)\,
 \E\big[e^{-\delta\vartheta_n}|v_{\vartheta_n-}-v_{0-}|\big].
\]
The boundary term comes from the adjoint identity
$\E\int_0^{\tau'}e^{-\delta t}(1-z_t)(v_t-v_{0-})\dd t
=\E\int_{[0,\tau')}e^{-\delta t}\Psi_t\dd v_t
-\E[e^{-\delta\tau'}\Psi_{\tau'}(v_{\tau'-}-v_{0-})]$ for stopping times
$\tau'$, which follows by optional projection or by integration by parts
with the bounded adjoint, first with stopped variation. The same argument
handles simultaneous jumps at the return time. Letting $n\to\infty$ is
justified by the integrable variation on $[0,T]$, and
\[
 \E\big[e^{-\delta T}|v_{T-}-v_{0-}|\big]
 \le\E\int_{[0,T)}e^{-\delta s}e^{-\delta(T-s)}\dd|v|_s\longrightarrow0
\]
as $T\to\infty$ by dominated convergence and admissibility. The nonnegative measure
integrated in the first display therefore has zero expected mass, and its two components vanish almost
surely. All arguments can be made conditionally on $\F_0$.
\end{proof}

\subsection{Adjustment Band and Consumption Cost}\label{ec:band}

\begin{theorem}\label{ec:thm:allrates}
Let $r\in\R$, $\kappa,\delta>0$, $\beta\ge0$, and $-\beta<\alpha<1/\delta$.
There is a unique pair $0<\za<\zb$ and a unique $C^2$ function $H$ on
$[\za,\zb]$ satisfying \eqref{eq:Hode} and \eqref{eq:Hbc}. Moreover
$H'(z)<0$ for $\za<z<\zb$ and $0<\za<1-\delta\alpha<1+\delta\beta<\zb$. The
endpoints depend smoothly on the parameters, with derivatives from the right
at $\beta=0$.
\end{theorem}

\begin{proof}
Write $\nu=\delta-r-\varsigma$, $A_\alpha=1-\delta\alpha$,
$B_\beta=1+\delta\beta$, and $L=B_\beta/A_\alpha>1$. For $d>0$ consider the
Neumann problem
\begin{equation}\label{ec:eq:neumannwidth}
 \varsigma f_d''+\nu f_d'-\delta f_d=-e^s\quad(0<s<d),
 \qquad f_d'(0)=f_d'(d)=0 .
\end{equation}
Multiplying the homogeneous equation by $fe^{\nu s/\varsigma}$ and
integrating shows that it has only the zero solution, so
\eqref{ec:eq:neumannwidth} has a unique solution, and the maximum principle
gives $f_d>0$. Differentiating the equation shows that $f_d'$ solves a Dirichlet
problem with strictly negative forcing $-e^s$ and killing $\delta>0$, and the
strong maximum principle and the boundary point lemma give
\begin{equation}\label{ec:eq:fshape}
 f_d'>0\ \text{on }(0,d),\qquad f_d''(0)>0,\qquad f_d''(d)<0 .
\end{equation}

Set $\varrho(d)=f_d(d)/f_d(0)$. We show that $\varrho$ increases strictly
from $1$ to $\infty$. Let $m_1>0>m_2$ be the roots of
$\varsigma m^2+\nu m-\delta=0$ and let
$w_0(s)=(m_1e^{m_2s}-m_2e^{m_1s})/(m_1-m_2)$, which solves the homogeneous
equation with $w_0(0)=1$, $w_0'(0)=0$, and satisfies $w_0>0$ and $w_0'>0$ on
$(0,\infty)$. The derivative $\dot f_d=\partial_df_d$ at fixed $s$ solves the
homogeneous equation with $\dot f_d'(0)=0$ and $\dot f_d'(d)=-f_d''(d)$, so
$\dot f_d=c_dw_0$ with $c_d=-f_d''(d)/w_0'(d)>0$. Also
$[e^{\nu s/\varsigma}(w_0f_d'-w_0'f_d)]'=-\varsigma^{-1}e^{\nu s/\varsigma}w_0(s)e^s<0$,
so $f_d/w_0$ is strictly decreasing. Since $f_d'(d)=0$,
\begin{equation}\label{ec:eq:Rprime}
 \frac{\varrho'(d)}{\varrho(d)}
 =c_d\Big(\frac{w_0(d)}{f_d(d)}-\frac{w_0(0)}{f_d(0)}\Big)>0 .
\end{equation}
For the limits, let $I_q(d)=\int_0^de^{qs}\dd s$, including $I_0(d)=d$.
Variation of constants gives
\begin{equation}\label{ec:eq:fwidth}
 \delta f_d(0)=\frac{m_1e^{m_1d}I_{1-m_1}(d)-m_2e^{m_2d}I_{1-m_2}(d)}
 {e^{m_1d}-e^{m_2d}},\qquad
 \delta f_d(d)=\frac{m_1I_{1-m_1}(d)-m_2I_{1-m_2}(d)}{e^{-m_2d}-e^{-m_1d}} .
\end{equation}
Both $f_d(0)$ and $f_d(d)$ tend to $1/\delta$ as $d\downarrow0$. As $d\uparrow\infty$,
$f_d(d)$ is asymptotic to $-m_2e^d/[\delta(1-m_2)]$, while $f_d(0)$ is of
order $1$, $d$, or $e^{(1-m_1)d}$ according as $m_1>1$, $m_1=1$, or
$0<m_1<1$. In all cases $\varrho(d)\to\infty$.

Hence there is exactly one $d>0$ with $\varrho(d)=L$. Define
\[
 \za=\frac{A_\alpha}{\delta f_d(0)},\qquad \zb=\za e^d,\qquad
 H(\za e^s)=\frac1\delta-\za f_d(s).
\]
These formulas give the equation and all four boundary conditions, and
\eqref{ec:eq:fshape} gives $H'<0$. At the endpoints
$\varsigma f_d''(0)=\delta f_d(0)-1>0$ and $\varsigma f_d''(d)=\delta f_d(d)-e^d<0$,
which give $\za<A_\alpha$ and $\zb>B_\beta$. Conversely, every solution of
\eqref{eq:Hode}--\eqref{eq:Hbc} yields, after this change of variables, the
unique solution of \eqref{ec:eq:neumannwidth} with $\varrho(d)=L$, which
proves uniqueness. Smooth dependence follows from the linear equation,
\eqref{ec:eq:Rprime}, and the implicit function theorem.
\end{proof}

Let $m_1>0>m_2$ be the roots of
$\varsigma m^2+(\delta-r-\varsigma)m-\delta=0$. For $r\ne0$ the solution
has the representation
\begin{equation}\label{ec:eq:H}
 H(z)=D_1(z/\za)^{m_1}+D_2(z/\za)^{m_2}+\frac1\delta-\frac zr,
\end{equation}
where the coefficients and the band satisfy \eqref{eq:Hbc}. For $r=0$ one has
$m_1=1$, and the resonant representation is
\begin{equation}\label{ec:eq:H0}
 H(z)=E_1\frac z\za+E_2(z/\za)^{m_2}+\frac1\delta
 +\frac{z\log(z/\za)}{\delta+\varsigma}.
\end{equation}

\begin{proposition}\label{ec:prop:rate}
With $(\kappa,\delta,\alpha,\beta)$ fixed, both $\za$ and $\zb$ are strictly
increasing functions of $r\in\R$.
\end{proposition}

\begin{proof}
Extend $H_r$ by $\alpha$ below its lower endpoint and by $-\beta$ above its
upper endpoint. The extension is $C^1$ and piecewise $C^2$. Write
$\mathcal L_rf=\varsigma z^2f''+(\delta-r)zf'$ and
$F_r(f)=\mathcal L_rf-\delta f+1-z$. The extension satisfies $F_r(H_r)\ge0$
where $H_r=\alpha$, $F_r(H_r)=0$ in its continuation interval, and
$F_r(H_r)\le0$ where $H_r=-\beta$. For $r_2>r_1$, on a component where
$w=H_{r_1}-H_{r_2}>0$,
\[
 (\mathcal L_{r_2}-\delta)w
 =F_{r_1}(H_{r_1})-(r_2-r_1)zH_{r_1}'-F_{r_2}(H_{r_2})\ge0,
\]
where $H_{r_1}>-\beta$, $H_{r_2}<\alpha$, and $H_{r_1}'\le0$ give the signs. Both functions
are constant outside a compact interval, and the maximum principle, which
remains valid across the $C^1$ junctions, excludes such a component. Thus
$H_{r_1}\le H_{r_2}$, and both endpoints increase weakly. If either endpoint
were the same, $H_{r_2}-H_{r_1}$ would be positive on the common
continuation interval and would solve an equation with strictly negative
right side, and the boundary point lemma would give a nonzero inward
derivative at the shared endpoint, contradicting smooth fit.
\end{proof}

\begin{lemma}\label{ec:lem:fk}
For the reflected index \eqref{eq:skor} and all $t\ge0$,
$H(Z_t)=\E[\int_t^\infty e^{-\delta(s-t)}(1-Z_s)\dd s\,|\,\F_t]$. This
quantity lies in $[-\beta,\alpha]$, and it equals $\alpha$ exactly when
$Z_t=\za$ and $-\beta$ exactly when $Z_t=\zb$.
\end{lemma}

\begin{proof}
In the variable $x=\log z$, \eqref{eq:Hode} reads
$\varsigma\tilde H''+\nu\tilde H'-\delta\tilde H+1-e^x=0$ for
$\tilde H(x)=H(e^x)$. It\^o's formula gives
$\dd(e^{-\delta t}\tilde H(x_t))=e^{-\delta t}\big(-(1-Z_t)\dd t
-\kappa\tilde H'(x_t)\dd\mathcal B_t+\tilde H'(x_t)(\dd\la_t-\dd\lb_t)\big)$.
The last term vanishes because $\tilde H'=0$ at both endpoints. Since
$\tilde H$ and $\tilde H'$ are bounded, the stochastic integral is a
martingale and $e^{-\delta T}\tilde H(x_T)\to0$. The remaining claims follow
from \eqref{eq:Hbc} and $H'<0$.
\end{proof}

For $r\ne0$ and $K(\rho)>0$ the cost function has the explicit form
\begin{equation}\label{ec:eq:X}
 \Xf_\rho(z)=\frac1r
 -\frac{D_1m_1(z/\za)^{m_1-1}}{(1-\rho+\rho m_1)\za}
 -\frac{D_2m_2(z/\za)^{m_2-1}}{(1-\rho+\rho m_2)\za}.
\end{equation}

\begin{lemma}\label{ec:lem:cost}
Let $\xi_t=\exp(-(r+\varsigma)t-\kappa\mathcal B_t)$, $\rho>0$, and
$C_t=\exp(\ell_t/\rho)$. The cost $\E\int_0^\infty\xi_tC_t\dd t$ is finite
if and only if $K(\rho)>0$. In that case, for every $r\in\R$,
$\E[\int_t^\infty\xi_sC_s\dd s\,|\,\F_t]=\xi_tC_t\Xf_\rho(Z_t)$, where
$\Xf_\rho$ is the unique solution of
\begin{equation}\label{ec:eq:costbvp}
 \varsigma z^2\Xf_\rho''+(\delta-r+\kappa^2)z\Xf_\rho'-r\Xf_\rho+1=0,
 \qquad \Xf_\rho+\rho z\Xf_\rho'=0\ \text{ at }\za,\zb .
\end{equation}
It is positive with $\Xf_\rho'<0$, satisfies \eqref{eq:XH}, and
$e^{-d/\rho}\le K(\rho)\Xf_\rho\le e^{d/\rho}$ with $d=\log(\zb/\za)$.
Formula \eqref{ec:eq:X} holds for $r\ne0$, and \eqref{ec:eq:costbvp} defines
$\Xf_\rho$ continuously through $r=0$.
\end{lemma}

\begin{proof}
Set $q=1/\rho$, $a=\log\za$, and $b=\log\zb$. From \eqref{eq:skor},
\[
 \xi_tC_t=e^{q(x_t-x_0)}e^{-K(\rho)t}
 \exp\big(-\kappa(1-q)\mathcal B_t-\varsigma(1-q)^2t\big).
\]
The last factor is a true exponential martingale on every finite interval.
Under the consistent change of measure the reflected diffusion has drift
$\tilde\nu=\nu+\kappa^2(1-q)$, and Tonelli's theorem yields
\begin{equation}\label{ec:eq:tilted}
 \Xf_\rho(e^x)=e^{-qx}\int_0^\infty e^{-K(\rho)t}\E_x^{\mathbb Q}[e^{qx_t}]\dd t .
\end{equation}
Since $x_t\in[a,b]$, this proves both directions of the finiteness claim,
positivity, and the upper and lower bounds. The Markov property gives the
conditional formula. The function $\Xi(x)=e^{qx}\Xf_\rho(e^x)$ solves
$\varsigma\Xi''+\tilde\nu\Xi'-K(\rho)\Xi=-e^{qx}$ with
$\Xi'(a)=\Xi'(b)=0$, and since its killing rate is positive, the maximum principle gives
uniqueness. Conjugating back gives \eqref{ec:eq:costbvp} and
smoothness. The derivative in $x$ of $\Xf_\rho(e^x)$ solves the homogeneous
cost equation with strictly negative boundary values $-\Xf_\rho/\rho$, and
conjugation by $e^{qx}$ again yields a Dirichlet problem with positive
killing, so this derivative is strictly negative. Finally let
$\mathcal G=\varsigma z^2\partial_{zz}+(\delta-r+\kappa^2)z\partial_z-r$.
Differentiating \eqref{eq:Hode} gives $\mathcal GH'=1$, while
$\mathcal G\Xf_\rho=-1$ and $\mathcal G(z\Xf_\rho')=0$. Hence
$F=\Xf_\rho+\rho z\Xf_\rho'+H'$ solves $\mathcal GF=0$ with zero boundary
values, and the same conjugation gives $F=0$. For $r\ne0$, \eqref{ec:eq:X}
satisfies \eqref{ec:eq:costbvp}, and its denominators are nonzero because the
characteristic polynomial evaluated at $1-1/\rho$ equals $-K(\rho)<0$.
\end{proof}

\begin{lemma}\label{ec:lem:variation}
Let $q=1/\rho$ and $K(\rho)>0$. Then
$\E\int_0^\infty e^{-\delta t}e^{(q-1)\ell_t}\dd t<\infty$ and
$\E\int_0^\infty e^{-\delta t}e^{(q-1)\ell_t}(\dd\la_t+\dd\lb_t)<\infty$.
Also $\E\int_0^\infty e^{-\delta t}|\ell_t|\dd t<\infty$ and
$\E\int_0^\infty e^{-\delta t}(\dd\la_t+\dd\lb_t)<\infty$.
\end{lemma}

\begin{proof}
Put $k=q-1$ and $G_t=e^{-\delta t+k\ell_t}$. Since
$\ell_t=x_t-x_0-\nu t+\kappa\mathcal B_t$,
$\E G_t\le e^{|k|(b-a)}\exp[-(\delta+k\nu-k^2\varsigma)t]
=e^{|k|(b-a)}e^{-K(\rho)t}$. Let $f(x)=(x-a)(b-x)/(b-a)$, which
vanishes at both endpoints and has $f'(a)=1$ and $f'(b)=-1$. With
$\mathcal L_x=\varsigma\partial_{xx}+\nu\partial_x$, It\^o's formula
for $G_tf(x_t)$, in which the regulator coefficient at either endpoint is
$G_t$, gives
\[
 \E\int_0^TG_t(\dd\la_t+\dd\lb_t)
 =\E[G_Tf(x_T)]-f(x_0)-\E\int_0^TG_t(\mathcal L_xf-\delta f)(x_t)\dd t .
\]
To justify this identity, first stop the regulators and then remove the
stopping by bounding $G$ with a geometric Brownian motion, which has moments
of all orders on bounded intervals.
Boundedness of $f$ and its derivatives and the exponential estimate give a
bound uniform in $T$. The case $k=0$ gives the unweighted statement
because its killing rate is $\delta>0$, and
$\E|\ell_t|\le b-a+|\nu|t+\kappa\sqrt t$ gives the last claim.
\end{proof}

\begin{proposition}\label{ec:prop:criticalcost}
Consider a continuous family of market parameters with fixed $\rho>0$,
$K(\rho)\downarrow0$, a nondegenerate limiting band
$[a_c,b_c]=[\log z_{\alpha,c},\log z_{\beta,c}]$, $\kappa_c>0$, and
$\delta_c>0$. Let $q=1/\rho$, let $\tilde\pi_c$ be the density proportional to
$\exp(2\tilde\nu_cx/\kappa_c^2)$ on $[a_c,b_c]$ with
$\tilde\nu_c=\delta_c-r_c-\kappa_c^2/2+\kappa_c^2(1-q)$, and let
$\Upsilon_c=\int e^{qx}\tilde\pi_c(x)\dd x$. Then
$K(\rho)\Xf_\rho(e^x)\to e^{-qx}\Upsilon_c$, uniformly after the bands are
mapped to $[0,1]$. In particular $1/\Xf_\rho(e^x)\sim K(\rho)e^{qx}/\Upsilon_c$.
\end{proposition}

\begin{proof}
For each member of the family let $\tilde\pi$ be the invariant density of the
tilted reflected diffusion in \eqref{ec:eq:tilted} and
$\Upsilon=\int e^{qx}\tilde\pi(x)\dd x$. The centered Poisson equation
$(\varsigma\partial_{xx}+\tilde\nu\partial_x)\chi=\Upsilon-e^{qx}$ with
$\chi'(a)=\chi'(b)=0$ has a solution, obtained by integrating with the factor
$e^{\tilde\nu x/\varsigma}$, and the centering gives the second Neumann
condition. With $\chi(a)=0$ this explicit construction is uniformly bounded
along the family. The discounted It\^o formula gives
$|\int_0^\infty e^{-Kt}(\E_x^{\mathbb Q}[e^{qx_t}]-\Upsilon)\dd t|\le2\|\chi\|_\infty$.
Multiply \eqref{ec:eq:tilted} by $K$, let $K\downarrow0$, and use continuity
of the band and of its invariant density.
\end{proof}

\begin{proof}[Proof of Lemma~\ref{lem:band}]
Part (i) is Theorem~\ref{ec:thm:allrates} together with
Proposition~\ref{ec:prop:rate}, and part (ii) is Lemma~\ref{ec:lem:cost}.
\end{proof}

\subsection{Proof of Theorem~\ref{thm:br}}\label{ec:proofbr}

\begin{proof}
Let $\lambda=\hat z_0u'(h_0)Y_0^{-\theta}$ and
$P_t=\lambda e^{\delta t}\xi_tY_t^\theta$. Since
$\xi_tY_t^\theta=Y_0^\theta\hat\xi_t$,
$\log(P_t/u'(h^\ast_t))=\log\hat z_0+(\delta-\hat r-\hat\kappa^2/2)t
-\hat\kappa\hat{\mathcal B}_t+\ell_t=x_t$, so the index \eqref{ec:eq:index}
of $h^\ast$ with $p=1$ is $\hat Z$. Lemma~\ref{ec:lem:fk} in the artificial
market gives $\Psi^\ast=\hat H(\hat Z)\in[-\beta,\alpha]$. The plan increases
only when $\dd\la>0$, that is when $\hat Z=\hat z_\alpha$ and
$\Psi^\ast=\alpha$, and decreases only when $\Psi^\ast=-\beta$. A jump at
time zero occurs only when $\hat z_0$ is an endpoint, and then in the
correct direction because $\hat\Xf_\gamma$ is decreasing. Hence \eqref{ec:eq:foc}
holds. By Lemma~\ref{ec:lem:cost} with $\rho=\gamma$,
$\E\int\hat\xi_th^\ast_t\dd t=h_0\hat\Xf_\gamma(\hat z_0)=\hat X_0$, so the
budget binds. For $\gamma\ne1$ put $q=1/\gamma$. Utility along the plan is a
constant times $e^{(q-1)\ell}$, and its variation after time zero is
$\gamma^{-1}(h^\ast)^{1-\gamma}(\dd\la+\dd\lb)$, so
Lemma~\ref{ec:lem:variation} gives admissibility. For $\gamma=1$ use
$\log h^\ast=\log h_0+\ell$ and the logarithmic estimates of that lemma. The
initial jump has finite utility variation, and the expenditure is finite, so
$\J_P(h^\ast)>-\infty$. Theorem~\ref{ec:thm:foc} applies with $p=1$ and
$\varpi=Y^\theta$. The wealth formula is Lemma~\ref{ec:lem:cost}, and the
loadings follow from It\^o's formula, since $h^\ast$ has finite variation,
$\dd\hat Z/\hat Z$ has loadings $-(\hat\kappa_a,\hat\kappa_p)$, and $Y^\theta$
has loadings $(\theta s,0)$.
\end{proof}

\subsection{Comparable Prices and Stability}\label{ec:comparable}

In this subsection $p=1$ and the initial reference $h_{0-}>0$ is fixed.
Initial types with finitely many values are treated conditionally. Write
$h^\ast(P,h_{0-})$ for an attained optimum of $\J_P$.

\begin{lemma}\label{ec:lem:lattice}
Write $\J_P(v)$ for the Lagrangian in utility coordinates $v=u(h)$, with the
reference $v_{0-}$ indicated when it matters. For prices $P\le Q$ and
admissible paths $v$ and $w$ with $\J_P(v)>-\infty$ and $\J_Q(w)>-\infty$,
\begin{equation}\label{ec:eq:lattice}
 \J_P(v\vee w)+\J_Q(v\wedge w)\ge\J_P(v)+\J_Q(w),
\end{equation}
where $v\vee w$ and $v\wedge w$ carry the references $v_{0-}\vee w_{0-}$ and
$v_{0-}\wedge w_{0-}$. Consequently, for attained finite optima, relative
consumption is decreasing in the price and increasing in the initial
reference: $P\le Q$ implies $h^\ast(P,h_{0-})\ge h^\ast(Q,h_{0-})$, and
$h_{0-}\le\tilde h_{0-}$ implies $h^\ast(P,h_{0-})\le h^\ast(P,\tilde h_{0-})$, up to
indistinguishability. The same holds for $p\in(0,1]$ whenever the optima
exist.
\end{lemma}

\begin{proof}
Write $\J_P(v)=\E\int e^{-\delta t}(v_t-P_t\Phi(v_t))\dd t-\mathcal C(v)$
with $\Phi=u^{-1}$, or $\Phi=\psi_p\circ u^{-1}$, and
$\mathcal C(v)=\E\int_{[0,\infty)}e^{-\delta t}k_{\alpha\beta}(\dd v_t)$,
$k_{\alpha\beta}(a)=\bar\alpha|a|+\tilde\alpha a$,
$\bar\alpha=(\alpha+\beta)/2>0$, and $\tilde\alpha=(\alpha-\beta)/2$. We
compare the three parts of \eqref{ec:eq:lattice} separately.

\emph{Running utility.} Since $v\vee w+v\wedge w=v+w$, the terms
$\E\int e^{-\delta t}v_t\dd t$ add up to the same value on both sides.

\emph{Expenditure.} On $\{v_t\ge w_t\}$ both sides contain
$P_t\Phi(v_t)+Q_t\Phi(w_t)$. On $\{v_t<w_t\}$ the left side contains
$P_t\Phi(w_t)+Q_t\Phi(v_t)$ and the right side contains
$P_t\Phi(v_t)+Q_t\Phi(w_t)$, which is larger by
$(Q_t-P_t)(\Phi(w_t)-\Phi(v_t))\ge0$. Since expenditure enters with a minus
sign, this part of the left side is at least that of the right side. The
expenditures on the left are finite because on the swapped set the larger
consumption faces the smaller price.

\emph{Adjustment costs.} The signed part of $k_{\alpha\beta}$ is linear, and
the increments of $v\vee w$ and $v\wedge w$ add up to those of $v+w$,
including the initial atom from the references. For the absolute part,
consider a partition $0-=t_0<t_1<\cdots$ and the elementary inequality
\[
 |\max(a,c)-\max(b,d)|+|\min(a,c)-\min(b,d)|\le|a-b|+|c-d|
\]
with $a=v_{t_{i+1}}$, $b=v_{t_i}$, $c=w_{t_{i+1}}$, and $d=w_{t_i}$.
Multiplying by $e^{-\delta t_i}$, summing, and refining the partition gives
$\mathcal C(v\vee w)+\mathcal C(v\wedge w)\le\mathcal C(v)+\mathcal C(w)$,
because the discounted variation of a right continuous path of finite
variation is the limit of these weighted sums. Adding the three parts
proves \eqref{ec:eq:lattice}, and the paths on the left are admissible.

\emph{Comparisons.} Attained finite optima are unique because $\J_P$ is
strictly concave in utility coordinates. Let $v$ be the optimum at $P$ and
$w$ the optimum at $Q\ge P$, with the same reference. By optimality the
right side of \eqref{ec:eq:lattice} is at least the left side, so equality
holds termwise, $v\vee w$ is optimal at $P$, and uniqueness gives
$v\vee w=v$, that is $v\ge w$. For the references take $P=Q$, let $v$ be
the optimum for $h_{0-}$ and $w$ the optimum for $\tilde h_{0-}\ge h_{0-}$.
Then $v\vee w$ carries the reference $\tilde h_{0-}$ and $v\wedge w$ the
reference $h_{0-}$, and the same argument gives $v\wedge w=v$, that is
$v\le w$.
\end{proof}

\begin{theorem}\label{ec:thm:comparable}
Let $P^0_t=\exp(\nu_0t-\bm\kappa_0\cdot\bm W_t)$ for a finite dimensional
Brownian motion $\bm W$ with $|\bm\kappa_0|>0$, and let
$r_0=\delta-\nu_0-|\bm\kappa_0|^2/2$. Assume $K_0(\gamma)>0$ in this market
and $\underline mP^0\le P\le\overline mP^0$ for constants
$0<\underline m\le\overline m<\infty$. Then $\J_P$ has a unique finite optimum
$h^\ast(P,h_{0-})$ in $\mathcal A$. If $[z_{\alpha,0},z_{\beta,0}]$ is the
band of Lemma~\ref{lem:band} in this market, then
\begin{equation}\label{ec:eq:comparableband}
 \frac{\underline m}{\overline m}z_{\alpha,0}\le P_t\,h^\ast(P,h_{0-})_t^{\gamma}
 \le\frac{\overline m}{\underline m}z_{\beta,0}.
\end{equation}
The optimum is locally interior and satisfies \eqref{ec:eq:foc}. If two
prices satisfy the hypotheses and $e^{-\varepsilon}P\le Q\le e^{\varepsilon}P$,
then
\begin{equation}\label{ec:eq:loglipschitz}
 \sup_{t\ge0}\Big|\log\frac{h^\ast(Q,h_{0-})_t}{h^\ast(P,h_{0-})_t}\Big|
 \le\frac{\varepsilon}{\gamma}\quad\text{almost surely.}
\end{equation}
For every $X>0$ the budget problem with expenditure
$\E\int e^{-\delta t}P_th_t\dd t\le X$ has a unique optimum and a unique
positive multiplier at which the budget binds.
\end{theorem}

\begin{proof}
\emph{Step 1: envelopes.} For a constant $c>0$, the price $cP^0$ is a
geometric price in the market $(r_0,\bm\kappa_0)$. By
Lemmas~\ref{ec:lem:fk} to~\ref{ec:lem:variation} and
Theorem~\ref{ec:thm:foc}, the reflected policy of Theorem~\ref{thm:br},
started at the projection of $ch_{0-}^\gamma$ onto
$[z_{\alpha,0},z_{\beta,0}]$, is the unique finite optimum
$h^\ast(cP^0,h_{0-})$, and its index $cP^0_th_t^\gamma$ stays in this band.
Let $\underline h=h^\ast(\overline mP^0,h_{0-})$, $\overline h=h^\ast(\underline mP^0,h_{0-})$,
$\underline v=u(\underline h)$, and $\overline v=u(\overline h)$. By
Lemma~\ref{ec:lem:lattice}, $\underline v\le\overline v$.

\emph{Step 2: clipping.} Let $v$ be admissible with $\J_P(v)>-\infty$. Since
$P\le\overline mP^0$, inequality \eqref{ec:eq:lattice} and optimality of
$\underline v$ at $\overline mP^0$ give
\[
 \J_P(v\vee \underline v)+\J_{\overline mP^0}(v\wedge \underline v)
 \ge\J_P(v)+\J_{\overline mP^0}(\underline v)
 \ge\J_P(v)+\J_{\overline mP^0}(v\wedge \underline v),
\]
so $\J_P(v\vee \underline v)\ge\J_P(v)$. Since $\underline mP^0\le P$, the same
argument with $\overline v$ at $\underline mP^0$ and the path $w=v\vee \underline v$ gives
$\J_P(w\wedge \overline v)\ge\J_P(w)$. Hence the clipped path
$(v\vee \underline v)\wedge \overline v$ is admissible and has an objective at least
$\J_P(v)$, and all expenditures involved are finite because the prices are
comparable.

\emph{Step 3: existence and uniqueness.} By Step 2 a maximizing sequence
$v_n$ can be chosen with $\underline v\le v_n\le \overline v$, so that $|v_n|$ and
$P\Phi(v_n)$ are dominated in $L^1(e^{-\delta t}\dd t\,\dd\mathbb P)$. For an
admissible path, discounted integration by parts gives
\begin{equation}\label{ec:eq:coercive}
 \J_P(v)=\tilde\alpha v_{0-}
 +\E\int_0^\infty e^{-\delta t}\big((1-\delta\tilde\alpha)v_t-P_t\Phi(v_t)\big)\dd t
 -\bar\alpha\,\E\int_{[0,\infty)}e^{-\delta t}\dd|v|_t,
\end{equation}
where the terminal term vanishes because its absolute expectation is at most
\[
 e^{-\delta T}|v_{0-}|+\E\int_{[0,T]}e^{-\delta s}e^{-\delta(T-s)}\dd|v|_s\longrightarrow0 .
\]
The first two terms of \eqref{ec:eq:coercive} are bounded along the
sequence, so the discounted total variations of the $v_n$ are bounded in
expectation. Apply the scalar Koml\'os theorem \citep{Komlos} simultaneously
to these total variations and to $v_n(t)$ at nonnegative rational times,
placing the countably many coordinates on a disjoint union with positive
weights chosen to make their common $L^1$ bound finite. Ces\`aro averages
$\check v_n$ of a subsequence then converge at rational times and have
pathwise bounded variation on compacts, and they remain a maximizing
sequence by concavity. Define the limit path by right limits of the rational
time limits. Helly's argument gives convergence at continuity points and in
$L^1$ on compact time intervals, and right continuity of the filtration
gives adaptedness. The limit is right continuous with locally bounded
variation, has the reference $v_{0-}$, and lies between $\underline v$ and $\overline v$.
Once every path is extended by $v_{0-}$ before zero, weighted variation,
including the initial atom, is lower semicontinuous. Fatou's lemma and the
domination in \eqref{ec:eq:coercive} show that the limit is an admissible
maximizer, and strict concavity gives uniqueness.

\emph{Step 4: index bounds and optimality conditions.} Write $h^\ast$ for
the optimum. Since $\underline h\le h^\ast\le \overline h$ and the indices of $\underline h$ and $\overline h$ lie in
$[z_{\alpha,0},z_{\beta,0}]$,
\[
 P_t(h^\ast_t)^\gamma\ge\underline mP^0_t(\underline h_t)^\gamma
 =\frac{\underline m}{\overline m}\,\overline mP^0_t(\underline h_t)^\gamma
 \ge\frac{\underline m}{\overline m}z_{\alpha,0},
\]
and the upper bound in \eqref{ec:eq:comparableband} follows in the same way
from $h^\ast\le \overline h$. The values and left limits of $\underline v$ and $\overline v$ are interior on
compacts, so $h^\ast$ is locally interior, and
Theorem~\ref{ec:thm:necessary} gives \eqref{ec:eq:foc}.

\emph{Step 5: stability.} For $a>0$ the substitution $h=a^{-1/\gamma}k$
multiplies running utility, expenditure at the price $aP$, and both
adjustment costs by $a^{(\gamma-1)/\gamma}$, and it maps the reference
$h_{0-}$ to $a^{1/\gamma}h_{0-}$. For logarithmic utility it changes the
objective by the constant $-(\log a)/\delta$. In both cases
\[
 h^\ast(aP,h_{0-})=a^{-1/\gamma}h^\ast(P,a^{1/\gamma}h_{0-}).
\]
Let $a\ge1$. The price comparison of Lemma~\ref{ec:lem:lattice} gives
$h^\ast(aP,h_{0-})\le h^\ast(P,h_{0-})$, and the reference comparison gives
$h^\ast(P,a^{1/\gamma}h_{0-})\ge h^\ast(P,h_{0-})$. Together,
\begin{equation}\label{ec:eq:pricescaling}
 a^{-1/\gamma}h^\ast(P,h_{0-})\le h^\ast(aP,h_{0-})\le h^\ast(P,h_{0-}).
\end{equation}
If $e^{-\varepsilon}P\le Q\le e^{\varepsilon}P$, the price comparison gives
$h^\ast(e^{\varepsilon}P,h_{0-})\le h^\ast(Q,h_{0-})\le h^\ast(e^{-\varepsilon}P,h_{0-})$.
Applying \eqref{ec:eq:pricescaling} with $a=e^{\varepsilon}$, once to $P$
and once to $e^{-\varepsilon}P$, bounds the outer terms by
$e^{-\varepsilon/\gamma}h^\ast(P,h_{0-})$ from below and by
$e^{\varepsilon/\gamma}h^\ast(P,h_{0-})$ from above, which is
\eqref{ec:eq:loglipschitz}.

\emph{Budgets.} Let
$\mathcal W(\lambda)=\E\int e^{-\delta t}P_th^\ast(\lambda P,h_{0-})_t\dd t$.
It is continuous and nonincreasing by \eqref{ec:eq:pricescaling} and
dominated convergence on compact intervals of multipliers. The envelopes at
$\lambda\underline mP^0$ and $\lambda\overline mP^0$ give
$c_-\lambda^{-1/\gamma}\le\mathcal W(\lambda)\le c_+\lambda^{-1/\gamma}$ with
$0<c_-\le c_+<\infty$, because the remaining exponential moment integral is
$\int_0^\infty e^{-K_0(\gamma)t}\dd t$. The function is strictly decreasing.
Equal values at $0<\lambda_1<\lambda_2$, the price ordering, and
right continuity would give a common optimizer $h$. With $v=u(h)$ and
$\Gamma_t=\E[\int_t^\infty e^{-\delta(s-t)}P_sh_s^\gamma\dd s\,|\,\F_t]$,
the two adjoints are $\Psi_j=\delta^{-1}-\lambda_j\Gamma$, and the index
bounds keep $\Gamma$ away from zero. Subtracting the upward support
identities gives $\int e^{-\delta t}\Gamma_t\dd v^+_t=0$, and the downward
identities give the same for $v^-$. Hence $h\equiv h_{0-}$, and
\eqref{ec:eq:comparableband} would confine $P^0_t$ to a deterministic compact
subinterval of $(0,\infty)$, which contradicts its nondegenerate lognormal
law for $t>0$. Therefore $\mathcal W$ takes every value in $(0,\infty)$
exactly once, and Theorem~\ref{ec:thm:foc} gives the unique primal optimum.
\end{proof}

\begin{lemma}\label{ec:lem:localstab}
Fix $(\delta,\alpha,\beta,\gamma,h_{0-})$. Let $P^n$ and $P$ be positive
progressively measurable prices with
$C_0^{-1}e^{\nu_nt-\kappa_nB_t}\le P^n_t\le C_0e^{\nu_nt-\kappa_nB_t}$ and
$C_0^{-1}e^{\nu t-\kappa B_t}\le P_t\le C_0e^{\nu t-\kappa B_t}$ for one
finite $C_0$, where $(\nu_n,\kappa_n)\to(\nu,\kappa)$, $\kappa\ne0$, and the
limiting market satisfies $K(\gamma)>0$. Suppose $P^n\to P$ in
$\dd t\otimes\mathbb P$ measure on every finite time interval. For large $n$
let $h^n=h^\ast(P^n,h_{0-})$ and $h=h^\ast(P,h_{0-})$. Then
$\E\int_0^T|\log h^n_t-\log h_t|\dd t\to0$ for every $T<\infty$, the
maximized Lagrangian values converge, and
$\E\int_0^\infty e^{-\delta t}|P^n_th^n_t-P_th_t|\dd t\to0$.
\end{lemma}

\begin{proof}
The strict inequalities of the limiting market persist near $(\nu,\kappa)$.
Theorem~\ref{ec:thm:comparable} and continuity of the band give a constant
$C_1$ such that $|\log h^n_t+(\nu_nt-\kappa_nB_t)/\gamma|\le C_1$ for all
large $n$, and the same bound holds for $h$. Consequently there exist
$a_0>0$ and $C_2<\infty$ such that, for all large $n,m$ including the limit,
\begin{equation}\label{ec:eq:crossprice}
 \E\big[e^{-\delta t}P^n_th^m_t\big]\le C_2e^{-a_0t}.
\end{equation}
The exponent inside the expectation is bounded above by
$-\delta t+(\nu_n-\nu_m/\gamma)t+(-\kappa_n+\kappa_m/\gamma)B_t+C_3$, and its
expected decay rate depends continuously on the four coefficients and equals
$K(\gamma)>0$ on the diagonal at the limit. On finite intervals the same
envelopes give uniform moments of every order. Local convergence in measure,
uniform integrability, and the exponentially small tails imply
\begin{equation}\label{ec:eq:priceerror}
 \E\int_0^\infty e^{-\delta t}|P^n_t-P_t|(h^n_t+h_t)\dd t\longrightarrow0 .
\end{equation}
Put $v^n=u(h^n)$, $v=u(h)$, and $\Phi=u^{-1}$. These controls are admissible
for both prices, because admissibility does not depend on the price and
cross expenditures are finite by \eqref{ec:eq:crossprice}. Optimality at $P^n$
and $P$ gives $0\le\J_P(v)-\J_P(v^n)\le
\E\int e^{-\delta t}|P^n_t-P_t|(h^n_t+h_t)\dd t=o(1)$. Convexity of the
adjustment cost and optimality of $v$ yield
\[
 \E\int_0^\infty e^{-\delta t}P_t
 \Big[\frac{\Phi(v^n_t)+\Phi(v_t)}2-\Phi\Big(\frac{v^n_t+v_t}2\Big)\Big]\dd t=o(1).
\]
The inverse utility is strictly convex, and on
$\{t\le T,\sup_{s\le T}|B_s|\le M\}$ the envelopes confine $v^n$ and $v$ to a
compact set and keep $P$ away from zero. The display therefore implies local
convergence in measure of $v^n$, hence of $h^n$ and $\log h^n$. Letting
$M\to\infty$ and using the logarithmic envelopes gives local $L^1$
convergence. Local convergence of $P^nh^n$ and \eqref{ec:eq:crossprice} give
convergence of expenditure on the infinite horizon, and
\eqref{ec:eq:priceerror} with the comparison of optimal values gives
convergence of the maximized Lagrangians.
\end{proof}

\subsection{A Certificate for All Deviations}\label{ec:certificate}

\begin{theorem}\label{ec:thm:regret}
Let $\gamma>0$ and $p\in(0,1]$. Let a candidate $h^0\in\mathcal A$ have a
bounded positive index $z^0=P^0\psi_p'(h^0)/u'(h^0)$ for some
price $P^0$, and let its adjoint $\Psi^0$ satisfy \eqref{ec:eq:foc}
with $h^\ast=h^0$. For another price $P$ define
$z^P_t=P_t\psi_p'(h^0_t)/u'(h^0_t)$ and $\mathcal M_t=z^P_t/z^0_t$. If the
candidate has finite expenditure at $P$, every admissible competitor with finite
$\J_P$ satisfies
\begin{equation}\label{ec:eq:regret}
 \J_P(h)-\J_P(h^0)\le\E\int_0^\infty e^{-\delta t}z^0_t(h^0_t)^{1-\gamma}
 G_{\gamma,p}(\mathcal M_t)\dd t,
\end{equation}
where $G_{\gamma,p}(M)=pM+\frac1{\gamma-1}-\big(p+\frac1{\gamma-1}\big)
M^{p(\gamma-1)/(1+p(\gamma-1))}$ for $\gamma\ne1$ and
$G_{1,p}(M)=p(M-1-\log M)$. If $P=\lambda e^{\delta t}\xi\varpi$ and the
candidate's budget cost is $w_0\le X_0$, its utility regret is bounded by the
right side plus $\lambda(X_0-w_0)$.
\end{theorem}

\begin{proof}
Conditions \eqref{ec:eq:foc} for the candidate give
$\mathcal C(v)-\mathcal C(v^0)\ge\E\int e^{-\delta t}(1-z^0_t)(v_t-v^0_t)\dd t$
by the Fubini argument of Theorem~\ref{ec:thm:foc}. Thus
$\J_P(h)-\J_P(h^0)\le\E\int e^{-\delta t}\{z^0_t(v_t-v^0_t)
-P_t(\Phi_p(v_t)-\Phi_p(v^0_t))\}\dd t$. Maximize the term in braces over scalar
consumption at each $(t,\omega)$. Writing that consumption as $sh^0_t$, the
maximizer is $s=\mathcal M_t^{-1/(\gamma+1/p-1)}$, and substitution yields
$z^0_t(h^0_t)^{1-\gamma}G_{\gamma,p}(\mathcal M_t)$, including the
logarithmic limit. Adding the budget multiplier proves the last claim.
\end{proof}

The function $G=G_{\gamma,p}$ satisfies $G(1)=G'(1)=0$ and $G''(1)=p/(1+p(\gamma-1))$, so the
penalty is quadratic for small errors in the log price.

\section{Existence and Finite Populations}\label{ec:existence}

\subsection{Proof of Theorem~\ref{thm:weakexist} and Corollary~\ref{cor:weakgrowth}}\label{ec:proofexist}

\begin{proof}[Proof of Theorem~\ref{thm:weakexist}]
\emph{Step 1: the benchmark for fixed multipliers.} Work on the canonical
product space of the finite initial type, the common Brownian path, and the
private Brownian path, with the usual completed filtrations. For a common
adapted c\`adl\`ag process $y=\log Y$ within bounded distance of
$k+gt+sB_t$, let
\[
 h_i[y]=h^\ast\big(\lambda_ie^{\delta_it}\xi^i_te^{\theta_iy_t},h_{i,0-}\big),\qquad
 (\mathcal T_\mu y)_t=\frac1{1-\Theta_1}\sum_i\pi_i\E\big[\log h_i[y]_t\,\big|\,\F^B_t\big].
\]
The consistency condition \eqref{eq:consistency} is the fixed point equation
$y=\mathcal T_\mu y$, because
$\log c_i=\theta_iy+\log h_i[y]$ and $\sum_i\pi_i\theta_i=\Theta_1$. The best
responses exist by Theorem~\ref{ec:thm:comparable}, and
\eqref{ec:eq:loglipschitz} gives, in the uniform essential path norm,
\begin{equation}\label{ec:eq:contraction}
 \|\mathcal T_\mu y-\mathcal T_\mu\tilde y\|_\infty
 \le\frac{\Theta_2}{1-\Theta_1}\|y-\tilde y\|_\infty=q\,\|y-\tilde y\|_\infty .
\end{equation}
To see that $\mathcal T_\mu$ maps this space into itself, recall that the
envelopes of Theorem~\ref{ec:thm:comparable} bound
$\sup_{t\le T}|\log h_i[y]_t|$ by a constant plus a multiple of
$\sup_{t\le T}|B_t|+\sup_{t\le T}|W_t|$. Future increments of $B$ are
independent of the full filtration at time $t$, so conditioning on the whole
common path agrees at time $t$ with conditioning on $\F^B_t$. Dominated
convergence over the private path then gives right continuity and left
limits of the conditional expectation on compact intervals, so
$\mathcal T_\mu y$ has an adapted c\`adl\`ag version. The space of common
adapted c\`adl\`ag bounded perturbations of a fixed continuous path, modulo
indistinguishability, is complete in this norm, since a Cauchy sequence has
an almost surely uniform limit along a subsequence.

Set $y^0_t=k(\mu)+gt+sB_t$. At this input the responses have log indices
$x^0_i\in[a_i,b_i]$ and
\[
 \log h_i[y^0]_t=\frac{x^0_i(t)-\mu_i-\theta_ik(\mu)-\hat\nu_it
 +\hat\kappa_{a,i}B_t+\kappa_{p,i}W_t}{\gamma_i},
 \qquad \hat\nu_i=\nu_i+\theta_ig .
\]
The definitions of $g$, $s$, and $k$ imply
$(\mathcal T_\mu y^0)_t-y^0_t=(1-\Theta_1)^{-1}\sum_i(\pi_i/\gamma_i)
(\E[x^0_i(t)\,|\,\F^B_t]-m_i)$, whose absolute value is at most $(1-q)R$.
Hence $\mathcal T_\mu$ maps the closed ball of radius $R$ about $y^0$ into
itself and is a contraction there. Its fixed point $y^\mu$ satisfies
\eqref{eq:existbound}, any two fixed points at finite distance coincide by
\eqref{ec:eq:contraction}, and the iterates satisfy
$\|\mathcal T_\mu^jy^0-y^\mu\|_\infty\le q^jR$.

\emph{Step 2: continuity in the multipliers.} Applying
\eqref{ec:eq:loglipschitz} to two multiplier vectors and then
\eqref{ec:eq:contraction} gives
$\|y^\mu-y^{\tilde\mu}\|_\infty\le(1-\Theta_1-\Theta_2)^{-1}
\sum_i\pi_i|\mu_i-\tilde\mu_i|/\gamma_i$. The log relative consumption of type $i$ then moves by
at most $\gamma_i^{-1}(|\mu_i-\tilde\mu_i|+\theta_i\|y^\mu-y^{\tilde\mu}\|_\infty)$,
and its log consumption by at most that amount plus
$\theta_i\|y^\mu-y^{\tilde\mu}\|_\infty$.
The expenditures
\[
 \mathcal E_i(\mu)=\E\int_0^\infty\xi^i_te^{\theta_iy^\mu_t}h_i[y^\mu]_t\,\dd t
\]
are finite for every multiplier vector by Step 3. On a neighborhood of any
multiplier vector the bounds on log consumption dominate the integrands on
the infinite horizon, so the expenditures are continuous.

\emph{Step 3: bounds and budgets.} At $y^\mu$ the price of type $i$ lies
between $e^{\pm\theta_iR}$ times its geometric price. Comparison with the two
reflected ILA policies at the bounding geometric prices and their bounded
indices gives
\begin{equation}\label{ec:eq:wealthenv}
 C_i^-e^{-\mu_i/\gamma_i+\theta_i(1-1/\gamma_i)k(\mu)}\le\mathcal E_i(\mu)
 \le C_i^+e^{-\mu_i/\gamma_i+\theta_i(1-1/\gamma_i)k(\mu)},\qquad
 C_i^\pm=\frac{e^{(m_i\pm d_i/2)/\gamma_i\pm\theta_i(1+1/\gamma_i)R}}
 {\hat K_i(\gamma_i)},
\end{equation}
because the remaining expectation is
$\E\int\hat\xi^i_t\exp(-\hat\nu_it/\gamma_i
+(\hat\kappa_{a,i}B_t+\kappa_{p,i}W_t)/\gamma_i)\dd t=1/\hat K_i(\gamma_i)$.
This proves finiteness, and the constants depend neither on $\mu$ nor on the
initial references. Let $\bar m=\sum_j\pi_jm_j/\gamma_j$ and
\[
 \mathbf A_{ij}=\frac{\ind_{i=j}}{\gamma_i}
 +\frac{\theta_i(1-1/\gamma_i)}{D}\frac{\pi_j}{\gamma_j},\qquad
 (\mathbf A^{-1})_{ij}=\gamma_i\ind_{i=j}-\theta_i(\gamma_i-1)\pi_j,
\]
where the inverse follows from the formula for the inverse of a diagonal
matrix plus a matrix of rank one and
$D+\sum_i\pi_i\theta_i(1-1/\gamma_i)=1$. The continuous functions
$F_i(\mu)=\log\mathcal E_i(\mu)+(\mathbf A\mu)_i$ take values in the fixed
intervals $I_i=[\log C_i^-+\theta_i(1-1/\gamma_i)\bar m/D,\,
\log C_i^++\theta_i(1-1/\gamma_i)\bar m/D]$. The map
$\mu\mapsto\mathbf A^{-1}(F(\mu)-(\log X_i)_i)$ therefore sends the compact
convex parallelepiped $\mathbf A^{-1}(\prod_iI_i-(\log X_i)_i)$ into itself,
and Brouwer's theorem gives $\mu$ with $\mathcal E_i(\mu)=X_i$ for every type.
At these multipliers all budgets bind, the individual plans are primal optima
by Theorem~\ref{ec:thm:foc}, and $y^\mu$ satisfies
\eqref{eq:consistency}. In the market of the representative household the
common risk and its own private risk are spanned, so the budgets finance
the actual consumption. Finally \eqref{ec:eq:comparableband} gives the index
bounds $e^{a_i-2\theta_iR}\le z_i\le e^{b_i+2\theta_iR}$.

The budget also yields a trading strategy. For
$G=\int_0^\infty\xi_sc_s\dd s$ and $\mathcal N_t=\E[G\,|\,\F_t]$ set
$X_t=\xi_t^{-1}(\mathcal N_t-\int_0^t\xi_sc_s\dd s)\ge0$. Brownian martingale
representation gives $\dd\mathcal N=\phi_a\dd B+\phi_p\dd W$, and It\^o's
formula yields $\dd X_t=(rX_t+\kappa_a\vartheta_{a,t}+\kappa_p\vartheta_{p,t}
-c_t)\dd t+\vartheta_{a,t}\dd B_t+\vartheta_{p,t}\dd W_t$ with
$\vartheta_{k,t}=\phi_{k,t}/\xi_t+\kappa_kX_t$ for $k\in\{a,p\}$. The two technologies implement
these loadings, and with common risk only the $B$ term is present.
\end{proof}

\begin{proof}[Proof of Corollary~\ref{cor:weakgrowth}]
Substitute homogeneous coefficients into \eqref{eq:weakcondition} and
\eqref{eq:generalcenter}, divide \eqref{eq:existbound} by $t$, and use
$B_t/t\to0$. Without common noise the conditional averaging in
$\mathcal T_\mu$ produces a deterministic benchmark. For logarithmic utility
$\hat K(1)=\hat r+(\delta-\hat r)=\delta$.
\end{proof}

\subsection{Direct Existence in Finite Games}\label{ec:finite}

With common risk only, put $p_i=1-\theta_i/N$, $\rho_i=1+p_i(\gamma_i-1)>0$,
and define the matrices
$\mathbf U_{ij}=\theta_i\ind_{j\ne i}/(Np_i)$,
$\mathbf E=\mathrm{diag}(1-1/\rho_i)\mathbf U$, and $\mathbf M=I-\mathbf E$.
Let $\bar u_i=\sum_j\mathbf U_{ij}=\theta_i(N-1)/(N-\theta_i)$ and
\begin{equation}\label{ec:eq:finiteweak}
 q_N=\max_i(1+1/\rho_i)\bar u_i<1 .
\end{equation}
Then $\|\mathbf E\|_\infty<1$ and $\mathbf M$ is invertible. With
$\nu_i=\delta_i-r-\kappa^2/2$ define the vectors
$\mathbf g=\mathbf M^{-1}(-\nu_i/\rho_i)_i$ and
$\mathbf s=\mathbf M^{-1}(\kappa/\rho_i)_i$, and for household $i$ the
artificial coefficients
$\tilde\nu_i=\nu_i+(\mathbf U\mathbf g)_i$,
$\tilde\kappa_i=\kappa-(\mathbf U\mathbf s)_i$, and
$\tilde r_i=\delta_i-\tilde\nu_i-\tilde\kappa_i^2/2$.

\begin{theorem}\label{ec:thm:finiteexist}
Assume \eqref{ec:eq:finiteweak}, $\tilde\kappa_i\ne0$, and
$\tilde K_i(\rho_i)>0$ for every household, together with the preference
restrictions of Section~\ref{sec:setup}. Then for every $X_i>0$ and
$h_{i,0-}>0$ there exists a Nash equilibrium of the finite game with common
risk. Neither synchronized initial data nor equal parameters
$(\delta_i,\alpha_i,\beta_i,\gamma_i,\theta_i)$ are required. All consumption
plans can be financed, and all relative plans are admissible. The profile is
also a mean field equilibrium for its empirical type distribution, so
Theorem~\ref{thm:exact} applies to its replicas.
\end{theorem}

For homogeneous parameters the vectors $\mathbf g$ and $\mathbf s$ have
constant components $(r+\kappa^2/2-\delta)/\hg$ and $\kappa/\hg$, and the
coefficient in \eqref{ec:eq:finiteweak} tends to $\theta(1+1/\gamma)$ as
$N\to\infty$. Condition \eqref{ec:eq:finiteweak} and the weighted condition
\eqref{eq:weakcondition} do not imply each other.

\begin{proof}
\emph{Reduction.} Hold the log consumptions $y_j=\log c_j$ of the others
fixed and put $\varpi_i=\exp((\mathbf Uy)_i)$. A deviation has consumption
$c_i=\varpi_i\zeta_i$ and relative consumption $h_i=\zeta_i^{p_i}$. For every
$\gamma_i>0$, $u_{\gamma_i}(\zeta_i^{p_i})=p_iu_{\rho_i}(\zeta_i)$, and the
same identity holds for both adjustment costs, with reference
$\zeta_{i,0-}=h_{i,0-}^{1/p_i}$. Dividing the Lagrangian by $p_i$ gives the
problem of Theorem~\ref{ec:thm:comparable} with risk aversion $\rho_i$ and
price $(\lambda_i/p_i)e^{\delta_it}\xi_t\varpi_i$. If $\zeta_i[y]$ is its
optimum, the best response map is
$(\mathcal T_\mu y)_i=(\mathbf Uy)_i+\log\zeta_i[y]$ with $\mu_i=\log\lambda_i$,
and \eqref{ec:eq:loglipschitz} gives
$\|\mathcal T_\mu y-\mathcal T_\mu\tilde y\|_\infty\le q_N\|y-\tilde y\|_\infty$.

\emph{Fixed multipliers.} Let $[e^{a_i},e^{b_i}]$ be the band of household
$i$ in its artificial market, $d_i=b_i-a_i$, $m_i=(a_i+b_i)/2$,
$\mathbf k=\mathbf M^{-1}\mathrm{diag}(1/\rho_i)(m_i+\log p_i-\mu_i)_i$, and
$R_N=(1-q_N)^{-1}\max_id_i/(2\rho_i)$. At the input
$y^0_i(t)=\mathbf k_i+\mathbf g_it+\mathbf s_iB_t$ the transformed plans have
log indices $x^0_i\in[a_i,b_i]$, and the definitions give
$(\mathcal T_\mu y^0)_i-y^0_i=(x^0_i-m_i)/\rho_i$. The closed ball of radius
$R_N$ about $y^0$ is invariant, it is complete as in the proof of
Theorem~\ref{thm:weakexist}, and Banach's theorem gives a unique profile
$y^\mu$ in this class, continuous in $\mu$.

\emph{Budgets.} Use the centering vector $\mathbf k$ as the variable, with
$\mu_i=m_i+\log p_i-\rho_i(\mathbf M\mathbf k)_i$, which is an invertible
affine change of coordinates, and let
$\mathcal E_i(\mathbf k)=\E\int\xi_te^{y^\mu_i(t)}\dd t$. The price of
household $i$ differs from its geometric center by a log amount at most
$\bar u_iR_N$, and the envelopes give
\[
 \mathbf k_i+J_i^-\le\log\mathcal E_i(\mathbf k)\le\mathbf k_i+J_i^+,\qquad
 J_i^\pm=-\log\tilde K_i(\rho_i)\pm\big(d_i/(2\rho_i)+\bar u_i(1+1/\rho_i)R_N\big),
\]
since the log scale of consumption in this expectation is
$-(\mu_i-\log p_i)/\rho_i+(1-1/\rho_i)(\mathbf U\mathbf k)_i=\mathbf k_i-m_i/\rho_i$
and the remaining integral equals $1/\tilde K_i(\rho_i)$. The map
$\mathbf k\mapsto(\log X_i-\log\mathcal E_i(\mathbf k)+\mathbf k_i)_i$ is
continuous and sends the box $\prod_i[\log X_i-J_i^+,\log X_i-J_i^-]$ into
itself, so Brouwer's theorem gives $\mathcal E_i=X_i$ for all households. Each
transformed optimum is then a best response that exhausts the budget, and the
utility identity transfers admissibility and optimality to the original game.
The index of household $i$ is
$z_i=(\lambda_i/p_i)e^{\delta_it}\xi_tY_t^{\theta_i}h_i^{\gamma_i}$, and the
same bounds and support conditions show that the plan is also optimal for a
price taker with multiplier $\lambda_i/p_i$. Budgets and aggregation are unchanged, so the
profile is a mean field equilibrium for the empirical type distribution.
\end{proof}

\subsection{Proof of Theorem~\ref{thm:populationlimit}}\label{ec:prooflimit}

\begin{proof}
The Nash property follows from Theorem~\ref{thm:exact}, so it remains to
prove the convergence assertions. All coefficients in \eqref{eq:generalcenter} and
\eqref{eq:existconstants} and the artificial bands depend continuously on
the weights near $\pi$. Because the hypotheses are strict inequalities,
there is a compact neighborhood of $\pi$ on which the contraction factors are
at most $q_\ast<1$, the artificial volatilities are bounded away from zero,
the cost rates are bounded below by a positive constant, and the radii are
bounded by $R_\ast$.

Let $w^n\to\pi$ and $\mu^n\to\mu$, write $\mathcal T_{w,\mu}$ for the map of
Step 1 in the proof of Theorem~\ref{thm:weakexist}, and start its iteration
from $y^{n,0}_t=k(w^n,\mu^n)+g(w^n)t+s(w^n)B_t$. Let $y^m$ be the iterates
at $(\pi,\mu)$. For every fixed iteration
number $m$, induction and Lemma~\ref{ec:lem:localstab} show
$\E\int_0^T|y^{n,m}_t-y^m_t|\dd t\to0$ for $T<\infty$, since each iterate
stays in a uniformly bounded geometric tube, the prices converge locally in
measure, and the individual log responses and their finite weighted sum
converge locally in $L^1$. The contraction estimate gives
$\|y^{n,m}-y^{w^n,\mu^n}\|_\infty\le R_\ast q_\ast^m$ uniformly in $n$, and
the same estimate holds at the limit. Letting $n\to\infty$ and then
$m\to\infty$ proves local $L^1$ continuity of the benchmark jointly in the
weights and the multipliers, and another application of the lemma proves continuity of log
consumption. The bound \eqref{ec:eq:crossprice} gives uniform tails of
expenditure on the infinite horizon. Therefore
\begin{equation}\label{ec:eq:jointcont}
 \mathcal E_i(w^n,\mu^n)\to\mathcal E_i(\pi,\mu),\qquad
 \E\int_0^\infty\xi_t|c_i^{w^n,\mu^n}(t)-c_i^{\pi,\mu}(t)|\dd t\to0,
\end{equation}
where the second limit uses local convergence in measure, local uniform
integrability, and the uniform tails of $\xi c_i$, which also follow from
$\lambda_ie^{\delta_it}\xi_tc_i=P_ih_i$ and the bounds on the convergent
multipliers.

The multipliers selected in the proof of Theorem~\ref{thm:weakexist} lie in
the parallelepipeds $\mathbf A(w)^{-1}(\prod_iI_i(w)-(\log X_i)_i)$, whose
matrices, inverses, and endpoints vary continuously in $w$ near $\pi$. They
therefore lie in one compact set, and a subsequence satisfies
$\mu^{(N)}\to\mu$. Since $\mathcal E_i(\pi^{(N)},\mu^{(N)})=X_i$, joint
continuity gives $\mathcal E_i(\pi,\mu)=X_i$ for every type, so the profile
with multipliers $\mu$ is a mean field equilibrium with weights $\pi$, and
the continuity already proved gives the assertions on consumption and the
benchmark. The individual Lagrangian is $\J_i=U_i-\lambda_i\mathcal E_i$, its
optimal values converge by Lemma~\ref{ec:lem:localstab}, and
$\mathcal E_i=X_i$ with convergent multipliers gives convergence of
utilities.

For wealth let $G_i^{(N)}=\int_0^\infty\xi_sc_i^{(N)}(s)\dd s$ and
$G_i=\int_0^\infty\xi_sc_i(s)\dd s$. Expenditure convergence gives
$G_i^{(N)}\to G_i$ in $L^1$, and the deflated wealth is
$\xi_tX_i^{(N)}(t)=\E[G_i^{(N)}\,|\,\F_t]-\int_0^t\xi_sc_i^{(N)}(s)\dd s$.
Doob's weak type maximal inequality gives uniform convergence in probability
of deflated wealth on compacts, and multiplication by the locally bounded
continuous process $1/\xi$ gives the claim. Every subsequence has a further
convergent subsequence, so uniqueness of the limiting profile yields
convergence of the whole sequence, whether or not its multipliers are
unique.
\end{proof}

\subsection{Proof of Theorem~\ref{thm:certificate}}\label{ec:proofcertificate}

\begin{proof}
Fix a household of type $a$ and write $c=c^a$, $h^0=h^{a,0}$,
$\theta=\theta_a$, $\gamma=\gamma_a$, $\lambda=\lambda_a$, and $X=X_a$. Let
$\tilde c$ be a plan whose relative consumption in the game is admissible
and that the household can finance, so that
$\E\int\xi_t\tilde c_t\dd t\le X$. As in the proof of Theorem~\ref{thm:exact}, with
$p=1-\theta/N$ the deviation produces the benchmark
$Y^N(\tilde c/c)^{1/N}$ and the relative consumption
$\tilde h=h^{a,N}(\tilde c/c)^p$, whose reference is $h_{a,0-}$. With
$C_N(h)=\E\int\xi_t(Y^N_t)^\theta h_t\dd t$, H\"older's inequality gives
\[
 C_N(\tilde h)=\E\int_0^\infty\xi_tc_t^{1-p}\tilde c_t^{\,p}\dd t
 \le w_a^{1-p}X^p\le X .
\]
Let $P^0_t=\lambda e^{\delta t}\xi_t(Y^0_t)^\theta$ and
$P_t=\lambda e^{\delta t}\xi_t(Y^N_t)^\theta=e^{\theta\Delta_t}P^0_t$, and
let $\J_N$ be the Lagrangian \eqref{eq:lagrangian} with $p=1$ and price
$P$, so that $U(h)=\J_N(h)+\lambda C_N(h)$. The candidate $h^0$ satisfies
the optimality conditions for $P^0$, its index is bounded, and its
expenditure at the price $P$ is
$C_N(h^0)=\E\int\xi_tc_te^{\theta\Delta_t}\dd t<\infty$.
Theorem~\ref{ec:thm:regret} with $p=1$ and $\mathcal M=e^{\theta\Delta}$
therefore applies, and since
$e^{-\delta t}z^0_t(h^0_t)^{1-\gamma}=\lambda\xi_t(Y^0_t)^\theta h^0_t=\lambda\xi_tc_t$
it gives
\[
 \J_N(\tilde h)-\J_N(h^0)\le
 \lambda\E\int_0^\infty\xi_tc_tG_\gamma\big(e^{\theta\Delta_t}\big)\dd t
\]
whenever $\J_N(\tilde h)>-\infty$, and trivially otherwise. Hence
\[
 U(\tilde h)\le\J_N(h^0)+\lambda\E\int_0^\infty\xi_tc_tG_\gamma\big(e^{\theta\Delta_t}\big)\dd t+\lambda X
 =U(h^0)+\lambda\E\int_0^\infty\xi_tc_tG_\gamma\big(e^{\theta\Delta_t}\big)\dd t
 +\lambda\big(X-C_N(h^0)\big),
\]
and subtracting the utility $U(h^{a,N})$ that the household obtains in the
profile gives the claim.
\end{proof}

\subsection{Proof of Corollary~\ref{cor:rounded}}\label{ec:proofrounded}

\begin{proof}
Use the equilibrium multipliers $\lambda_a$ and bounds
$0<\underline z_a\le z^a_t\le\overline z_a<\infty$ for
$z^a_t=\lambda_ae^{\delta t}\xi_tY_t^\theta h^a_t$, which follow from
Theorem~\ref{thm:weakexist}. Put $e_a=N_a/N-\pi_a$. Since utility is
logarithmic and $\delta$ and $\theta$ are common,
$\log h^a_t-\log h^1_t=\log z^a_t-\log z^1_t-\log\lambda_a+\log\lambda_1$, so
there is a deterministic $M_\Delta<\infty$ with $|\log h^a_t-\log h^1_t|\le
M_\Delta$ for all $a$ and $t$. Define
\[
 \bar v_a=\E\int_{[0,\infty)}e^{-\delta t}\dd|\log h^a|_t,\qquad
 \mathcal V_N=\sum_a|e_a|\big(\bar v_a+|\log(h_{a,0-}/h_{1,0-})|\big),\qquad
 \eta_N=\theta M_\Delta d_N,
\]
where each $\bar v_a$ is finite by admissibility, and let
$\Lambda_a=\max\{|\alpha_a|,\beta_a\}$. We prove the certificate
\begin{equation}\label{ec:eq:roundedcertificate}
 \varepsilon_N=\max_a\Big\{\frac{\overline z_a}{\delta}
 \big(e^{\eta_N}-1-\eta_N\big)+\frac{\eta_N}{\delta}+\theta\Lambda_a\mathcal V_N
 +\lambda_aX_a\big(1-e^{-\eta_N}\big)\Big\}.
\end{equation}
The assigned plans keep their budgets $\E\int\xi_tc^a_t\dd t=X_a$ and can be
financed. Let $Y^N_t=\prod_a(c^a_t)^{N_a/N}$ and $\Delta_t=\log(Y^N_t/Y_t)$.
Since $\sum_ae_a=0$ and $c^a=h^aY^\theta$,
$\Delta_t=\sum_ae_a(\log h^a_t-\log h^1_t)$ and $|\Delta_t|\le M_\Delta d_N$.

The relative consumption of a household of type $a$ in the finite profile is
$h^{a,N}_t=c^a_t(Y^N_t)^{-\theta}=h^a_te^{-\theta\Delta_t}$, with prescribed
value $h_{a,0-}$ at $0-$, so $\Delta$ is extended by $\Delta_{0-}=0$. If
$\sum_ae_a\log h^a$ is extended with the original references instead, its
value at $0-$ is $\sum_ae_a\log h_{a,0-}$, and correcting this one initial
atom gives
\begin{equation}\label{ec:eq:roundedvariation}
 \E\int_{[0,\infty)}e^{-\delta t}\dd|\Delta|_t\le\sum_a|e_a|\bar v_a
 +\Big|\sum_ae_a\log(h_{a,0-}/h_{1,0-})\Big|\le\mathcal V_N .
\end{equation}
Thus $\log h^{a,N}=\log h^a-\theta\Delta$ is admissible, including its
initial revision. For a signed measure $\dd v$ let
$\mathcal C_a(v)=\E\int e^{-\delta t}(\alpha_a\dd v^++\beta_a\dd v^-)$. Its
integrand is Lipschitz with constant $\Lambda_a$, so with a common dominating measure
$|\mathcal C_a(v)-\mathcal C_a(w)|\le\Lambda_a\E\int e^{-\delta t}\dd|v-w|_t$,
even when $\alpha_a<0$. Applied to $\log h^a$ and $\log h^a-\theta\Delta$ with
\eqref{ec:eq:roundedvariation}, this gives
$|U_a(h^{a,N})-U_a(h^a)|\le\eta_N/\delta+\theta\Lambda_a\mathcal V_N$.

The hypotheses of Theorem~\ref{thm:certificate} hold with $Y^0=Y$, the
mean field benchmark, because $h^{a,0}=h^a$ is the mean field plan, which
satisfies the optimality conditions with the bounded index $z^a$. All
expenditures are finite because $e^{-\eta_N}\le(Y^N/Y)^\theta\le e^{\eta_N}$.
We bound the three terms of $\varepsilon_a$. For logarithmic utility
$\lambda_a\xi_tc^a_t=e^{-\delta t}z^a_t\le e^{-\delta t}\overline z_a$ and
$G_1(e^x)=e^x-1-x$, so with $|\theta\Delta_t|\le\eta_N$ the first term is at
most $\overline z_a(e^{\eta_N}-1-\eta_N)/\delta$. The second term is
at most $\eta_N/\delta+\theta\Lambda_a\mathcal V_N$ by the estimate above.
The third term is at most $\lambda_aX_a(1-e^{-\eta_N})$, because
$\E\int\xi_tc^a_te^{\theta\Delta_t}\dd t\ge e^{-\eta_N}X_a$. Summing gives
\eqref{ec:eq:roundedcertificate}. Finally
$\mathcal V_N\le d_N\max_a(\bar v_a+|\log(h_{a,0-}/h_{1,0-})|)$, $\eta_N=O(d_N)$,
$e^{\eta_N}-1-\eta_N=O(d_N^2)$, and $1-e^{-\eta_N}=O(d_N)$. If
$|N_a/N-\pi_a|\le1/N$ for every $a$, then $d_N\le n/N$.
\end{proof}

\subsection{Unspanned Private Risk}\label{ec:private}

\begin{proposition}\label{ec:prop:private}
Let $N\ge2$ and $\theta\in(0,1)$, and suppose household $i$ can trade the
bank account and its own private technology, driven by $W^i$, but cannot
trade risks driven by $W^j$, $j\ne i$, where the Brownian motions are
independent. Financing requires nonnegative wealth that funds consumption
through trading. For homogeneous parameters let the proposed relative
consumption be $h^j_t=h^j_0e^{\ell^j_t/\gamma}$ with the reflected
regulators $\ell^j$ of the mean field model,
$x^j_t=x^j_0+\hat\nu t-\kappa W^j_t+\ell^j_t\in[a,b]$, and $\kappa>0$.
Impose the finite consistency $c^i=h^i(Y^N)^\theta$ with
$Y^N=(\prod_jc^j)^{1/N}$. Then the resulting consumption of every household
cannot be financed from any finite initial wealth, and the same conclusion
holds under the usual supermartingale admissibility condition for every
equivalent local martingale measure. If the coefficients are those of the stationary mean field
equilibrium with $\hat K(\gamma)>0$, the cost under the deflator of the
mean field model is nevertheless finite for all large $N$. Multiplying the
profile by positive constants does not remove the obstruction.
\end{proposition}

\begin{proof}
The consistency equations give $Y^N=(\prod_jh^j)^{1/[N(1-\theta)]}$. Put
$q=1/\gamma$, $\eta=\theta/[N\gamma(1-\theta)]>0$, and $\bar q=q+N\eta$. Then
$c^i_t=c^i_0\exp\big((q+\eta)\ell^i_t+\eta\sum_{j\ne i}\ell^j_t\big)$. Since
$x^j_t-x^j_0\ge-(b-a)$,
\begin{equation}\label{ec:eq:privatebound}
 c^i_t\ge c^i_0e^{-(b-a)\bar q}\exp\Big(-\hat\nu\bar qt
 +\kappa\Big((q+\eta)W^i_t+\eta\sum_{j\ne i}W^j_t\Big)\Big).
\end{equation}
Fix $0<t_0<T<\infty$ and one $j\ne i$. On $[0,T]$ let $\mathbb Q_m$ be the
equivalent measure under which $W^i$ has drift $-\kappa$, $W^j$ has drift
$m>0$, and the other private motions keep zero drift. Girsanov densities with
constant drift are strictly positive martingales on finite horizons, and
all traded discounted assets of household $i$ are local martingales under
every $\mathbb Q_m$, since the drift of $W^j$ prices an untraded risk.
Taking expectations in \eqref{ec:eq:privatebound} and using independence and
Gaussian exponential moments,
\[
 \E^{\mathbb Q_m}[e^{-rt}c^i_t]\ge c^i_0e^{-(b-a)\bar q}
 \exp\Big[\Big(-r-\hat\nu\bar q-\kappa^2(q+\eta)
 +\tfrac12\kappa^2\{(q+\eta)^2+(N-1)\eta^2\}+\kappa\eta m\Big)t\Big],
\]
whose integral over $[t_0,T]$ tends to infinity as $m\to\infty$. A
nonnegative wealth process that funds $c^i$ would make discounted wealth plus
accumulated discounted consumption a nonnegative local martingale under
$\mathbb Q_m$, hence a supermartingale, which would require
$X^i_0\ge\E^{\mathbb Q_m}\int_0^Te^{-rt}c^i_t\dd t$ for every $m>0$. This is
impossible for finite $X^i_0$.

For the finiteness claim, exponential tilting of the independent Brownian
motions bounds the expected cost under the mean field deflator above and
below by constants times $\int_0^\infty e^{-K_Nt}\dd t$ with
$K_N=r+\varsigma+\hat\nu\bar q-\varsigma\{(q+\eta-1)^2+(N-1)\eta^2\}$. At
the stationary growth rate, $\hat\nu=\nu\gamma(1-\theta)/\hg$, so
$K_N\to\hat K(\gamma)>0$. A positive constant only changes the prefactor in
\eqref{ec:eq:privatebound}.
\end{proof}

\subsection{Types without Intertemporal Loss Aversion}\label{ec:frictionless}

If $L=1$, that is $\beta=-\alpha\ge0$, the cost terms
$-\alpha\,\dd v_t^+-\beta\,\dd v_t^-$ in \eqref{eq:U} combine to
$-\alpha\,\dd v_t$. For an admissible path,
discounted integration by parts as in \eqref{ec:eq:coercive}, now with
$\bar\alpha=0$ and $\tilde\alpha=\alpha$, gives
\begin{equation}\label{ec:eq:Ufrictionless}
 U(h)=\alpha\,u(h_{0-})+(1-\delta\alpha)\,\E\int_0^\infty e^{-\delta t}u(h_t)\,\dd t .
\end{equation}
For a type with $L=1$ we take the right side as the definition of utility on
the class $\mathcal A_1\supset\mathcal A$ of positive progressively
measurable processes with $\E\int_0^\infty e^{-\delta t}|u(h_t)|\dd t<\infty$.
Households of such a type choose from $\mathcal A_1$ in the mean field
problem and in the game, and their Lagrangian $\J_P$ is
\eqref{ec:eq:Ufrictionless} minus the expenditure term of
\eqref{eq:lagrangian}. For these types put $a_i=b_i=\log(1-\delta_i\alpha_i)$,
so that $d_i=0$ and $m_i=a_i$. For $L>1$ the band of
Lemma~\ref{lem:band}(i) contains the two points $1-\delta\alpha<1+\delta\beta$,
and these coincide when $\beta=-\alpha$.

\begin{proposition}\label{ec:prop:frictionless}
Let every type satisfy the parameter restrictions of Section~\ref{sec:setup},
except that $\alpha_i=-\beta_i$, and hence $L_i=1$, is also allowed.
\begin{enumerate}[label=(\roman*),nosep,leftmargin=2em]
\item Under the hypotheses of Theorem~\ref{ec:thm:comparable}, a type with
$L=1$ has the optimum $h^\ast_t=((1-\delta\alpha)/P_t)^{1/\gamma}$ in
$\mathcal A_1$, unique up to $\dd t\otimes\mathbb P$ null sets. Its index
$P_t(h^\ast_t)^\gamma$ equals $1-\delta\alpha$, inequality
\eqref{ec:eq:loglipschitz} holds, and every budget binds at exactly one
multiplier.
\item Theorems~\ref{thm:weakexist}, \ref{thm:exact},
and~\ref{thm:populationlimit} hold with these $a_i$ and $b_i$ in
\eqref{eq:existconstants}, where households of types with $L_i=1$ choose
from $\mathcal A_1$.
\item Theorem~\ref{thm:certificate} holds when, for a type with $L_a=1$,
admissibility refers to $\mathcal A_1$ and the optimality conditions of
Theorem~\ref{ec:thm:foc} are replaced by the condition that the index of
$h^{a,0}$ for the price $\lambda_ae^{\delta_at}\xi_t(Y^0_t)^{\theta_a}$ equals
$1-\delta_a\alpha_a$.
\end{enumerate}
\end{proposition}

\begin{proof}
(i) By \eqref{ec:eq:Ufrictionless},
$\J_P(h)=\alpha u(h_{0-})+\E\int_0^\infty e^{-\delta t}\big((1-\delta\alpha)u(h_t)
-P_th_t\big)\dd t$. The integrand is strictly concave in $h_t$ and is
maximized at $h^\ast_t$, so $h^\ast$ is the unique maximizer once
$\J_P(h^\ast)$ is finite. Both $P_th^\ast_t$ and $|u(h^\ast_t)|$ are bounded
by constants times $(P^0_t)^{1-1/\gamma}$, or by an affine function of
$|\log P^0_t|$ when $\gamma=1$, because $\underline mP^0\le P\le\overline mP^0$.
These bounds are integrable, since
$\E\int_0^\infty e^{-\delta t}(P^0_t)^{1-1/\gamma}\dd t=\int_0^\infty
e^{-K_0(\gamma)t}\dd t<\infty$ and
$\E\int_0^\infty e^{-\delta t}|\log P^0_t|\dd t<\infty$. The index identity follows from the first order
condition, and $\log h^\ast(Q)-\log h^\ast(P)=-\gamma^{-1}\log(Q/P)$ gives
\eqref{ec:eq:loglipschitz}. At the price $\lambda P$ the expenditure is
$\lambda^{-1/\gamma}$ times its value at $P$, so each budget binds at exactly
one multiplier, and the Lagrangian argument at the end of the proof of
Theorem~\ref{ec:thm:foc} gives primal optimality.

(ii) The proof of Theorem~\ref{thm:weakexist} in
Section~\ref{ec:proofexist} uses, for each type, existence and uniqueness of
best responses at comparable prices, the estimate
\eqref{ec:eq:loglipschitz}, indices in $[e^{a_i},e^{b_i}]$ at the input $y^0$
and the bounds \eqref{ec:eq:comparableband}, the envelopes
\eqref{ec:eq:wealthenv}, and primal optimality at binding budgets. By (i),
types with $L_i=1$ have all of these properties, where we use the version
$((1-\delta_i\alpha_i)/P)^{1/\gamma_i}$ of their optimum, which is
c\`adl\`ag because $y$ is. Their index is identically
$e^{a_i}$, so they contribute nothing to $\mathcal T_\mu y^0-y^0$, and their
expenditure is
$(1-\delta_i\alpha_i)^{1/\gamma_i}\lambda_i^{-1/\gamma_i}\E\int_0^\infty
e^{-\delta_it/\gamma_i}(\xi^i_t)^{1-1/\gamma_i}e^{\theta_i(1-1/\gamma_i)y_t}\dd t$,
which satisfies \eqref{ec:eq:wealthenv} with $d_i=0$ when
$|y-y^0|\le R$. The proof of Theorem~\ref{thm:exact} compares budgets and
uses optimality in the mean field problem only (Remark~\ref{rem:scope}). The
relative consumption of a deviation in the game is the process that a price
taker facing $Y$ evaluates, so the admissible class is the same in both
problems. In the proof of Theorem~\ref{thm:populationlimit}, the responses
$((1-\delta_i\alpha_i)/P^n)^{1/\gamma_i}$ of a type with $L_i=1$ converge
locally in measure when the prices do and are dominated by the envelopes of
(i), so the conclusions of Lemma~\ref{ec:lem:localstab} hold for them by
dominated convergence and \eqref{ec:eq:crossprice}.

(iii) For a type with $L=1$, \eqref{ec:eq:Ufrictionless} gives
$\J_P(h)-\J_P(h^0)=\E\int_0^\infty e^{-\delta t}\{(1-\delta\alpha)(v_t-v^0_t)
-P_t(\Phi_p(v_t)-\Phi_p(v^0_t))\}\dd t$, and $1-\delta\alpha=z^0_t$ by the
stated condition. The first inequality in the proof of
Theorem~\ref{ec:thm:regret} therefore holds with equality, and the rest of
that proof and the proof of Theorem~\ref{thm:certificate} apply unchanged.
\end{proof}

Corollary~\ref{cor:rounded} is not covered by
Proposition~\ref{ec:prop:frictionless}. Its proof needs the relative
consumption $h^ae^{-\theta\Delta}$ of every type with $L_a>1$ in the rounded
profile to have finite variation. Under its hypotheses the log relative
consumptions of the types with $L=1$ differ only by constants, so $\Delta$
has finite variation exactly when rounding leaves the total weight of these
types unchanged. When $\theta>0$ and both kinds of types are present,
changing that total weight gives the relative consumption of the types
with $L_a>1$ a nonzero Brownian component, so it is not admissible. If
$\theta=0$, the benchmark drops out of every objective, and the assigned
plans remain an exact Nash equilibrium for any composition. If all types
have $L=1$, finite variation is not required.

\section{Computation}\label{ec:compute}

\subsection{Index Lattice for Asynchronous Equilibria}\label{ec:lattice}

This subsection describes the scheme used in
Section~\ref{sec:computeexample}. All risk is common, all types share the
discount rate $\delta$, and $\nu=\delta-r-\kappa^2/2$. The state is the
vector $x=(x^1,\ldots,x^n)$ of log indices
$x^a=\log(\lambda_ae^{\delta t}\xi_tY_t^{\theta_a}(h^a_t)^{\gamma_a})$.
Between revisions $\dd x^a=\nu\,\dd t-\kappa\,\dd B_t$ for every $a$, and a
revision that changes $\log h^b$ by $g_b$ moves $x$ by $g_bV_{\cdot b}$ with
$V_{ab}=\gamma_a\ind_{a=b}+\theta_a\pi_b/(1-\Theta_1)$. By
Theorem~\ref{ec:thm:foc}, the plans are best responses to the benchmark they
generate when the adjoint
$\Psi^a_t=\E[\int_t^\infty e^{-\delta(s-t)}(1-e^{x^a_s})\dd s\,|\,\F_t]$ of every
type lies in $[-\beta_a,\alpha_a]$, equals $\alpha_a$ when $h^a$ increases,
and equals $-\beta_a$ when $h^a$ decreases.

Fix a time step $\Delta t$ and the lattice spacing $\Delta x=\kappa\sqrt{\Delta t}$.
Given the state $y$ after the revisions of a period, all coordinates move
by $+\Delta x$ with probability $p_+=\frac12(1+\nu\sqrt{\Delta t}/\kappa)$ and by
$-\Delta x$ otherwise, which matches the drift and the variance of the log index
to first order. Let $\mathbf 1=(1,\ldots,1)$ and define the continuation adjoint
\[
 C^a(y)=\big(1-e^{y^a}\big)\Delta t
 +e^{-\delta\Delta t}\big(p_+\Psi^a(y+\Delta x\mathbf 1)+(1-p_+)\Psi^a(y-\Delta x\mathbf 1)\big),
\]
where $\Psi^a$ is the adjoint at states before revisions, stored on lattice
points and evaluated elsewhere by multilinear interpolation (trilinear for $n=3$). At a lattice
point $x$ the revisions $g\in\R^n$ solve the complementarity conditions
\[
 C^b(x+Vg)\le\alpha_b\ \text{with equality if }g_b>0,\qquad
 C^b(x+Vg)\ge-\beta_b\ \text{with equality if }g_b<0,
\]
and $\Psi^a(x)=C^a(x+Vg)$. A type that revises upward at $x$ therefore has
$\Psi^a(x)=\alpha_a$, and a type that revises downward has $\Psi^a(x)=-\beta_a$, so each
adjoint solves a discrete double obstacle problem in which the other types
enter through their revisions. The conditions are solved by Gauss--Seidel
iteration over types, with a one dimensional root search along the
direction $V_{\cdot b}$ for each type $b$. We build the continuation
adjoint from the adjoints before revisions because it then crosses the
obstacle transversally along a revision direction. The adjoint after
revisions touches the obstacle with zero slope, which makes the root search
ill conditioned.

Howard's algorithm alternates two steps. Given the revisions at all lattice
points, the adjoints solve a sparse linear system in which rows of revising
types are replaced by the obstacle values. Given the adjoints,
the revisions are recomputed from the complementarity conditions. We solve
the linear systems with a sparse direct solver, except at
$\Delta t=1/128$, where its fill exceeds the available memory. There we use
GMRES with an incomplete LU preconditioner and require relative residuals
below $10^{-10}$. At $\Delta t=1/32$ the multipliers, initial ratios, and risky
shares from the two solvers differ by less than $2\times10^{-14}$. In the
example of Section~\ref{sec:computeexample} the groups have equal weights
and identical preferences, so $\Psi^a(x)$ depends only on $x^a$ and on the
other coordinates as a set. We store one function for each rank on the
sorted wedge $x^{(1)}\le x^{(2)}\le x^{(3)}$, with the dispersion
$x^{(3)}-x^{(1)}$ capped at $1.6$, and evaluate every labeled point by
sorting the corners of its interpolation cell. In the computed equilibrium
only the groups with the lowest or the highest index revise, and their
revisions reduce the dispersion, so along simulated paths it never exceeds
its initial value, $1.19$ after the initial revision and $1.37$ before it,
which is below the cap. The Markov chain implied by
the interpolation in the linear systems can reach rows whose interpolation
cell crosses the cap, where corners are deleted, but from the initial state
its discounted probability of doing so is below $8\times10^{-12}$
(\texttt{data/saved\_lattice\_audit.json}). Without this symmetry each type needs its own
adjoint on the full set of states with bounded dispersion, which for three
types has about six times as many points as the wedge.
Section~\ref{ec:twotypes} computes two examples with two types of different
preferences on the full lattice.

Given the revisions, the cost of consumption per unit of current
consumption solves the linear system
\[
 F^a(x)=e^{\chi_a(g)}\Big(\Delta t+e^{-\delta\Delta t}\big(p_+e^{\Delta x}F^a(y+\Delta x\mathbf 1)
 +(1-p_+)e^{-\Delta x}F^a(y-\Delta x\mathbf 1)\big)\Big),\qquad y=x+Vg,
\]
where $\chi_a(g)=g_a+\sum_b\theta_a\pi_bg_b/(1-\Theta_1)$ is the change in
log consumption at the revision, and the factors $e^{\pm\Delta x}$ come from
$\xi_{t+\Delta t}/\xi_t=e^{-\delta\Delta t}e^{\pm\Delta x}$ on the lattice. The
wealth of type $a$ is its consumption times the cost after revisions, and its
risky share is $-(\kappa/\sigma)$ times the derivative of the log cost along
$\mathbf 1$.

The revision rules do not involve the multipliers. Given $\mu=\log\lambda$,
the initial log index before revisions is
$x^a_{0-}=\mu_a+\gamma_a\log c^a_{0-}-\theta_a(\gamma_a-1)\log\bar c_{0-}$
with $\log\bar c_{0-}=\sum_b\pi_b\log c^b_{0-}$. The revisions at $x_{0-}$
give the initial consumption $c^a_0$ and the state $x_0$ after them, and the
budget equations, which equate $c^a_0$ times the cost at $x_0$ to $X_a$, are
$n$ equations in $\mu$. We solve them with a hybrid Powell method
followed by Newton steps with a finite difference Jacobian.

\subsection{Diagnostics}\label{ec:diagnostics}

Table~\ref{tab:async} and Section~\ref{sec:computeexample} report the
following checks. The Howard iteration stops when
the largest change in the adjoints is below $10^{-10}$. The complementarity
residual is the largest violation of the conditions above at lattice points
whose coordinates all lie in $(\log\za,\log\zb)$. On the diagonal of the lattice the groups are
synchronized, and Theorem~\ref{thm:common} gives the adjoint $H$, the cost
$\Xf_{\hg}$, and the risky share in closed form, which we compare with the
lattice solution. The budget residual is that of the outer equations, and
an independent Monte Carlo estimate of the cost of consumption, with
$20{,}000$ paths over $200$ years, provides a separate budget diagnostic.
The simulations apply the interpolated revision rule at the state after
each binomial shock, which need not be a lattice point, instead of moving
to a lattice point at random as the chain behind the linear systems does.
The Monte Carlo discrepancy therefore contains the difference between the
two processes as well as sampling and truncation errors, and it is not an
unbiased estimate of the residual of the linear systems.
Along simulated paths we also verify the bound \eqref{eq:existbound} of
Theorem~\ref{thm:weakexist}. Finally the results are compared across time
steps. We do not prove convergence of the scheme as the time step tends to
zero, and the table reports the observed changes. The initial risky share
of group 1 changes by $0.64$, $0.48$, and $0.37$ percentage points as the
time step is halved. The changes shrink by a factor of about $0.76$, close
to the factor $1/\sqrt2$ by which the lattice spacing shrinks, and
geometric extrapolation with either factor gives a limit between $60.1\%$
and $60.4\%$. The initial
ratios and the risky shares of the other groups change by less than $0.05$
percentage points between the two finest lattices.

Three further checks separate sources of error. First, on synchronized
states all types revise together, and the scheme reduces to a binomial
chain in one dimension with the obstacles $\alpha$ and $-\beta$. Solving
this chain directly reproduces the adjoint errors of Table~\ref{tab:async}
in all digits shown. Its error is not monotone in the time step, because it
depends on where the ends of the band lie relative to the lattice points.
For example, it is $5.0\times10^{-4}$ at $\Delta t=1/20$ and
$1.1\times10^{-3}$ at $\Delta t=1/24$.
It falls to $1.0\times10^{-5}$ at $\Delta t=1/512$ and to
$6.1\times10^{-6}$ at $\Delta t=1/1024$. The adjoint error in the table
is therefore the error of the approximation of the reflected index in one
dimension, to which the coupling of the types and the interpolation add
nothing on synchronized states.
Second, at $\Delta t=1/32$ we recompute the equilibrium with the dispersion
of the sorted wedge capped at $1.45$ and at $2.0$ instead of $1.6$.
Multipliers, initial ratios, and risky shares change by less than
$10^{-13}$, and the Monte Carlo statistics do not change, since no simulated
path reaches a state with dispersion above that of the initial point.
Third, starting the outer root
finder from eight points drawn uniformly within $0.5$ of the solution in
each log multiplier gives the same multipliers to $10^{-14}$.

\subsection{Evaluations at Fixed Wealth}\label{ec:fixedwealth}

The comparisons of Section~\ref{sec:advice} require expectations of the form
$\E_x\int_0^\infty e^{-\delta t+k\ell_t}f(x_t)\dd t$ and
$\E_x\int_0^\infty e^{-\delta t+k\ell_t}\dd l^{\alpha,\beta}_t$ for a
reflected log index $x$ with regulator $\ell=l^\alpha-l^\beta$ and real $k$.
As functions of the initial point they solve $\varsigma w''+\nu w'-\delta w+f=0$ on the band with the Robin
conditions $w'+kw=c$ at both ends. For the running term $c=0$. For the
regulators $f=0$, and $c=-1$ at the lower end for $l^\alpha$ and $c=1$ at the
upper end for $l^\beta$, with $c=0$ at the other end. The expectations are
finite when $\delta+k\nu-k^2\varsigma>0$. The
utility of a plan $h_0e^{q\ell}$ and the terms of
Theorem~\ref{thm:certificate} for synchronized plans are combinations of
such expressions with $f$ equal to $1$ or $e^x$. For the misjudged
composition, the state is the logarithm of the ratio of wealth to
consumption, which is reflected at the ends of
the believed band and has the drift and volatility implied by the believed
portfolio rule under the true benchmark growth. Utility then solves a
linear equation with variable coefficients, which we discretize with $4{,}000$
points. When the believed rule is the true one, this computation agrees
with the evaluation in index coordinates to a relative error of
$2\times10^{-7}$.

\subsection{Sensitivity of the Losses}\label{ec:sensitivity}

Table~\ref{ec:tab:sensitivity} repeats the comparisons of
Section~\ref{sec:advice} with one parameter changed at a time. When the
inherited standard is ignored, the loss of the first policy is largest, and
constant, for inherited ratios of consumption to wealth at or above the
upper end of the band. The optimal policy and the first policy then both
cut at time zero, and the difference between their utilities does not
depend on the inherited standard. Similarly, the loss of the second policy
is largest, and constant, below the band. The table reports these largest losses. For the
misjudged composition, the household of the second type starts at the
middle of its band, as in Figure~\ref{fig:advice}(b).

\begin{table}[!t]\centering\small
\caption{Losses at Fixed Wealth for Other Parameters.}\label{ec:tab:sensitivity}
\begin{tabular}{@{}lccccc@{}}
\hline
 & \multicolumn{2}{c}{Inherited standard ignored} &
 \multicolumn{3}{c}{Misjudged composition}\\
Change from baseline & Lower end & $\hat K(\gamma)$ & Share off by $0.1$ &
 Only own type & Only other type\\
\hline
None & $16.8$ & $4.4$ & $0.27$ & $16.1$ & $4.5$\\
$\theta=0.2$ & $17.4$ & $4.6$ & $0.08$ & $2.9$ & $1.6$\\
$\theta=0.35$ & $16.5$ & $4.3$ & $0.46$ & $65.3$ & $7.0$\\
$L=1.25$ & $8.6$ & $2.0$ & $0.26$ & $14.7$ & $4.4$\\
$L=2$ & $29.0$ & $9.2$ & $0.28$ & $17.7$ & $4.6$\\
Lower curvature & $25.2$ & $6.4$ & $0.10$ & $3.3$ & $1.9$\\
Higher curvature & $11.0$ & $2.8$ & $0.41$ & $56.0$ & $6.2$\\
\hline
\end{tabular}

\vspace{4pt}
\begin{minipage}{\textwidth}\footnotesize
\emph{Notes.} Losses in percent of initial wealth. The baseline is that of
Figure~\ref{fig:advice}, with $\theta=0.3$ and $L=1.5$, curvature $\gamma=3.5$
for inherited standards, and types with $\gamma_1=2$ and $\gamma_2=5$ for the
misjudged composition. Lower curvature means $\gamma=2.5$ and $\gamma_2=4$,
and higher curvature means $\gamma=5$ and $\gamma_2=5.5$. The first two
columns are the largest losses of the policies that start at the lower end
of the band and at $\hat K(\gamma)$. The third column is the larger loss for
believed shares $0.4$ and $0.6$, and the last two columns correspond to
believed shares $0$ and $1$ of type 1.
\end{minipage}
\end{table}

The cost of ignoring the inherited standard comes from an unnecessary cut
and therefore rises with loss aversion. A misjudgment of the share by $0.1$ costs less than $0.5\%$ of wealth
throughout, whereas the cost of believing that the group consists only of
the household's own type rises steeply with $\theta$ and $\gamma_2$. With
$\theta=0.4$, or with $\gamma_2=6$, and believed shares of type 1 close to
zero, the discretized linear equation of Section~\ref{ec:fixedwealth} for
the expected discounted value of $(h_t/h_0)^{1-\gamma_2}$ has no positive
solution. This indicates that the expected utility under the believed rule
is minus infinity, so that no finite wealth would compensate for the error.
Such a household underestimates benchmark growth, chooses a band with too
high a ratio of consumption to wealth, and cuts its relative standard at a
rate at which the discounted disutility appears to diverge. We have not
proved this divergence for the continuous time problem, which would require
a bound on the growth rate of this expectation, and we report these cases
only as a numerical diagnosis.

\subsection{Two Types with Different Preferences}\label{ec:twotypes}

The example of Section~\ref{sec:computeexample} has groups with identical
preferences. This subsection computes two equilibria in which the types
differ in preferences as well as in initial data, so that the closed forms of
Theorem~\ref{thm:common} do not apply at any state and no symmetry of the
lattice is available. Both examples use the market of
Figure~\ref{fig:async}, $r=0.02$, $\delta=0.05$, $\kappa=0.30$, and
$\sigma=0.20$, two types of equal weight, $\alpha=0$ for both, and inherited
consumption of one, except for type 2 in the second example. In the first
example both types have
$\gamma=3$ and $\theta=0.3$, type 1 has $L=1.5$ and wealth $32$, and type 2
has $L=2.5$ and wealth $27$. The log bands of Lemma~\ref{lem:band} are
$[-0.777,1.046]$ for type 1 and $[-0.986,1.632]$ for type 2, so that the
band of the less loss averse type lies strictly inside the band of the
other. The weak interaction condition holds with $\Theta_1+\Theta_2=0.4$,
$q=1/7$, and $R=0.57$. In the second example type 1 has
$(\gamma,\theta,L)=(2,0.3,1.5)$, type 2 has $(5,0.5,2.5)$ and inherited
consumption $0.95$, both have wealth $30$, and $\Theta_1+\Theta_2=0.525$,
$q=0.21$, and $R=0.67$. In both examples neither
type revises at time zero.

The scheme is that of Section~\ref{ec:lattice} with two changes. The
adjoint of each type is stored on the full rectangular lattice, with each
coordinate covering the band of that type extended by $0.4$ on both sides,
and values between lattice points are obtained by bilinear interpolation.
The obstacles $\alpha_a$ and $-\beta_a$ and the directions $V_{\cdot b}$ are
type specific. With $L_1=L_2=2$ the same code reproduces the closed forms of
Theorem~\ref{thm:common} on the diagonal of the lattice. At $\Delta t=1/64$
the errors in the adjoint, the cost, and the risky share are
$8.3\times10^{-5}$, $1.5\times10^{-3}$, and $2.8\times10^{-3}$, of the same
order as in Table~\ref{tab:async}. Table~\ref{ec:tab:twotypes} reports the diagnostics
of Section~\ref{ec:diagnostics} for the two examples, together with
statistics of the revisions from $20{,}000$ simulated paths over $200$
years. Since the two types never share a band, synchronization in the sense
of Theorem~\ref{thm:common} cannot occur, and we report instead the expected
upward and downward revisions of the log relative standard per year and the
share of each type's revisions, weighted by their size, that occur at a
time step at which the other type also revises.

\begin{table}[!t]\centering\footnotesize
\caption{Two Types with Different Preferences.}\label{ec:tab:twotypes}
\setlength{\tabcolsep}{3pt}
\begin{tabular}{@{}lcccccc@{}}
\hline
 & \multicolumn{4}{c}{Different loss aversion} & \multicolumn{2}{c}{All preferences differ}\\
Time step & $1/16$ & $1/32$ & $1/64$ & $1/128$ & $1/64$ & $1/128$\\
\hline
Lattice points & $1{,}692$ & $3{,}366$ & $6{,}603$ & $13{,}000$ & $6{,}603$ & $13{,}000$\\
Howard iterations & 7 & 8 & 8 & 8 & 10 & 12\\
Largest complementarity residual & $2.5\times10^{-14}$ & $7.1\times10^{-15}$ & $6.7\times10^{-15}$ & $7.1\times10^{-15}$ & $2.5\times10^{-14}$ & $2.1\times10^{-14}$\\
Largest budget residual & $<10^{-15}$ & $<10^{-15}$ & $<10^{-15}$ & $<10^{-15}$ & $<10^{-15}$ & $<10^{-15}$\\
Monte Carlo cost minus wealth, in \% & $0.16$ & $0.27$ & $0.20$ & $0.05$ & $0.22$ & $0.03$\\
Initial risky share, in \% & $55.7,\,40.1$ & $55.6,\,40.0$ & $55.7,\,40.1$ & $55.7,\,40.0$ & $70.5,\,41.6$ & $70.5,\,41.6$\\
Upward revisions per year, in \% & $1.05,\,0.69$ & $1.05,\,0.69$ & $1.05,\,0.70$ & $1.05,\,0.69$ & $1.49,\,0.40$ & $1.49,\,0.39$\\
Downward revisions per year, in \% & $0.50,\,0.31$ & $0.50,\,0.31$ & $0.50,\,0.31$ & $0.50,\,0.31$ & $0.73,\,0.13$ & $0.73,\,0.13$\\
Share of joint revisions & $0.65,\,0.99$ & $0.65,\,0.99$ & $0.64,\,0.99$ & $0.64,\,0.99$ & $0.68,\,1.00$ & $0.67,\,1.00$\\
Type 2 alone after a joint revision & $0$ & $0$ & $0$ & $0$ & $0$ & $0$\\
Largest distance of $\log Y$ from center & $0.46$ & $0.46$ & $0.47$ & $0.47$ & $0.50$ & $0.50$\\
\hline
\end{tabular}

\vspace{4pt}
\begin{minipage}{\textwidth}\footnotesize
\emph{Notes.} The first example has $\gamma=3$, $\theta=0.3$, $L=1.5$ and
$2.5$, and wealth $32$ and $27$. The second has $(\gamma,\theta,L)=(2,0.3,1.5)$
and $(5,0.5,2.5)$, inherited consumption $1$ and $0.95$, and wealth $30$
for both types. Pairs of numbers refer to types 1 and 2. The Monte Carlo
row reports the larger deviation over the two types of the simulated cost
of consumption from wealth, from $20{,}000$ paths over $200$ years, with
standard errors of at most $0.72\%$. Revisions per year are the expected
increments of $\log h$ per year in each direction over the $200$ years.
The share of joint revisions is the share of a type's revisions, weighted
by their size, that occur at a time step at which the other type also
revises. The next row is the share of the revisions of type 2 that occur
alone after the first joint revision on the path. The last row is the
largest value over all paths and dates of $|\log Y_t-k(\mu)-gt-sB_t|$, to
be compared with $R=0.57$ in the first example and $R=0.67$ in the second.
\end{minipage}
\end{table}

The table shows three things. First, the type with the wider band revises
only together with the other type. In both examples, on all $20{,}000$
paths, once a joint revision has occurred type 2 never revises at a time
step at which type 1 does not revise, whereas about a third of the
revisions of type 1 are not joint. The $1.3\%$ of the revisions of type 2
that are not joint in the first example occur before the first joint
revision, because at the initial state type 2 is closer to its upper
boundary than type 1 is to its own. This is an observation about the
computed equilibria and their simulated paths, and we do not prove it for
the continuous time equilibrium. The following account explains it. In the
computed equilibria the revision regions are nested, since the indices of
type 1 stay in $[-0.76,1.03]$ and those of type 2 in $[-0.97,1.61]$. The
two log indices receive the same common increments, so their difference
changes only at revisions. A joint revision sets it to the difference
between the two boundaries on that side, and a revision of type 1 alone
keeps the index of type 1 at its boundary while moving the index of type 2
away from the boundary on the same side, since $V_{21}$ has the sign of
$V_{11}$ and is smaller. Type 1 then remains at least as close to each of
its boundaries as type 2 is to its own, so type 2 reaches a boundary only
while type 1 is already revising. Second, the
standard of the less loss averse type moves more in both directions, by
$1.05\%$ upward and $0.50\%$ downward per year against $0.69\%$ and $0.31\%$
in the first example, while the spending of type 2 follows the benchmark
between its own revisions. Third, the bound \eqref{eq:existbound} holds on
every path, with largest distances of $0.47$ and $0.50$ against $R=0.57$ and
$0.67$. The results change little across time steps. The initial risky
shares move by at most $0.1$ percentage points and the revision statistics
by about $0.01$.

Figure~\ref{fig:twotypes} in the main text follows one path of the first
example at $\Delta t=1/64$, and Figure~\ref{ec:fig:twotypes} shows the
ratios of consumption to wealth and the risky shares along the same path.
Type 1 cuts its standard at about 14 and 16 years and between 26 and 28
years. The risky share of type 1 exceeds that of type 2 by about ten
percentage points on average, and the two nearly coincide only when both
indices are near a boundary, at about 7.5 and 28 years.

\begin{figure}[!t]\centering
\includegraphics[width=\textwidth]{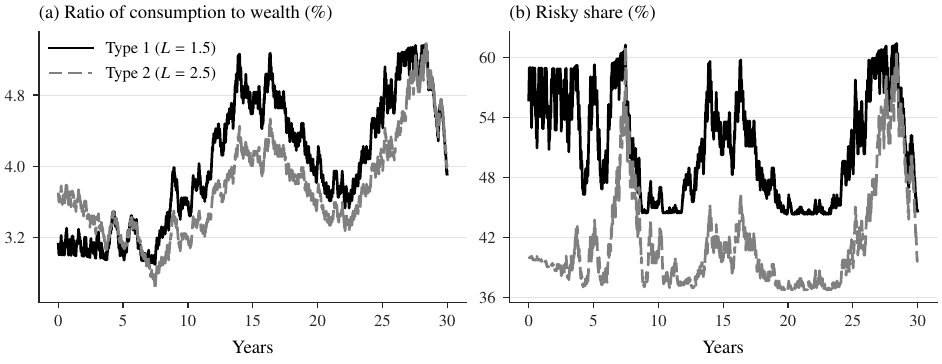}
\caption{Two Types with Different Loss Aversion: Spending and Risk Taking.
\emph{Notes.} The path of Figure~\ref{fig:twotypes}, with the parameters
given there.}
\label{ec:fig:twotypes}
\end{figure}

\subsection{Loss Aversion and Risk Taking inside the Band}\label{ec:lossfig}

Figure~\ref{ec:fig:loss} keeps the relative concern fixed and varies loss
aversion in the model of Section~\ref{sec:commonimpl}. The curves are
computed from the formulas of Lemma~\ref{lem:band} for the band and the
cost function. A higher $L$ widens the band for the ratio of consumption to
wealth and lowers the risky share inside the band, while the share at the
endpoints equals $\kappa/(\sigma\hg)$ for every $L$. Curvature and
adjustment costs could be separated empirically by combining portfolio
exposure with the ratio of consumption to wealth and the time since
spending was last revised, but we do not develop a formal identification
strategy.

\begin{figure}[!t]\centering
\includegraphics[width=\textwidth]{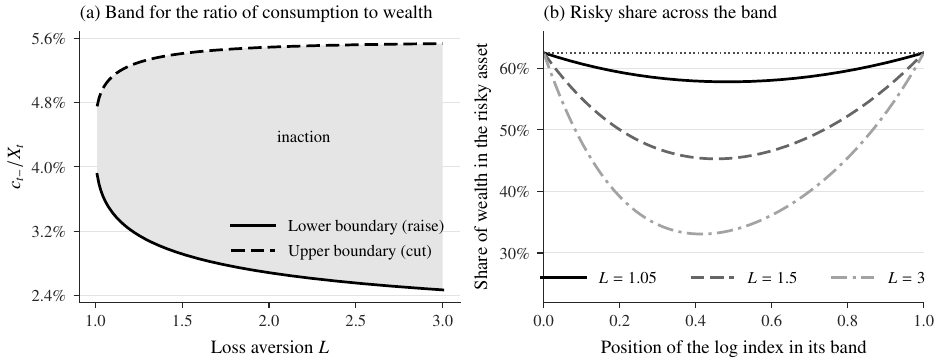}
\caption{Loss Aversion at Fixed Relative Concern.
\emph{Notes.} Common risk with $r=0.02$, $\delta=0.05$, $\kappa=0.30$,
$\sigma=0.20$, $\gamma=3$, $\theta=0.3$, and $\alpha=0$, so that $\hg=2.4$.
Panel (a) varies $L$ through the cost of reductions. Panel (b) compares risky
shares at corresponding positions in the band of each log index, and the
dotted line is the share $\kappa/(\sigma\hg)$ at the endpoints. The figure
shows one parameter family, and we do not claim monotonicity in $L$ in
general.}
\label{ec:fig:loss}
\end{figure}

\input{ec4a_uniqueness}
\input{ec4d_sync}

%% file: ec4a_uniqueness.tex
\section{Uniqueness and Synchronization}\label{ec:further}

An equilibrium lies in the \emph{bounded class} if its benchmark is an
$\FB$ adapted c\`adl\`ag process whose logarithm stays within a bounded
distance of $gt+sB_t$ in the uniform essential path norm, the class in which
Theorem~\ref{thm:weakexist} constructs equilibria and the benchmark is
unique for fixed multipliers. Equilibria are identified up to
indistinguishability, the hypotheses of Theorem~\ref{thm:weakexist} are in
force unless stated otherwise, and the notation of Sections~\ref{sec:exist}
and~\ref{ec:proofexist} is retained.

\subsection{Forgetting the Initial Reference}\label{ec:forget}

Uniqueness of the budget multipliers requires a lower bound on the response
of each type's expenditure to its own multiplier (Section~\ref{sec:scope}).
The next lemma shows that the inherited reference is forgotten at the first
revision in the relevant direction, and the one after it turns this into
such a bound.

\begin{lemma}\label{ec:lem:forget}
Fix the preference parameters and a price $P$ satisfying the hypotheses of
Theorem~\ref{ec:thm:comparable}. Let $0<h^1_{0-}\le h^2_{0-}$ be two
inherited references, let $h^j=h^\ast(P,h^j_{0-})$ and $v^j=u(h^j)$ for
$j=1,2$, and put
\[
 \Lambda_t=v^{1,-}([0,t])+v^{2,+}([0,t]),
\]
where the variations include the atoms at time zero. Then $h^1\le h^2$, and
almost surely, for every $t\ge0$, $\Lambda_t>0$ implies $h^1_s=h^2_s$ for
all $s\ge t$. In particular $h^1_t<h^1_{0-}$ or $h^2_t>h^2_{0-}$ implies that
the two plans coincide from time $t$ on.
\end{lemma}

\begin{proof}
Reference monotonicity in Lemma~\ref{ec:lem:lattice} gives $h^1\le h^2$.
With $z^j=P(h^j)^\gamma$, the adjoints of \eqref{ec:eq:index} satisfy
\[
 \Delta^\Psi_t=\Psi^1_t-\Psi^2_t
 =\E\Big[\int_t^\infty e^{-\delta(s-t)}P_s\big((h^2_s)^\gamma-(h^1_s)^\gamma\big)\dd s
 \,\Big|\,\F_t\Big]\ge0 .
\]
Both adjoints lie in $[-\beta,\alpha]$ and satisfy \eqref{ec:eq:foc} by
Theorem~\ref{ec:thm:comparable}.
The support conditions in \eqref{ec:eq:foc} give $\Psi^1=-\beta$ for
$\dd v^{1,-}$ almost every $t$, and then $\Psi^2\le\Psi^1=-\beta$ forces
$\Psi^2=-\beta$ there. Likewise $\Psi^2=\alpha$ for $\dd v^{2,+}$ almost
every $t$ forces $\Psi^1=\alpha$. Hence $\Delta^\Psi=0$ almost everywhere
for $\dd\Lambda$, that is $\int_{[0,\infty)}e^{-\delta t}\Delta^\Psi_t\dd\Lambda_t=0$.
Replacing the conditional expectation by the integrand, which is allowed
because $\dd\Lambda$ is an optional random measure and the integrand is
nonnegative, and applying Tonelli's theorem as in the proof of
Theorem~\ref{ec:thm:foc} yields
\[
 0=\E\int_0^\infty e^{-\delta s}P_s\big((h^2_s)^\gamma-(h^1_s)^\gamma\big)\Lambda_s\dd s .
\]
The integrand is nonnegative, so $\Lambda_s(h^2_s-h^1_s)=0$ for
$\dd s\otimes\mathbb P$ almost every $(s,\omega)$. Outside a null set, if
$\Lambda_t>0$ then $\Lambda_s\ge\Lambda_t>0$ for all $s\ge t$ by
monotonicity, so $h^1_s=h^2_s$ for almost every $s\ge t$, and right
continuity of both paths extends the equality to every $s\ge t$. Finally,
$h^1_t<h^1_{0-}$ requires $v^{1,-}([0,t])>0$, and $h^2_t>h^2_{0-}$ requires
$v^{2,+}([0,t])>0$.
\end{proof}

\begin{lemma}\label{ec:lem:elasticity}
Let $P^0_t=\exp(\nu_0t-\bm\kappa_0\cdot\bm W_t)$ with
$\sigma_0=|\bm\kappa_0|>0$, let $r_0=\delta-\nu_0-\sigma_0^2/2$, and assume
$K=K_0(\gamma)>0$ in the market $(r_0,\sigma_0)$, with band
$[z_-,z_+]$. Fix $\mu^-\le\mu^+$, $M\ge0$, $\theta\ge0$, and an inherited
reference $h_{0-}>0$. For an adapted price $Q$ with
$e^{-\theta M}P^0\le Q\le e^{\theta M}P^0$ define
\[
 \mathcal W_Q(\mu)=\E\int_0^\infty e^{-\delta t}Q_t\,h^\ast(e^\mu Q,h_{0-})_t\dd t,
 \qquad\mu^-\le\mu\le\mu^+ .
\]
Put $\upsilon=1-\gamma^{-1}$, $\bar\nu=\nu_0+\upsilon\sigma_0^2$,
$\Delta_0=(\bar\nu^2+2K\sigma_0^2)^{1/2}$,
$\bar x=\log z_+-\mu^-+\theta M-\gamma\log h_{0-}$, and
\begin{equation}\label{ec:eq:tailfunction}
 \mathcal P(x)=
 \begin{cases}
 \dfrac{\Delta_0+\bar\nu}{2\Delta_0}
   \exp\Big(-\dfrac{\Delta_0-\bar\nu}{\sigma_0^2}x\Big),&x\ge0,\\[8pt]
 1-\dfrac{\Delta_0-\bar\nu}{2\Delta_0}
   \exp\Big(\dfrac{\Delta_0+\bar\nu}{\sigma_0^2}x\Big),&x<0,
 \end{cases}
\end{equation}
which is the probability that $\bar\nu T+\sigma_0B_T$ exceeds $x$ for an
exponential time $T$ of rate $K$ independent of $B$. Let
\begin{equation}\label{ec:eq:zeta}
 \zeta=\Big(\frac{z_-}{z_+}\Big)^{1/\gamma}
 \exp\Big(-\frac{\mu^+-\mu^-}{\gamma}-2\theta M(1+\gamma^{-1})\Big)\mathcal P(\bar x)>0 .
\end{equation}
Then, uniformly over such prices, for $\mu^-\le\mu<\tilde\mu\le\mu^+$,
\begin{equation}\label{ec:eq:elasticity}
 \log\mathcal W_Q(\mu)-\log\mathcal W_Q(\tilde\mu)\ge\frac{\zeta}{\gamma}(\tilde\mu-\mu).
\end{equation}
\end{lemma}

\begin{proof}
Let $d=\tilde\mu-\mu>0$. The scaling identity in Step~5 of the proof of
Theorem~\ref{ec:thm:comparable} gives
\[
 h^\ast(e^{\mu+d}Q,h_{0-})=e^{-d/\gamma}h^\ast(e^\mu Q,e^{d/\gamma}h_{0-}) .
\]
The two optima on the right face the same price. By price monotonicity,
every optimum at a price $e^{\mu'}Q$ with $\mu'\ge\mu^-$ lies below the
reflected policy at the geometric price $e^{\mu^--\theta M}P^0$, whatever
its reference, and the index of that policy lies in $[z_-,z_+]$. Hence
\[
 h^\ast(e^{\mu'}Q,\cdot)_t\le\Big(\frac{z_+}{e^{\mu^--\theta M}P^0_t}\Big)^{1/\gamma},
\]
and this envelope is strictly below $h_{0-}$ on the event
$E_t=\{\log P^0_t>\bar x\}$. On $E_t$ the optimum with reference $h_{0-}$ has
fallen below its reference, so by Lemma~\ref{ec:lem:forget} it coincides
with the optimum with the larger reference $e^{d/\gamma}h_{0-}$ from time
$t$ on. Therefore
$h^\ast(e^{\mu+d}Q,h_{0-})_t=e^{-d/\gamma}h^\ast(e^\mu Q,h_{0-})_t$ on $E_t$,
while outside $E_t$ price monotonicity gives
$h^\ast(e^{\mu+d}Q,h_{0-})\le h^\ast(e^\mu Q,h_{0-})$. The lower envelope,
obtained in the same way from the reflected policy at $e^{\mu^++\theta M}P^0$,
and the lower bound on $Q$ give
\begin{align*}
 \mathcal W_Q(\mu)-\mathcal W_Q(\mu+d)
 &\ge(1-e^{-d/\gamma})\E\int_0^\infty e^{-\delta t}Q_th^\ast(e^\mu Q,h_{0-})_t\ind_{E_t}\dd t\\
 &\ge(1-e^{-d/\gamma})z_-^{1/\gamma}e^{-\mu^+/\gamma-\theta M(1+1/\gamma)}
 \E\int_0^\infty e^{-\delta t}(P^0_t)^{\upsilon}\ind_{E_t}\dd t ,
\end{align*}
and the upper envelope gives
$\mathcal W_Q(\mu)\le z_+^{1/\gamma}e^{-\mu^-/\gamma+\theta M(1+1/\gamma)}/K$,
since $\E\int_0^\infty e^{-\delta t}(P^0_t)^{\upsilon}\dd t=1/K$. Under the
exponential tilting by $(P^0_t)^\upsilon$ the drift of $\log P^0$ becomes
$\bar\nu$, so
\[
 K\E\int_0^\infty e^{-\delta t}(P^0_t)^{\upsilon}\ind_{E_t}\dd t
 =\int_0^\infty Ke^{-Kt}\,\mathbb P(\bar\nu t+\sigma_0B_t>\bar x)\dd t=\mathcal P(\bar x),
\]
where the density of $\bar\nu T+\sigma_0B_T$ is
$(K/\Delta_0)\exp((\bar\nu x-\Delta_0|x|)/\sigma_0^2)$, whose tails give
\eqref{ec:eq:tailfunction}. Combining the three displays,
$\mathcal W_Q(\mu+d)\le\mathcal W_Q(\mu)[1-\zeta(1-e^{-d/\gamma})]$. Applying
this bound on each of $N$ equal pieces of $[\mu,\mu+d]$ and using
$-\log(1-x)\ge x$ gives
$\log\mathcal W_Q(\mu)-\log\mathcal W_Q(\mu+d)\ge N\zeta(1-e^{-d/(N\gamma)})$,
and $N\to\infty$ proves \eqref{ec:eq:elasticity}.
\end{proof}

A household with a low inherited standard waits long before it cuts, so
$\bar x$ is large and $\zeta$ small, but $\zeta>0$ for every reference.

\subsection{Uniqueness with All Budgets Imposed}\label{ec:budgetunique}

In Step~3 of the proof of Theorem~\ref{thm:weakexist}, every equilibrium
multiplier vector satisfies $\mathbf A\mu\in\prod_j\tilde I_j$ with
$\tilde I_j=I_j-\log X_j$. Let $\mathbf B=\mathbf A^{-1}$, with entries
$\mathbf B_{ij}=\gamma_i\ind_{i=j}-\theta_i(\gamma_i-1)\pi_j$, and put
\begin{equation}\label{ec:eq:multiplierbox}
 \mu_i^-=\sum_j\min_{x\in\tilde I_j}\mathbf B_{ij}x,\qquad
 \mu_i^+=\sum_j\max_{x\in\tilde I_j}\mathbf B_{ij}x,\qquad
 k_\pm=\frac1D\Big(\bar m-\sum_i\frac{\pi_i\mu_i^\mp}{\gamma_i}\Big),\qquad
 \bar R=R+\max\{|k_-|,|k_+|\},
\end{equation}
where $\bar m=\sum_i\pi_im_i/\gamma_i$. These numbers depend on the
preferences, the markets, the weights, and the initial wealth, but not on
the inherited references. For each type apply Lemma~\ref{ec:lem:elasticity}
with
\begin{equation}\label{ec:eq:zetaassignment}
 (\nu_0,\sigma_0,K,z_-,z_+,\mu^-,\mu^+,\theta,M,h_{0-})
 =\big(\nu_i+\theta_ig,\hat\kappa_i,\hat K_i(\gamma_i),e^{a_i},e^{b_i},
 \mu_i^-,\mu_i^+,\theta_i,\bar R,h_{i,0-}\big),
\end{equation}
and write $\zeta_i$ for the resulting constant. The geometric price
$e^{\delta_it}\xi^i_te^{\theta_i(gt+sB_t)}$ of the artificial market of
type $i$ has exactly this form, with the band $[e^{a_i},e^{b_i}]$.

\begin{theorem}\label{ec:thm:budgetunique}
Under the hypotheses of Theorem~\ref{thm:weakexist}, if
\begin{equation}\label{ec:eq:uniquecondition}
 \Theta_1+\Theta_2+\sum_i\pi_i\frac{\theta_i(1+\gamma_i^{-1})}{\zeta_i}<1,
\end{equation}
then the mean field equilibrium is unique in the bounded class, and so is
the vector of budget multipliers. For fixed individual parameters, weights, wealth,
and references that satisfy $K_i(\gamma_i)>0$ and $\kappa_i>0$ at
$\theta_i=0$, condition \eqref{ec:eq:uniquecondition} holds whenever all
relative concerns are small enough.
\end{theorem}

\begin{proof}
Let $(Y,h)$ be an equilibrium in the bounded class and $y=\log Y$. Each
type's price is comparable to its artificial geometric price, so by the
budget part of Theorem~\ref{ec:thm:comparable} its plan is
$h^\ast(\lambda_ie^{\delta_it}\xi^ie^{\theta_iy},h_{i,0-})$ for a unique
multiplier $\lambda_i=e^{\mu_i}$, and \eqref{eq:consistency} says that $y$
is a fixed point of $\mathcal T_\mu$. By the fixed point uniqueness in
Theorem~\ref{thm:weakexist}, $y=y^\mu$ and
$|y_t-k(\mu)-gt-sB_t|\le R$. The envelopes \eqref{ec:eq:wealthenv} with
$\mathcal E_i(\mu)=X_i$ put $\mathbf A\mu$ in $\prod_j\tilde I_j$, hence
$\mu_i^-\le\mu_i\le\mu_i^+$, $k_-\le k(\mu)\le k_+$, and
$\|y-gt-sB\|_\infty\le\bar R$. This enclosure holds for every equilibrium
in the class.

Now let $(y,\mu)$ and $(\tilde y,\tilde\mu)$ be two equilibria and
$\Delta_y=\|y-\tilde y\|_\infty$. For a frozen benchmark $y$ write
$Q_i[y]_t=e^{\delta_it}\xi^i_te^{\theta_iy_t}$ and
$\mathcal W_i(\mu_i,y)=\E\int_0^\infty e^{-\delta_it}Q_i[y]_t
h^\ast(e^{\mu_i}Q_i[y],h_{i,0-})_t\dd t$. By the enclosure, $Q_i[y]$ lies
between $e^{\pm\theta_i\bar R}$ times the artificial geometric price, so
Lemma~\ref{ec:lem:elasticity} applies on $[\mu_i^-,\mu_i^+]$ with the constant
$\zeta_i/\gamma_i$. The price moves by at most $e^{\theta_i\Delta_y}$
between $y$ and $\tilde y$, so \eqref{ec:eq:loglipschitz} gives
$|\log\mathcal W_i(\mu_i,y)-\log\mathcal W_i(\mu_i,\tilde y)|
\le\theta_i(1+\gamma_i^{-1})\Delta_y$. Both equilibrium expenditures equal
$X_i$. If $\mu_i<\tilde\mu_i$, then
\[
 \frac{\zeta_i}{\gamma_i}(\tilde\mu_i-\mu_i)
 \le\log\mathcal W_i(\mu_i,y)-\log\mathcal W_i(\tilde\mu_i,y)
 =\log\mathcal W_i(\tilde\mu_i,\tilde y)-\log\mathcal W_i(\tilde\mu_i,y)
 \le\theta_i(1+\gamma_i^{-1})\Delta_y ,
\]
and the case $\mu_i>\tilde\mu_i$ is symmetric, so
$|\mu_i-\tilde\mu_i|\le\gamma_i\theta_i(1+\gamma_i^{-1})\Delta_y/\zeta_i$.
Subtracting the two consistency equations and applying
\eqref{ec:eq:loglipschitz} to each type gives
\[
 (1-\Theta_1)\Delta_y\le\Theta_2\Delta_y+\sum_i\frac{\pi_i}{\gamma_i}|\mu_i-\tilde\mu_i|
 \le\Big(\Theta_2+\sum_i\pi_i\frac{\theta_i(1+\gamma_i^{-1})}{\zeta_i}\Big)\Delta_y .
\]
Condition \eqref{ec:eq:uniquecondition} forces $\Delta_y=0$, and then the
multipliers and the plans coincide.

At zero relative concern the artificial markets are the given markets and
$\Theta_1=\Theta_2=0$. The bands and the constants in
\eqref{ec:eq:multiplierbox} depend continuously on the parameters by
Theorem~\ref{ec:thm:allrates}, and each $\zeta_i$ has a positive limit, so
the left side of \eqref{ec:eq:uniquecondition} tends to zero as
$\max_i\theta_i\to0$, while the hypotheses of Theorem~\ref{thm:weakexist}
persist.
\end{proof}

Condition \eqref{ec:eq:uniquecondition} is sufficient and can be
conservative. The lower bound on the own response that
Section~\ref{sec:scope} asks for follows from reference forgetting and
nondegenerate risk.

\begin{corollary}\label{ec:cor:outer}
Under the hypotheses of Theorem~\ref{thm:weakexist} and
\eqref{ec:eq:uniquecondition}, for $\mu\in\prod_i[\mu_i^-,\mu_i^+]$ let
$y^\mu$ be the benchmark of Theorem~\ref{thm:weakexist} at the fixed
multipliers $\mu$, let $r_i(y^\mu)$ be the unique log multiplier at which
the budget of type $i$ binds when $y^\mu$ is held fixed, and define
$(\mathcal S\mu)_i=\min\{\mu_i^+,\max\{\mu_i^-,r_i(y^\mu)\}\}$. Then
$\mathcal S$ is a contraction of $\prod_i[\mu_i^-,\mu_i^+]$ in the norm
$\|v\|_w=\sum_i\pi_i|v_i|/\gamma_i$, with factor
\[
 \chi=\frac{\sum_i\pi_i\theta_i(1+\gamma_i^{-1})/\zeta_i}{1-\Theta_1-\Theta_2}<1 ,
\]
its fixed point is the equilibrium multiplier vector $\mu^\ast$, and
$\|\mu^k-\mu^\ast\|_w\le\chi^k\|\mu^0-\mu^\ast\|_w$ along the iteration
$\mu^{k+1}=\mathcal S\mu^k$. Every $\mu$ in the box satisfies
$\|\mu-\mu^\ast\|_w\le\|\mu-\mathcal S\mu\|_w/(1-\chi)$.
\end{corollary}

\begin{proof}
For $\mu$ in the box, $\|y^\mu-gt-sB\|_\infty\le\bar R$, so
Lemma~\ref{ec:lem:elasticity} applies on each interval $[\mu_i^-,\mu_i^+]$
at every benchmark used here, and the root $r_i(y^\mu)$ exists by
Theorem~\ref{ec:thm:comparable}. Let $y=y^\mu$ and $\tilde y=y^{\tilde\mu}$
be two such benchmarks, with clipped roots $\varkappa_i=(\mathcal S\mu)_i$ and
$\tilde\varkappa_i=(\mathcal S\tilde\mu)_i$, and suppose
$\varkappa_i<\tilde\varkappa_i$. Then $\varkappa_i<\mu_i^+$, so
$\varkappa_i\ge r_i(y)$ and $\mathcal W_i(\varkappa_i,y)\le X_i$, and
$\tilde\varkappa_i>\mu_i^-$, so $\tilde\varkappa_i\le r_i(\tilde y)$ and
$\mathcal W_i(\tilde\varkappa_i,\tilde y)\ge X_i$. Hence
\[
 \frac{\zeta_i}{\gamma_i}(\tilde\varkappa_i-\varkappa_i)
 \le\log\mathcal W_i(\varkappa_i,y)-\log\mathcal W_i(\tilde\varkappa_i,y)
 \le\log\mathcal W_i(\tilde\varkappa_i,\tilde y)-\log\mathcal W_i(\tilde\varkappa_i,y)
 \le\theta_i(1+\gamma_i^{-1})\|y-\tilde y\|_\infty ,
\]
and the same bound holds with the roles exchanged. Step~2 of the proof of
Theorem~\ref{thm:weakexist} gives
$\|y^\mu-y^{\tilde\mu}\|_\infty\le\|\mu-\tilde\mu\|_w/(1-\Theta_1-\Theta_2)$.
Weighting by $\pi_i/\gamma_i$ and summing shows
$\|\mathcal S\mu-\mathcal S\tilde\mu\|_w\le\chi\|\mu-\tilde\mu\|_w$. The
equilibrium multipliers lie in the box by the enclosure in the proof of
Theorem~\ref{ec:thm:budgetunique}, and their unclipped roots are their own
coordinates, so they form a fixed point. Banach's theorem gives the rest.
\end{proof}

As an example take $r=0.025$ and two types with
\[
\begin{array}{cccccccccc}
 i&\pi_i&\gamma_i&\delta_i&\theta_i&\kappa_{a,i}&\kappa_{p,i}&X_i&h_{i,0-}&(\alpha_i,\beta_i)\\[2pt]
 1&0.4&2&0.04&0.08&0.20&0.10&25&1&(0,0.5)\\
 2&0.6&3&0.05&0.12&0.25&0.15&18&0.8&(0,0.6)
\end{array}
\]
so that $L_1=1.02$ and $L_2=1.03$. Then $\Theta_1+\Theta_2=0.144$,
$(\hat K_1(\gamma_1),\hat K_2(\gamma_2))=(0.0389,0.0436)$,
$(\zeta_1,\zeta_2)=(0.178,0.221)$, and the left side of
\eqref{ec:eq:uniquecondition} is $0.847$, so the equilibrium is unique. The
script \texttt{uniqueness\_constants.py} in the code package evaluates
these constants.

%% file: ec4d_sync.tex
\subsection{Synchronization of Revisions}\label{ec:sync}

On every simulated path of Section~\ref{sec:computeexample} the three
groups eventually revise together. This is a property of the exact
equilibrium. The relevant objects are the standards normalized by the
multipliers rather than the consumption levels, and a normalization of the
utility allows the groups to differ in $(\alpha_i,\beta_i)$ as long as they
share $L$. Put $A=1-\delta\alpha>0$. Discounted integration by parts, as in
\eqref{ec:eq:coercive}, gives for every admissible $h$
\begin{equation}\label{ec:eq:normalized}
 \frac{U(h)-\alpha u(h_{0-})}{A}
 =\E\int_0^\infty e^{-\delta t}u(h_t)\dd t
 -\frac{L-1}{\delta}\,\E\int_{[0,\infty)}e^{-\delta t}\dd u(h)^-_t ,
\end{equation}
because $-\alpha\dd v^+-\beta\dd v^-=-\alpha\dd v-(\alpha+\beta)\dd v^-$ and
$(\alpha+\beta)/A=(L-1)/\delta$. At a fixed expenditure deflator and budget,
the optimal policy therefore depends on $(\alpha,\beta)$ only through $L$.
For the Lagrangian, the price must also be normalized: its optimum is that
for adjustment parameters $(0,(L-1)/\delta)$ with the multiplier $\lambda/A$,
where $\lambda$ is the multiplier of the original problem, and the band of
this normalized problem is $[\za/A,\zb/A]$, which depends on
$(\alpha,\beta)$ only through $L$ by the proof of
Theorem~\ref{ec:thm:allrates}.

\begin{theorem}\label{ec:thm:sync}
Consider finitely many groups with weights $\pi_i>0$, deterministic
initial wealth and inherited standards, common $(\gamma,\theta,\delta,L)$,
adjustment parameters $(\alpha_i,\beta_i)$ that may differ across groups,
and common risk only, with $\kappa>0$. Let $(Y,h)$ be a mean field
equilibrium in the bounded class, with multipliers $\lambda_i$, and put
$A_i=1-\delta\alpha_i$,
\[
 \tilde h_i(t)=\Big(\frac{\lambda_i}{A_i}\Big)^{1/\gamma}h_i(t),\qquad
 \tau=\inf\{t\ge0:\tilde h_1(t)=\cdots=\tilde h_n(t)\}.
\]
Then the spread $\max_i\tilde h_i-\min_i\tilde h_i$ is nonincreasing in time,
$\tau<\infty$ almost surely, $\E e^{\eta\tau}<\infty$ for some $\eta>0$,
and for every $t\ge\tau$ and all $i,j$
\[
 \frac{z_i(t)}{A_i}=\frac{z_j(t)}{A_j},\qquad
 \frac{h_i(t)}{h_j(t)}=\frac{c_i(t)}{c_j(t)}=\Big(\frac{A_i\lambda_j}{A_j\lambda_i}\Big)^{1/\gamma}.
\]
In particular all groups revise at the same dates and in the same direction
after $\tau$. Before $\tau$, an upward revision is possible only for a group
with the lowest normalized standard and a downward revision only for a
group with the highest, and the groups form ordered clusters that merge but
never split. The hypotheses hold for the equilibrium of
Theorem~\ref{thm:weakexist} whenever $\theta<\gamma/(\gamma+1)$ and
$\hat K(\gamma)>0$.
\end{theorem}

\begin{proof}
Put $P_t=e^{\delta t}\xi_tY_t^\theta$, which is the same price for every
group once the benchmark is fixed. By \eqref{ec:eq:normalized}, $h_i$
maximizes the normalized Lagrangian with adjustment parameters $(0,\beta')$,
$\beta'=(L-1)/\delta$, at the price $(\lambda_i/A_i)P$. The scaling identity
in Step~5 of the proof of Theorem~\ref{ec:thm:comparable} then shows that
$\tilde h_i$ is the unique optimum of that normalized problem at the common
price $P$ with the reference $\tilde h_{i,0-}=(\lambda_i/A_i)^{1/\gamma}h_{i,0-}$.
Relabel the groups so that $\tilde h_{1,0-}\le\cdots\le\tilde h_{n,0-}$. The
reference comparison of Lemma~\ref{ec:lem:lattice} gives, up to
indistinguishability,
\begin{equation}\label{ec:eq:syncorder}
 \tilde h_1(t)\le\cdots\le\tilde h_n(t),\qquad t\ge0 .
\end{equation}
Write $\tilde z_i=P\tilde h_i^\gamma=z_i/A_i$ for the normalized index and
$\tilde\Psi_i$ for the normalized adjoint, which lies in $[-\beta',0]$ and
satisfies the support conditions of \eqref{ec:eq:foc} with $\alpha$
replaced by zero. For two groups $i<j$,
\begin{equation}\label{ec:eq:syncadjoint}
 \tilde\Psi_i(t)-\tilde\Psi_j(t)
 =\E\Big[\int_t^\infty e^{-\delta(s-t)}P_s\big(\tilde h_j(s)^\gamma-\tilde h_i(s)^\gamma\big)\dd s
 \,\Big|\,\F_t\Big]\ge0,
\end{equation}
and the inequality is strict on $\{\tilde h_i(t)<\tilde h_j(t)\}$. To see
this at all times at once, let $S$ be a bounded stopping time. On
$\{\tilde h_i(S)<\tilde h_j(S)\}$ right continuity and positivity of $P$
make the integral started at $S$ strictly positive, hence so is its
conditional expectation, and the optional section theorem shows that
$\{\tilde h_i<\tilde h_j,\tilde\Psi_i=\tilde\Psi_j\}$ is evanescent. Hence
on $\{\tilde h_i(t)<\tilde h_j(t)\}$ we have $\tilde\Psi_j(t)<0$ and
$\tilde\Psi_i(t)>-\beta'$, so by the support conditions $\tilde h_j$ cannot
increase and $\tilde h_i$ cannot decrease there. The support conditions also
constrain the atoms of $\dd\tilde h_i$, so they apply at jump times. Let $\mathcal D^{ij}_t=\tilde h_j(t)-\tilde h_i(t)\ge0$.
On $\{\mathcal D^{ij}>0\}$ the continuous part of $\dd\mathcal D^{ij}$ is
nonpositive, and on $\{\mathcal D^{ij}=0\}$ the continuous part of the
differential of a c\`adl\`ag path of finite variation has no mass, which
follows from the change of variables formula for $x\mapsto x^+$ with the
left and the right derivative. At a jump time with $\mathcal D^{ij}_t>0$
neither an upward jump of $\tilde h_j$ nor a downward jump of $\tilde h_i$
is allowed, so $\Delta\mathcal D^{ij}_t\le0$, and at a jump time with
$\mathcal D^{ij}_t=0$ we have $\Delta\mathcal D^{ij}_t=-\mathcal D^{ij}_{t-}\le0$.
This includes time zero with $\mathcal D^{ij}_{0-}=\tilde h_{j,0-}-\tilde h_{i,0-}$.
Hence every $\mathcal D^{ij}$ is nonincreasing and stays at zero once it
reaches zero, so the groups form ordered clusters that merge but never
split, only the lowest cluster can move up and only the highest can move
down, and the spread $\mathcal D=\mathcal D^{1n}$ is nonincreasing. Its
first zero is $\tau$, a stopping time attained whenever finite, and
\begin{equation}\label{ec:eq:synctrap}
 \tilde h_{1,0-}\le\tilde h_i(t)\le\tilde h_{n,0-}\qquad(0\le t<\tau).
\end{equation}

It remains to prove that $\tau$ is finite. Put
$\bar\mu=\gamma^{-1}\sum_i\pi_i\log(\lambda_i/A_i)$,
$\underline y=(\log\tilde h_{1,0-}-\bar\mu)/(1-\theta)$,
$\overline y=(\log\tilde h_{n,0-}-\bar\mu)/(1-\theta)$, and
$\Xi_t=\nu t-\kappa B_t$ with $\nu=\delta-r-\kappa^2/2$. Consistency gives
$(1-\theta)\log Y_t=\sum_i\pi_i\log\tilde h_i(t)-\bar\mu$, so
$\underline y\le\log Y_t\le\overline y$ for $t<\tau$ by
\eqref{ec:eq:synctrap}. Since
$\log\tilde z_i(t)=\Xi_t+\theta\log Y_t+\gamma\log\tilde h_i(t)$ and the
normalized indices lie between deterministic constants
$0<\underline z\le\overline z$ by \eqref{ec:eq:comparableband}, we obtain
\[
 \underline\Xi\le\Xi_t\le\overline\Xi\quad(t<\tau),\qquad
 \underline\Xi=\log\underline z-\theta\overline y-\gamma\log\tilde h_{n,0-},\qquad
 \overline\Xi=\log\overline z-\theta\underline y-\gamma\log\tilde h_{1,0-}.
\]
Let $\underline\Xi'=\min\{\underline\Xi,0\}-1$,
$\overline\Xi'=\max\{\overline\Xi,0\}+1$, and
$\tau^\ast=\inf\{t\ge0:\Xi_t\notin(\underline\Xi',\overline\Xi')\}$. If
$\tau>\tau^\ast$, the bounds would hold at $\tau^\ast$, which is impossible,
so $\tau\le\tau^\ast$. Fix $t_0>0$ and put
$p^\ast=\mathbb P(|\nu t_0-\kappa B_{t_0}|>\overline\Xi'-\underline\Xi')>0$.
An increment of this size over a block of length $t_0$ forces an exit
during the block from any point of the interval, so independence of the
increments gives $\mathbb P(\tau^\ast>jt_0)\le(1-p^\ast)^j$ for every $j$, and
$\E e^{\eta\tau^\ast}<\infty$ whenever $e^{\eta t_0}(1-p^\ast)<1$. The ratio
identities follow from $\mathcal D=0$ after $\tau$ and from
\eqref{ec:eq:syncorder}, and $c_i/c_j=h_i/h_j$ because the groups share
$Y^\theta$. Finally, Theorem~\ref{thm:weakexist} together with
Theorem~\ref{ec:thm:comparable} supplies deterministic multipliers, finite
Lagrangians, the support conditions, and deterministic index bounds.
\end{proof}

The theorem allows any finite distribution of initial wealth and inherited
standards, does not need unique multipliers, and transfers to the finite
populations of Theorem~\ref{thm:exact}. Groups with different $L$ cannot be
included, even without interaction.

\begin{proposition}\label{ec:prop:nonsync}
Let $\theta=0$, let two households share $(r,\kappa,\delta,\gamma)$ with
$K(\gamma)>0$, and let their coefficients of loss aversion differ, so that
their bands $[a_1,b_1]$ and $[a_2,b_2]$ have different widths. Then almost
surely there is no finite time after which the closed supports of the
upward regulators of the two optimal policies coincide and those of the
downward regulators coincide as well.
\end{proposition}

\begin{proof}
With $\theta=0$ the two problems do not involve the benchmark, and by
Theorem~\ref{thm:br} the optimal log indices are the reflected diffusions
$x_i(t)=x_i(0)+\nu t-\kappa B_t+l^{\alpha,i}_t-l^{\beta,i}_t\in[a_i,b_i]$,
where the regulators are continuous after the possible initial revision,
$\dd l^{\alpha,i}$ is carried by $\{x_i=a_i\}$ and $\dd l^{\beta,i}$ by
$\{x_i=b_i\}$. Suppose the upward supports coincide after some time and
the downward supports do as well. The difference $\mathcal D^x_t=x_1(t)-x_2(t)$ is a
continuous path of finite variation whose differential is carried by
$\{\mathcal D^x=a_1-a_2\}\cup\{\mathcal D^x=b_1-b_2\}$, because simultaneous upward
revisions require both lower endpoints and simultaneous downward revisions
both upper endpoints. The differential of a continuous path of finite
variation gives no mass to level sets, so $\dd\mathcal D^x=0$ after that
time.
Each reflected diffusion regulates both endpoints at arbitrarily late
times, which follows from independent Brownian increments with large
excursions in either direction on successive blocks. The constant
$\mathcal D^x$ would then equal both $a_1-a_2$ and $b_1-b_2$, contradicting the
unequal widths.
\end{proof}

In the example $\alpha=0$, so $A_i=1$, and the event that the code detects,
a spread of the log indices below $10^{-9}$, is the event of
Theorem~\ref{ec:thm:sync}. The theorem does not prove convergence of the
simulated synchronization times.

%% file: supplement_body.tex
\section{Further Results}\label{sup:intro}

This appendix collects results that build on Appendix~\ref{ec:further}: sharper constants
for the uniqueness condition, equilibria with given directions of the
initial revisions, the limit of vanishing loss aversion, populations with
one benchmark sensitive type, infinitely many types, arithmetic
benchmarks, certificates for approximate best responses, and stability
bounds for the lattice scheme of Section~\ref{ec:lattice}. The notation of
Sections~\ref{sec:exist} and~\ref{ec:proofexist} and of Appendix~\ref{ec:further} is
retained, and so are the bounded class and the standing hypotheses of
Appendix~\ref{ec:further}.

\input{sup_twotail}
\input{ec4b_directions}
\input{ec4c_vanishing}
\input{ec4e_oneconcern}
\input{ec4f_infinite}
\input{ec4g_arithmetic}
\input{ec3_additions}

%% file: sup_twotail.tex
\subsection{Sharper Constants for Uniqueness}\label{ec:sharper}

Using both directions of forgetting and the asymmetric enclosure of the
benchmark improves the constants of Theorem~\ref{ec:thm:budgetunique}.

\begin{proposition}\label{ec:prop:twotail}
Put $\underline\vartheta_i=\theta_i(k_--R)$,
$\overline\vartheta_i=\theta_i(k_++R)$,
$\bar x^+_i=b_i-\mu_i^--\underline\vartheta_i-\gamma_i\log h_{i,0-}$,
$\bar x^-_i=a_i-\mu_i^+-\overline\vartheta_i-\gamma_i\log h_{i,0-}$, let
$\mathcal P_i$ be the function \eqref{ec:eq:tailfunction} of type $i$, and
\[
 \zeta_i^\sharp=\exp\Big(-\frac{b_i-a_i+\mu_i^+-\mu_i^-}{\gamma_i}
 -(\overline\vartheta_i-\underline\vartheta_i)(1+\gamma_i^{-1})\Big)
 \big[\mathcal P_i(\bar x^+_i)+1-\mathcal P_i(\bar x^-_i)\big].
\]
Then $\zeta_i^\sharp\ge\zeta_i$, and Theorem~\ref{ec:thm:budgetunique} and
Corollary~\ref{ec:cor:outer} remain true with $\zeta_i$ replaced by
$\zeta_i^\sharp$.
\end{proposition}

\begin{proof}
Every frozen benchmark in the proof of Theorem~\ref{ec:thm:budgetunique}
satisfies $k_--R\le y_t-gt-sB_t\le k_++R$, so $Q_i[y]$ lies between
$e^{\underline\vartheta_i}$ and $e^{\overline\vartheta_i}$ times the artificial
geometric price. Suppress the type subscript and fix
$\mu^-\le\mu<\mu+d\le\mu^+$. The upper envelope falls below the reference $h_{0-}$ on
$\{\log P^0_t>\bar x^+\}$, exactly as in Lemma~\ref{ec:lem:elasticity}. On
$\{\log P^0_t<\bar x^-\}$ the lower envelope at the price $e^{\mu}Q$,
$(z_-/(e^{\mu+\overline\vartheta}P^0_t))^{1/\gamma}$, exceeds
$e^{(\mu^+-\mu)/\gamma}h_{0-}\ge e^{d/\gamma}h_{0-}$, so the optimum with the
larger reference $e^{d/\gamma}h_{0-}$ has risen above its reference, and
Lemma~\ref{ec:lem:forget} again gives equality of the two optima from that
time on. The two events are disjoint because $\bar x^-<\bar x^+$. The
coefficient of the lower expenditure envelope is
$z_-^{1/\gamma}e^{\underline\vartheta-(\mu^++\overline\vartheta)/\gamma}$ and
that of the upper one is
$z_+^{1/\gamma}e^{\overline\vartheta-(\mu^-+\underline\vartheta)/\gamma}$, whose
ratio is the exponential prefactor of $\zeta^\sharp$. The tilted Gaussian
computation now produces both tails, and the rest of the proof is
unchanged. Finally $\underline\vartheta_i\ge-\theta_i\bar R$ and
$\overline\vartheta_i\le\theta_i\bar R$, so $\bar x^+_i\le\bar x_i$, the
prefactor is no smaller than before, and the second tail is nonnegative,
which gives $\zeta_i^\sharp\ge\zeta_i$.
\end{proof}

In the example of Section~\ref{ec:budgetunique} the sharper constants are
$(\zeta_1^\sharp,\zeta_2^\sharp)=(0.261,0.282)$, and the left side of
\eqref{ec:eq:uniquecondition} falls from $0.847$ to $0.668$. Raising the
concerns to $\theta=(0.10,0.15)$ gives $0.882$ with the sharper constants,
whereas the original constants give $1.139$ and no longer certify
uniqueness.

%% file: ec4b_directions.tex
\subsection{Initial Revision Directions}\label{ec:directions}

Without a smallness condition we can still count equilibria by the
directions of their initial revisions. The tool is a problem in which the
initial level of the relative standard is free and charged linearly. Unlike
the problem with an inherited reference, this problem is homogeneous in the
price, so the budget normalization becomes explicit.

For one type and a sign $\mathfrak s\in\{+,-\}$ put $\iota^+=\alpha$ and
$\iota^-=-\beta$. For a positive adapted c\`adl\`ag process $h$ of locally
finite variation with $v=u(h)$ and the integrability of
Section~\ref{sec:setup}, but with a freely chosen initial value $h_0>0$,
define
\begin{equation}\label{ec:eq:relaxedobjective}
 U^\mathfrak s(h)=\E\Big[\int_0^\infty e^{-\delta t}v_t\dd t
 -\int_{(0,\infty)}e^{-\delta t}\big(\alpha\dd v^+_t+\beta\dd v^-_t\big)-\iota^\mathfrak s v_0\Big],
 \qquad \J^\mathfrak s_P(h)=U^\mathfrak s(h)-\E\int_0^\infty e^{-\delta t}P_th_t\dd t .
\end{equation}

\begin{lemma}\label{ec:lem:relaxed}
Let $P$ satisfy the hypotheses of Theorem~\ref{ec:thm:comparable}. Then
$\J^\mathfrak s_P$ has a unique finite maximizer $h^\mathfrak s(P)$. It satisfies the
obstacle and support conditions of \eqref{ec:eq:foc} on $(0,\infty)$ and
$\Psi_0=\iota^\mathfrak s$, its index lies in the envelopes of
Theorem~\ref{ec:thm:comparable}, and for $a>0$
\begin{equation}\label{ec:eq:relaxedscaling}
 h^\mathfrak s(aP)=a^{-1/\gamma}h^\mathfrak s(P).
\end{equation}
It is decreasing in $P$ and satisfies \eqref{ec:eq:loglipschitz}. If the
optimum $h^\ast(P,h_{0-})$ of the original problem revises strictly at time
zero in the direction $\mathfrak s$, then $h^\ast(P,h_{0-})=h^\mathfrak s(P)$.
Conversely, if $h^\mathfrak s(P)_0\ge h_{0-}$ for $\mathfrak s=+$, or
$h^\mathfrak s(P)_0\le h_{0-}$ for $\mathfrak s=-$, then $h^\mathfrak s(P)=h^\ast(P,h_{0-})$.
\end{lemma}

\begin{proof}
The geometric envelopes of Theorem~\ref{ec:thm:comparable} give
deterministic numbers $0<l_\ast<l^\ast<\infty$ such that every original
optimum at $P$, whatever its reference, has $h_0\in[l_\ast,l^\ast]$. For
$\mathfrak s=+$ choose a reference below $l_\ast$, and for $\mathfrak s=-$ one above
$l^\ast$. The original
optimum $h$ for this reference revises strictly at time zero in the
direction $\mathfrak s$, so its initial support condition gives $\Psi_0=\iota^\mathfrak s$,
and it satisfies \eqref{ec:eq:foc} on $(0,\infty)$. We show that these
conditions make $h$ optimal for $\J^\mathfrak s_P$. Let $k$ be a competitor with
finite expenditure and put $\Delta v=u(k)-u(h)$. Concavity of
$v\mapsto v-P\Phi_1(v)$ gives
\[
 \J^\mathfrak s_P(k)-\J^\mathfrak s_P(h)
 \le\E\int_0^\infty e^{-\delta t}(1-P_th_t^\gamma)\Delta v_t\dd t
 -\big(\mathcal C_{>0}(u(k))-\mathcal C_{>0}(u(h))\big)-\iota^\mathfrak s\E\Delta v_0 ,
\]
where $\mathcal C_{>0}$ is the discounted adjustment cost on $(0,\infty)$.
The Fubini identity in the proof of Theorem~\ref{ec:thm:foc}, applied
conditionally on $\F_0$, writes the first term as
$\E[\Psi_0\Delta v_0]+\E\int_{(0,\infty)}e^{-\delta t}\Psi_t\dd\Delta v_t$, and
the obstacle and support conditions on $(0,\infty)$ bound the integral by
$\mathcal C_{>0}(u(k))-\mathcal C_{>0}(u(h))$. The initial terms cancel
because $\Psi_0=\iota^\mathfrak s$. Competitors with infinite expenditure cannot
improve the objective, and strict concavity in utility coordinates gives
uniqueness. The same verification shows that any original optimum with a
strict initial revision in the direction $\mathfrak s$ maximizes $\J^\mathfrak s_P$,
hence equals $h^\mathfrak s(P)$.

The lattice inequality of Lemma~\ref{ec:lem:lattice} holds for the
objectives $\J^\mathfrak s_P$, because the initial charge is linear and its
values at a maximum and a minimum of two paths add up to the original
values, so $h^\mathfrak s$ is decreasing in $P$. The substitution
$h=a^{-1/\gamma}k$ multiplies $\J^\mathfrak s_{aP}$ by $a^{(\gamma-1)/\gamma}$ for
$\gamma\ne1$ and shifts it by $-(\delta^{-1}-\iota^\mathfrak s)\log a$ for
$\gamma=1$, and the admissible set has no reference, so uniqueness gives
\eqref{ec:eq:relaxedscaling}. Price ordering and \eqref{ec:eq:relaxedscaling}
give \eqref{ec:eq:loglipschitz} as in Step~5 of the proof of
Theorem~\ref{ec:thm:comparable}. Finally, if $h^\mathfrak s(P)$ lies on the side
of $h_{0-}$ prescribed by $\mathfrak s$, its initial adjoint value is the initial
support condition of the original problem, or no initial support condition
is needed when the levels are equal, and Theorem~\ref{ec:thm:foc} identifies
it with $h^\ast(P,h_{0-})$.
\end{proof}

A type with $L_i=1$ has $\alpha_i=-\beta_i$, so $\iota^+_i=\iota^-_i=\alpha_i$,
and by \eqref{ec:eq:Ufrictionless} its objective \eqref{ec:eq:relaxedobjective}
is $(1-\delta_i\alpha_i)\E\int_0^\infty e^{-\delta_it}u_i(h_t)\dd t$ on
$\mathcal A_1$, with the unique response $((1-\delta_i\alpha_i)/P)^{1/\gamma_i}$
of Proposition~\ref{ec:prop:frictionless}. Such a type imposes no direction
and is included below with $a_i=b_i$.

\begin{theorem}\label{ec:thm:directions}
Assume the hypotheses of Theorem~\ref{thm:weakexist}, where types with
$L_i=1$ are allowed as in Proposition~\ref{ec:prop:frictionless}, and fix
$\mathfrak s\in\{+,-\}^n$. Replace the objective of type $i$ by $U_i^{\mathfrak s_i}$.
Then the relaxed population has exactly one equilibrium in the bounded
class for every positive wealth vector, obtained as follows. At
log multipliers $\mu=0$ let $y^{\mathfrak s,0}$ be the fixed point of the benchmark
map, let $h_i^{\mathfrak s,0}$ be the responses, and let
$\mathcal E_i^\mathfrak s(0)$ be their expenditures. With $\mathbf A$ as in
the proof of Theorem~\ref{thm:weakexist} put
\begin{equation}\label{ec:eq:signnormalization}
 \mu^\mathfrak s=\mathbf A^{-1}\big(\log\mathcal E^\mathfrak s(0)-\log X\big),\qquad
 y^\mathfrak s=y^{\mathfrak s,0}+k(\mu^\mathfrak s)-k(0),\qquad
 h_i^\mathfrak s=\exp\Big(-\frac{\mu_i^\mathfrak s+\theta_i(k(\mu^\mathfrak s)-k(0))}{\gamma_i}\Big)h_i^{\mathfrak s,0}.
\end{equation}
This relaxed equilibrium is an equilibrium of the original model whenever
$h^\mathfrak s_{i,0}\ge h_{i,0-}$ for every $i$ with $\mathfrak s_i=+$ and
$h^\mathfrak s_{i,0}\le h_{i,0-}$ for every $i$ with $\mathfrak s_i=-$, where types
with $L_i=1$ are exempt. Conversely, every equilibrium of the original model
in the bounded class in which every type with $L_i>1$ revises strictly at
time zero is one of these candidates. For a given vector of strict
directions there is at most one such equilibrium, including its
multipliers, and the number of equilibria in which all frictional types
revise strictly at time zero is at most $2^{n_f}$, where $n_f$ is the
number of types with $L_i>1$.
\end{theorem}

\begin{proof}
The relaxed responses are decreasing in the price and satisfy
\eqref{ec:eq:loglipschitz}, so the fixed multiplier contraction of
Theorem~\ref{thm:weakexist} holds. At the geometric center the relaxed
response is the reflected policy started at $e^{a_i}$ for $\mathfrak s_i=+$ and
at $e^{b_i}$ for $\mathfrak s_i=-$, so the same center, radius, invariant ball,
and expenditure envelopes apply, which gives the base fixed point and its
finite positive expenditures.

For arbitrary $\mu$ put $\varkappa=k(\mu)-k(0)=-D^{-1}\sum_i\pi_i\mu_i/\gamma_i$
and $\varrho_i=\exp(-(\mu_i+\theta_i\varkappa)/\gamma_i)$. Homogeneity
\eqref{ec:eq:relaxedscaling} gives
$h_i^{\mathfrak s_i}(e^{\mu_i}e^{\delta_it}\xi^ie^{\theta_i(y^{\mathfrak s,0}+\varkappa)})=\varrho_ih_i^{\mathfrak s,0}$,
and $\sum_i\pi_i\log\varrho_i=(1-\Theta_1)\varkappa$, so these responses
aggregate to $y^{\mathfrak s,0}+\varkappa$. Fixed multiplier uniqueness identifies them with the
relaxed fixed point at $\mu$. Their expenditures are
\[
 \mathcal E_i^\mathfrak s(\mu)=\mathcal E_i^\mathfrak s(0)
 \exp\big(-\mu_i/\gamma_i+\theta_i(1-1/\gamma_i)\varkappa\big)
 =\mathcal E_i^\mathfrak s(0)e^{-(\mathbf A\mu)_i},
\]
so all budgets bind exactly at the vector \eqref{ec:eq:signnormalization}.
The relaxed optima have bounded indices, and Theorem~\ref{ec:thm:foc}, or
its version with the linear initial charge, gives primal optimality at
binding budgets. This proves relaxed existence and uniqueness.

If the candidate satisfies the stated initial inequalities,
Lemma~\ref{ec:lem:relaxed} transfers individual optimality to the original
objectives, while budgets and aggregation are unchanged. Conversely, in an
original equilibrium with strict directions $\mathfrak s$, every frictional type
satisfies the initial adjoint equality of the relaxed problem, so by the
lemma its plan is the relaxed optimum at its own price, and the profile is a
relaxed equilibrium with these directions, hence the candidate. Here the
multipliers are the unique supporting multipliers of
Theorem~\ref{ec:thm:comparable}.
\end{proof}

For logarithmic utility the multipliers are explicit. With $\gamma_i=1$ the
index is $\lambda_ie^{\delta_it}\xi^i_tc^i_t$, and conditioning on the type
gives $\Psi_i(0)=\delta_i^{-1}-\lambda_iX_i$ at binding budgets, so the
obstacles give $(\delta_i^{-1}-\alpha_i)/X_i\le\lambda_i\le(\delta_i^{-1}+\beta_i)/X_i$,
with equality on the left for a strict upward initial revision and on the
right for a strict downward one.

The next corollary describes regions of inherited references on which the
initial directions are the same in every equilibrium.

\begin{corollary}\label{ec:cor:forced}
With $\tilde I_i=[\tilde I_i^-,\tilde I_i^+]$ as in
Section~\ref{ec:budgetunique} define
\begin{align*}
 \mathcal L_i&=\exp\Big(\frac{a_i-\theta_iR}{\gamma_i}-\frac{\theta_i\bar m}{\gamma_iD}
 -(1-\theta_i\pi_i)\tilde I_i^++\theta_i\sum_{j\ne i}\pi_j\tilde I_j^-\Big),\\
 \mathcal U_i&=\exp\Big(\frac{b_i+\theta_iR}{\gamma_i}-\frac{\theta_i\bar m}{\gamma_iD}
 -(1-\theta_i\pi_i)\tilde I_i^-+\theta_i\sum_{j\ne i}\pi_j\tilde I_j^+\Big).
\end{align*}
Under the hypotheses of Theorem~\ref{thm:weakexist}, if every type with
$L_i>1$ has $h_{i,0-}<\mathcal L_i$ or $h_{i,0-}>\mathcal U_i$, the mean field
equilibrium is unique in the bounded class, and it is the candidate of
Theorem~\ref{ec:thm:directions} with the corresponding directions. For
fixed parameters and wealth these conditions describe $2^{n_f}$ nonempty
open regions of references.
\end{corollary}

\begin{proof}
In every equilibrium in the class, the enclosure in the proof of
Theorem~\ref{ec:thm:budgetunique} gives $\tilde I_i^-\le(\mathbf A\mu)_i\le\tilde I_i^+$.
Writing $\jmath_i=(\mathbf A\mu)_i$ and $S=\sum_j\pi_j\jmath_j$, the inverse
of $\mathbf A$ gives $\mu_i=\gamma_i\jmath_i-\theta_i(\gamma_i-1)S$,
$k(\mu)=\bar m/D-S$, and
$(\mu_i+\theta_ik(\mu))/\gamma_i=\jmath_i-\theta_iS+\theta_i\bar m/(\gamma_iD)$.
The comparison with the two bounding geometric prices at time zero, where
$B_0=0$ and $|y_0-k(\mu)|\le R$, confines the initial level to
\[
 e^{(a_i-\mu_i-\theta_ik(\mu)-\theta_iR)/\gamma_i}
 \le h_{i,0}\le e^{(b_i-\mu_i-\theta_ik(\mu)+\theta_iR)/\gamma_i}.
\]
Minimizing and maximizing the affine function of $\jmath$ in the exponents
over the box gives $\mathcal L_i\le h_{i,0}\le\mathcal U_i$. A reference below
$\mathcal L_i$ therefore forces a strict upward initial revision, and one
above $\mathcal U_i$ a strict downward one, in every equilibrium. Existence
is Theorem~\ref{thm:weakexist}, and Theorem~\ref{ec:thm:directions} gives
uniqueness.
\end{proof}

Outside these regions, or when some type is inactive at time zero, the
arguments do not exclude several equilibria. The remaining question has an
exact finite dimensional form. For $\iota\in\prod_i[-\beta_i,\alpha_i]$
consider the relaxed objective with initial charge $-\iota_iu_i(h_{i,0})$.
The proof of Lemma~\ref{ec:lem:relaxed} extends from the endpoint slopes to
every $\iota_i$ in the interval, because the initial adjoint
$\Psi_i(0)=\delta_i^{-1}-\E\int_0^\infty e^{-\delta_it}P_t(h^\ast_t)^{\gamma_i}\dd t$
of the original optimum is continuous in the reference, equals $\alpha_i$
for small references and $-\beta_i$ for large ones, and so takes every
intermediate value, and the verification argument is unchanged. The
continuity holds because reference comparison and the scaling identity give
$h^\ast(P,h_{0-})\le h^\ast(P,ah_{0-})\le ah^\ast(P,h_{0-})$ for $a\ge1$, so the
optimum is continuous in the logarithm of the reference in uniform log
ratio, and the reference free index envelopes justify dominated
convergence. Homogeneity
and the explicit normalization of Theorem~\ref{ec:thm:directions} then give,
for every $\iota$, a unique relaxed equilibrium in the bounded class, with
initial relative consumption $h^\circ_i(\iota)>0$ that does not involve the
inherited references.

\begin{proposition}\label{ec:prop:complementarity}
Under the hypotheses of Theorem~\ref{thm:weakexist}, the equilibria of the
original model in the bounded class correspond one to one to the solutions
$\iota\in\prod_i[-\beta_i,\alpha_i]$ of the complementarity system
\[
 h^\circ_i(\iota)>h_{i,0-}\ \Rightarrow\ \iota_i=\alpha_i,\qquad
 h^\circ_i(\iota)<h_{i,0-}\ \Rightarrow\ \iota_i=-\beta_i,\qquad
 -\beta_i<\iota_i<\alpha_i\ \Rightarrow\ h^\circ_i(\iota)=h_{i,0-}.
\]
\end{proposition}

\begin{proof}
In an original equilibrium set $\iota_i=\Psi_i(0)$. The initial obstacle
and support conditions are exactly the three implications, and the
verification with a free initial level makes the profile the unique relaxed
equilibrium at this $\iota$, so its initial level is $h^\circ_i(\iota)$.
Conversely, a solution of the system supplies the initial support
conditions, and the remaining support conditions, budgets, and aggregation
already hold for the relaxed equilibrium, so Theorem~\ref{ec:thm:foc} gives
an original equilibrium. A fixed original equilibrium has unique
multipliers and hence unique initial adjoints, so the correspondence is one
to one.
\end{proof}

Uniqueness of the solution of this system is what general uniqueness
requires. The estimates above do not establish the monotonicity of
$\iota\mapsto(h^\circ_i(\iota))_i$ that would give it.

%% file: ec4c_vanishing.tex
\subsection{Vanishing Loss Aversion}\label{ec:vanishing}

This subsection describes how equilibria behave as the coefficients of
loss aversion tend to one. It uses the normalization
\eqref{ec:eq:normalized} of Section~\ref{ec:sync}.

\begin{lemma}\label{ec:lem:smallband}
Fix $r\in\R$ and $\kappa,\delta>0$, and let $d=\log(\zb/\za)$ be the width
of the band of Lemma~\ref{lem:band}. As $L\downarrow1$,
\begin{equation}\label{ec:eq:smallband}
 d=\Big(\frac{6\kappa^2}{\delta}\Big)^{1/3}(L-1)^{1/3}+O\big((L-1)^{2/3}\big),\qquad
 \log\frac{\za}{A}=-\frac d2+O(d^2),\qquad
 \log\frac{\zb}{A}=\frac d2+O(d^2),
\end{equation}
with constants that do not depend on $A$.
\end{lemma}

\begin{proof}
Use the Neumann problem \eqref{ec:eq:neumannwidth} of the proof of
Theorem~\ref{ec:thm:allrates}, where $L=f_d(d)/f_d(0)$ and
$\za/A=1/(\delta f_d(0))$. Set $\tilde g_d(x)=d^{-2}f_d'(dx)$ for
$0\le x\le1$. Differentiating \eqref{ec:eq:neumannwidth} shows that
$\varsigma\tilde g_d''+\nu d\tilde g_d'-\delta d^2\tilde g_d=-e^{dx}$ with
$\tilde g_d(0)=\tilde g_d(1)=0$. The Dirichlet operator $\varsigma\dd^2/\dd x^2$
has a bounded inverse from $C([0,1])$ to the $C^2$ functions vanishing at
the endpoints, and a Neumann series for small $d$ gives
$\tilde g_d(x)=x(1-x)/(2\varsigma)+O(d)$ in $C^2([0,1])$. Hence
$f_d(d)-f_d(0)=d^3\int_0^1\tilde g_d(x)\dd x=d^3/(12\varsigma)+O(d^4)$ and
$\sup_{s\le d}|f_d(s)-f_d(0)|=O(d^3)$. Multiplying
\eqref{ec:eq:neumannwidth} by $e^{\nu s/\varsigma}$ and integrating over
$[0,d]$ gives $\delta\int_0^de^{\nu s/\varsigma}f_d(s)\dd s=\int_0^de^{(1+\nu/\varsigma)s}\dd s$,
and the oscillation bound turns this into
$\delta f_d(0)=1+d/2+O(d^2)$. Therefore $L-1=\delta d^3/(12\varsigma)+O(d^4)$,
and since $L$ increases strictly from one to infinity with $d$, $L\downarrow1$
forces $d\downarrow0$. Inverting the expansion gives the first claim, and
$\log(\za/A)=-\log(\delta f_d(0))$ gives the other two.
\end{proof}

Now fix a finite type distribution, all market parameters,
$(\gamma_i,\delta_i,\theta_i)$, wealth $X_i>0$, and references
$h_{i,0-}>0$, and let only the adjustment parameters vary. The trend
coefficients $g$ and $s$, the artificial markets, and $K_i=\hat K_i(\gamma_i)$
do not involve $(\alpha_i,\beta_i)$. Define
\begin{gather}
 k_F=\sum_i\pi_i\log(X_iK_i),\qquad Y^F_t=\exp(k_F+gt+sB_t),\label{ec:eq:frictionlessbenchmark}\\
 g_i^F=-\frac{\nu_i}{\gamma_i}+\theta_i(1-\gamma_i^{-1})g,\qquad
 s_i^F=\frac{\kappa_{a,i}}{\gamma_i}+\theta_i(1-\gamma_i^{-1})s,\label{ec:eq:frictionlessloads}\\
 c_i^F(t)=X_iK_i\exp\Big(g_i^Ft+s_i^FB_t+\frac{\kappa_{p,i}}{\gamma_i}W_t\Big),\qquad
 h_i^F=c_i^F(Y^F)^{-\theta_i},\label{ec:eq:frictionlessplan}\\
 \mu_i^F=-\gamma_i\log(X_iK_i)+\theta_i(\gamma_i-1)k_F .\label{ec:eq:frictionlessmultiplier}
\end{gather}
The definitions of $D$, $g$, and $s$ give $\sum_i\pi_ig_i^F=g$ and
$\sum_i\pi_is_i^F=s$, and a Gaussian computation with \eqref{eq:hat} gives
\begin{equation}\label{ec:eq:frictionlesscost}
 \E\Big[\xi^i_t\exp\Big(g_i^Ft+s_i^FB_t+\frac{\kappa_{p,i}}{\gamma_i}W_t\Big)\Big]=e^{-K_it}.
\end{equation}

\begin{proposition}\label{ec:prop:allfrictionless}
Let every type have $L_i=1$ and $\kappa_i>0$, and let households choose
from $\mathcal A_1$. A mean field equilibrium exists if and only if $K_i>0$
for every type. It is then unique up to $\dd t\otimes\mathbb P$ null sets,
and up to indistinguishability among c\`adl\`ag versions, and it is given
by \eqref{ec:eq:frictionlessbenchmark} to \eqref{ec:eq:frictionlessmultiplier},
with the log multiplier $\mu_i^F$ of the normalized problem
\eqref{ec:eq:normalized} for type $i$. Neither $\Theta_1+\Theta_2<1$ nor
$\hat\kappa_i>0$ is needed.
\end{proposition}

\begin{proof}
Since $\Theta_1<1$ and $\Theta_2\ge0$, $D>0$ and $g$, $s$ are well defined.
Normalize utility by \eqref{ec:eq:normalized}. Consider any equilibrium,
fix its benchmark and one type, suppress the type index, and write
$P^\circ_t=\xi_tY_t^\theta$. The budget binds, since otherwise a small constant
upward scaling of $h$ stays in $\mathcal A_1$, remains affordable, and
raises utility. Put $M=\E\int_0^\infty e^{-\delta t}h_t^{1-\gamma}\dd t$,
which equals $1/\delta$ for $\gamma=1$ and $|1-\gamma|$ times the finite
absolute utility integral otherwise. For bounded progressive $f$ with
$\E\int P^\circ_th_tf_t\dd t=0$, the plans $h(1+\epsilon f)$ are positive, lie in
$\mathcal A_1$, and have the same budget for small $|\epsilon|$, and
dominated differentiation of their utility at $\epsilon=0$ gives
$\E\int_0^\infty e^{-\delta t}h_t^{1-\gamma}f_t\dd t=0$. For arbitrary
bounded progressive $f$ apply this to $f-X^{-1}\E\int P^\circ hf\dd t$. The
two finite measures on the progressive sigma field with densities
$e^{-\delta t}h^{1-\gamma}$ and $P^\circ h$ are therefore proportional with
ratio $M/X$, so that $h_t^{-\gamma}=(M/X)e^{\delta t}P^\circ_t$ almost
everywhere. Every
equilibrium therefore has a positive multiplier $e^{\mu_i}$ for every type.

Taking logarithms, averaging conditionally on $\F^B_t$, and imposing
\eqref{eq:consistency} yields
$D\log Y_t=-\sum_i\pi_i\mu_i/\gamma_i-t\sum_i\pi_i\nu_i/\gamma_i+B_t\sum_i\pi_i\kappa_{a,i}/\gamma_i$
almost everywhere, so the benchmark is geometric with the coefficients $g$
and $s$. With intercept $k$ the consumption of type $i$ is
$\varkappa_i\exp(g_i^Ft+s_i^FB_t+\kappa_{p,i}W_t/\gamma_i)$ with
$\varkappa_i=\exp(-\mu_i/\gamma_i+\theta_i(1-1/\gamma_i)k)$. By
\eqref{ec:eq:frictionlesscost} its cost is $\varkappa_i/K_i$ if $K_i>0$ and
infinite otherwise, so positivity of every $K_i$ is necessary. If all
$K_i>0$, binding budgets force $\varkappa_i=X_iK_i$, consistency at time
zero forces $k=\sum_i\pi_i\log\varkappa_i=k_F$, and the first order
conditions give $\mu_i=\mu_i^F$. Conversely these choices have finite
expenditure, satisfy all budgets and \eqref{eq:consistency}, and have
finite absolute discounted utility because
$e^{-\delta_it}(h_i^F)^{1-\gamma_i}=e^{\mu_i^F}\xi^i_tc_i^F(t)$ when
$\gamma_i\ne1$, while the logarithmic case follows from Gaussian moments.
Pointwise strict concavity gives optimality against every competitor in
$\mathcal A_1$, which proves sufficiency and uniqueness.
\end{proof}

For a family of adjustment parameters indexed by $\varepsilon$ let
$L_i^\varepsilon\ge1$ and $A_i^\varepsilon=1-\delta_i\alpha_i^\varepsilon$, and
let $d_i^\varepsilon$, $R^\varepsilon$, and the band endpoints be the
quantities of Theorem~\ref{thm:weakexist}, with degenerate bands for types
with $L_i^\varepsilon=1$. Put
\begin{equation}\label{ec:eq:frictionlesserror}
 E_i^\varepsilon=\frac{d_i^\varepsilon}{2\gamma_i}+\theta_i(1+\gamma_i^{-1})R^\varepsilon .
\end{equation}

\begin{theorem}\label{ec:thm:vanishing}
Assume $\hat\kappa_i>0$, $K_i>0$ for every type, and $\Theta_1+\Theta_2<1$,
and suppose $L_i^\varepsilon\to1$ for every type as $\varepsilon\to0$. For
each $\varepsilon$ choose any equilibrium furnished by
Theorem~\ref{thm:weakexist} and Proposition~\ref{ec:prop:frictionless},
with benchmark $Y^\varepsilon$, consumption $c_i^\varepsilon$, log multipliers
$\mu_i^\varepsilon$, and center $k^\varepsilon=k(\mu^\varepsilon)$, all on
the same Brownian space. Then the equilibria converge to the frictionless
equilibrium of Proposition~\ref{ec:prop:allfrictionless} in the following
sense.
\begin{enumerate}[label=(\roman*),nosep,leftmargin=2em]
\item $|k^\varepsilon-k_F|\le R^\varepsilon$ and, almost surely,
\[
 \sup_{t\ge0}\Big|\log\frac{Y^\varepsilon_t}{Y^F_t}\Big|\le2R^\varepsilon,\qquad
 \sup_{t\ge0}\Big|\log\frac{c_i^\varepsilon(t)}{c_i^F(t)}\Big|\le2E_i^\varepsilon ,
\]
so the whole family converges uniformly in time in these logarithmic
norms, whichever multipliers were selected, with error
$O(\max_i(L_i^\varepsilon-1)^{1/3})$.
\item The normalized multipliers satisfy
$|\mu_i^\varepsilon-\log A_i^\varepsilon-\mu_i^F|\le d_i^\varepsilon+2\theta_i\max\{1,\gamma_i\}R^\varepsilon$.
\item The wealth processes that finance the consumption streams without
terminal surplus satisfy $X^F_{i,t}=c_i^F(t)/K_i$ and
$\sup_{t\ge0}|\log(X^\varepsilon_{i,t}/X^F_{i,t})|\le2E_i^\varepsilon$ almost
surely.
\item The equilibrium utilities satisfy
\[
 \frac{U_i(h_i^\varepsilon)-\alpha_i^\varepsilon u_i(h_{i,0-})}{A_i^\varepsilon}
 \longrightarrow\E\int_0^\infty e^{-\delta_it}u_i(h_i^F(t))\dd t .
\]
\end{enumerate}
\end{theorem}

\begin{proof}
\emph{Uniform bounds.} The existence construction gives
$|\log Y^\varepsilon_t-k^\varepsilon-gt-sB_t|\le R^\varepsilon$. The price of
type $i$ lies between $e^{\pm\theta_iR^\varepsilon}$ times its artificial
geometric price, and comparison with the reflected policies at these two
geometric prices, whose log indices lie in $[a_i^\varepsilon,b_i^\varepsilon]$,
gives
\[
 \log h_i^\varepsilon(t)=\frac{m_i^\varepsilon-\mu_i^\varepsilon-\theta_ik^\varepsilon}{\gamma_i}
 -\frac{\nu_i+\theta_ig}{\gamma_i}t+\frac{\kappa_{a,i}-\theta_is}{\gamma_i}B_t
 +\frac{\kappa_{p,i}}{\gamma_i}W_t+\eta_i^\varepsilon(t),\qquad
 |\eta_i^\varepsilon(t)|\le\frac{d_i^\varepsilon}{2\gamma_i}+\frac{\theta_iR^\varepsilon}{\gamma_i},
\]
where $m_i^\varepsilon$ is the midpoint of the band. Multiplying by
$(Y^\varepsilon)^{\theta_i}$ yields
\begin{equation}\label{ec:eq:consumptionenvelope}
 c_i^\varepsilon(t)=\varkappa_i^\varepsilon
 \exp\Big(g_i^Ft+s_i^FB_t+\frac{\kappa_{p,i}}{\gamma_i}W_t+e_i^\varepsilon(t)\Big),\qquad
 |e_i^\varepsilon(t)|\le E_i^\varepsilon,\qquad
 \log\varkappa_i^\varepsilon=\frac{m_i^\varepsilon-\mu_i^\varepsilon}{\gamma_i}+\theta_i(1-\gamma_i^{-1})k^\varepsilon .
\end{equation}
By \eqref{ec:eq:frictionlesscost} and the binding budget,
$|\log\varkappa_i^\varepsilon-\log(X_iK_i)|\le E_i^\varepsilon$. The definitions
of $k(\mu)$ and $D$ give $\sum_i\pi_i\log\varkappa_i^\varepsilon=k^\varepsilon$,
and
$\sum_i\pi_iE_i^\varepsilon=\sum_i\pi_id_i^\varepsilon/(2\gamma_i)+(\Theta_1+\Theta_2)R^\varepsilon=R^\varepsilon$
by \eqref{eq:existconstants}. Averaging the bound on $\log\varkappa_i^\varepsilon$
proves $|k^\varepsilon-k_F|\le R^\varepsilon$, and combining it with the bound
on the benchmark and with \eqref{ec:eq:consumptionenvelope} proves the two
uniform bounds of (i). The artificial markets do not depend on
$\varepsilon$, so Lemma~\ref{ec:lem:smallband} in each of them gives
$d_i^\varepsilon=O((L_i^\varepsilon-1)^{1/3})$, hence the rate.

\emph{Multipliers and wealth.} By Theorem~\ref{ec:thm:allrates},
$\log A_i^\varepsilon\in[a_i^\varepsilon,b_i^\varepsilon]$, so
$|m_i^\varepsilon-\log A_i^\varepsilon|\le d_i^\varepsilon/2$. Solving the
definition of $\varkappa_i^\varepsilon$ for $\mu_i^\varepsilon$ and subtracting
$\log A_i^\varepsilon+\mu_i^F$ gives a bound of
$d_i^\varepsilon/2+\gamma_iE_i^\varepsilon+\theta_i|\gamma_i-1|R^\varepsilon
=d_i^\varepsilon+2\theta_i\max\{1,\gamma_i\}R^\varepsilon$, which is (ii). The
wealth without terminal surplus is
$X_{i,t}=(\xi^i_t)^{-1}\E[\int_t^\infty\xi^i_sc_i(s)\dd s\,|\,\F_t]$, and
conditional Gaussian integration gives $X^F_{i,t}=c_i^F(t)/K_i$. The
multiplicative bounds on consumption pass through the conditional
integral, which proves (iii) at every fixed time, and continuity of the
wealth processes gives one null set for all times.

\emph{Utility.} The bounds give
$\sup_t|\log(h_i^\varepsilon(t)/h_i^F(t))|\le2E_i^\varepsilon+2\theta_iR^\varepsilon\to0$,
and $\E\int e^{-\delta_it}|u_i(h_i^F(t))|\dd t<\infty$ by the proof of
Proposition~\ref{ec:prop:allfrictionless}. For power utility the
deterministic multiplicative bounds dominate $|u_i(h_i^\varepsilon)|$ by a
fixed multiple of $|u_i(h_i^F)|$, and for logarithmic utility
$|\log h_i^\varepsilon|\le|\log h_i^F|+2E_i^\varepsilon+2\theta_iR^\varepsilon$,
so the running utility integrals converge in $L^1$. It remains to show
that the variation penalty in \eqref{ec:eq:normalized} vanishes, since
convergence of consumption alone does not control variation. In the
normalized problem of type $i$ the price is
$(\lambda_i^\varepsilon/A_i^\varepsilon)e^{\delta_it}\xi^i_t(Y^\varepsilon_t)^{\theta_i}$
and the index $\tilde z_i^\varepsilon$ is $z_i^\varepsilon/A_i^\varepsilon$, so
\eqref{ec:eq:comparableband} and the normalized band
$[e^{a_i^\varepsilon}/A_i^\varepsilon,e^{b_i^\varepsilon}/A_i^\varepsilon]$, which
contains one, give
$|\log\tilde z_i^\varepsilon|\le d_i^\varepsilon+2\theta_iR^\varepsilon=:\bar q_i^\varepsilon$.
Put $v_i^\varepsilon=u_i(h_i^\varepsilon)$ and
$\beta_i'=(L_i^\varepsilon-1)/\delta_i$. For $L_i^\varepsilon>1$ the normalized
adjoint $\tilde\Psi_i^\varepsilon$ lies in $[-\beta_i',0]$ and satisfies the
support conditions of the normalized problem, so the Fubini identity of
Theorem~\ref{ec:thm:foc} applied to $v_i^\varepsilon-v_{i,0-}$ gives
\[
 \beta_i'\,\E\int_{[0,\infty)}e^{-\delta_it}\dd(v_i^\varepsilon)^-_t
 =\E\int_0^\infty e^{-\delta_it}\big(1-\tilde z_i^\varepsilon(t)\big)v_i^\varepsilon(t)\dd t
 -v_{i,0-}\tilde\Psi_i^\varepsilon(0)
 \le\big(e^{\bar q_i^\varepsilon}-1\big)\E\int_0^\infty e^{-\delta_it}|v_i^\varepsilon(t)|\dd t
 +\beta_i'|v_{i,0-}|\longrightarrow0 .
\]
For $L_i^\varepsilon=1$ the penalty is zero. Identity \eqref{ec:eq:normalized}
now proves (iv).
\end{proof}

The theorem does not assert convergence in total variation, since a
continuous limit can be approached uniformly by plans of finite variation
whose variations diverge, and the wealth bounds do not by themselves give
convergence of the portfolio integrands. If only some types lose their
adjustment costs, $R^\varepsilon$ need not vanish and the argument does not
identify a limit.

Combined with Theorem~\ref{ec:thm:directions}, the quantitative limit gives
uniqueness for populations with small adjustment costs. A condition on the
inherited standards is essential, because convergence to a unique limit
does not by itself exclude several equilibria at positive costs. With the
notation above let
\begin{equation}\label{ec:eq:frictionlessgap}
 h_i^F(0)=X_iK_ie^{-\theta_ik_F},\qquad
 \Delta_i=\log\frac{h_i^F(0)}{h_{i,0-}},\qquad
 \varrho_i^\varepsilon=2E_i^\varepsilon+2\theta_iR^\varepsilon .
\end{equation}

\begin{theorem}\label{ec:thm:nearlimit}
Assume $\hat\kappa_i>0$, $K_i>0$ for every type, and $\Theta_1+\Theta_2<1$.
If $|\Delta_i|>\varrho_i^\varepsilon$ for every type with $L_i^\varepsilon>1$, there
is exactly one mean field equilibrium in the bounded class, and in it every
frictional type revises strictly at time zero, upward if $\Delta_i>0$ and
downward if $\Delta_i<0$. In particular, if $\Delta_i\ne0$ for every type,
there is $\varepsilon_0>0$ such that this uniqueness holds whenever
$1\le L_i^\varepsilon<1+\varepsilon_0$ for all types.
\end{theorem}

\begin{proof}
Existence is Theorem~\ref{thm:weakexist} with Proposition~\ref{ec:prop:frictionless}.
The bounds of Theorem~\ref{ec:thm:vanishing} apply to every equilibrium in
the bounded class and not only to the one constructed there, because the
enclosure in the proof of Theorem~\ref{ec:thm:budgetunique} shows that any
such equilibrium satisfies the radius bound at its own multipliers, and the
bounds only use that radius, the envelopes, and binding budgets. Hence
$\sup_t|\log(h_i(t)/h_i^F(t))|\le \varrho_i^\varepsilon$, and at time zero
$\log(h_{i,0}/h_{i,0-})\in[\Delta_i-\varrho_i^\varepsilon,\Delta_i+\varrho_i^\varepsilon]$.
Under the hypothesis each frictional type's interval lies strictly on one
side of zero, so every equilibrium has the same vector of strict initial
directions, and Theorem~\ref{ec:thm:directions} gives at most one. Finally
$\varrho_i^\varepsilon\to0$ as $\max_j(L_j^\varepsilon-1)\to0$ by
Lemma~\ref{ec:lem:smallband}, which gives the last claim.
\end{proof}

\begin{corollary}\label{ec:cor:localinvariance}
Under the strict inequalities of Theorem~\ref{ec:thm:nearlimit}, the
equilibrium consumption, benchmark, wealth, and multipliers do not change
when the inherited standards are replaced by sufficiently close positive
numbers. Only the utility levels change, through the initial adjustment
costs.
\end{corollary}

\begin{proof}
Every frictional type revises strictly at time zero, and a small change of
its reference preserves the direction. The plan, its price, and its adjoint
are unchanged, the adjoint already equals the required obstacle at time
zero, and all support conditions after zero are unchanged, so
Theorem~\ref{ec:thm:foc} gives optimality for the new reference, while
budgets and consistency are unchanged. The inequalities of
Theorem~\ref{ec:thm:nearlimit} persist, so the profile remains the unique
equilibrium.
\end{proof}

The condition $\Delta_i\ne0$ excludes references that coincide with the
frictionless initial optimum. The theorem is not a uniqueness result for
arbitrary adjustment costs, since equilibria with initially inactive types
require a separate argument.

%% file: ec4e_oneconcern.tex
\subsection{One Benchmark Sensitive Type}\label{ec:oneconcern}

When only one type is concerned with the benchmark, the other types form
an exogenous comparison group, the outer problem reduces to one individual
problem, and existence and global uniqueness hold without any weak
interaction condition.

\begin{theorem}\label{ec:thm:oneconcern}
Let risk be common, let all types trade at
$\xi_t=\exp(-(r+\kappa^2/2)t-\kappa B_t)$, and let there be finitely many
types with weights $\pi_i>0$. Suppose $\theta_j=0$ for every $j\ne a$ and
$0\le\theta_a<1$, with the parameter restrictions of Section~\ref{sec:setup}
and $K_j(\gamma_j)>0$ for every $j\ne a$, while the other parameters may
differ across types. Let $c^j$ be the unique optimal plan of type $j\ne a$,
which does not involve the benchmark, and put
\[
 Y^{-a}_t=\prod_{j\ne a}(c^j_t)^{\pi_j},\qquad p=1-\pi_a\theta_a,\qquad
 \rho=1+p(\gamma_a-1),\qquad\vartheta=\frac{\theta_a}{p},
\]
\[
 g^{-a}=-\sum_{j\ne a}\frac{\pi_j\nu_j}{\gamma_j},\qquad
 s^{-a}=\kappa\sum_{j\ne a}\frac{\pi_j}{\gamma_j},\qquad
 \kappa^{-a}=\kappa-\vartheta s^{-a},\qquad
 r^{-a}=r-\vartheta(g^{-a}-\kappa s^{-a})-\tfrac12\vartheta^2(s^{-a})^2 ,
\]
with $\nu_j=\delta_j-r-\kappa^2/2$. Assume $\kappa^{-a}\ne0$ and
$K^{-a}(\rho)=r^{-a}+(\delta_a-r^{-a})/\rho+(\rho-1)(\kappa^{-a})^2/(2\rho^2)>0$.
Then for every positive wealth vector and reference vector there is
exactly one mean field equilibrium, and its multipliers are unique. It is
obtained from the unique optimum $\breve h$ of the ordinary ILA problem
with curvature $\rho$, parameters $(\delta_a,\alpha_a,\beta_a)$, reference
$h_{a,0-}^{1/p}$, wealth $X_a$, and expenditure deflator
$\xi_t(Y^{-a}_t)^{\vartheta}$ by
\[
 c^a_t=(Y^{-a}_t)^{\vartheta}\breve h_t,\qquad h^a_t=\breve h_t^{\,p},\qquad
 Y_t=(c^a_t)^{\pi_a}Y^{-a}_t .
\]
\end{theorem}

\begin{proof}
The benchmark does not enter the problems of the types $j\ne a$, whose
plans and multipliers are unique by Theorem~\ref{ec:thm:comparable} applied
to geometric prices. Their indices lie in deterministic bands and
$\gamma_j\log c^j_t=\log z_j(t)-\log\lambda_j-\nu_jt+\kappa B_t$, so
$\log Y^{-a}_t-g^{-a}t-s^{-a}B_t$ is bounded by deterministic constants. The
deflator $\xi(Y^{-a})^{\vartheta}$ is therefore bounded above and below by
constant multiples of $\exp(-(r^{-a}+(\kappa^{-a})^2/2)t-\kappa^{-a}B_t)$,
and Theorem~\ref{ec:thm:comparable} in the market $(r^{-a},|\kappa^{-a}|)$
gives the unique optimum $\breve h$ of the ordinary problem with curvature
$\rho$, its unique multiplier $\lambda_a$, a bounded index
$z_t=\lambda_ae^{\delta_at}\xi_t(Y^{-a}_t)^{\vartheta}\breve h_t^{\,\rho}$,
and the conditions \eqref{ec:eq:foc} for its adjoint.

For power and for logarithmic utility,
$u_{\gamma_a}(\breve h^{\,p})=p\,u_\rho(\breve h)$ holds exactly, including the
two signed variation measures and the initial jump, because
$1-\rho=p(1-\gamma_a)$. Hence $h^a=\breve h^{\,p}$ is admissible, has the
reference $h_{a,0-}$, and revises in the same directions as $\breve h$. The
definition of $Y$ and $p=1-\pi_a\theta_a$ give
$c^aY^{-\theta_a}=(c^a)^p(Y^{-a})^{-\theta_a}=\breve h^{\,p}=h^a$, so the
profile satisfies \eqref{eq:consistency}. The index of type $a$ at the
benchmark $Y$ is
$\lambda_ae^{\delta_at}\xi_tY_t^{\theta_a}(h^a_t)^{\gamma_a}
=\lambda_ae^{\delta_at}\xi_tc^a_t(h^a_t)^{\gamma_a-1}=z_t$, so its adjoint and
support conditions are those already verified for $\breve h$, its
Lagrangian is finite, and its budget binds because
$\E\int\xi_tc^a_t\dd t=\E\int\xi_t(Y^{-a}_t)^{\vartheta}\breve h_t\dd t=X_a$.
Theorem~\ref{ec:thm:foc} makes $h^a$ the unique best response, and the
other types are best responses by construction. This proves existence.

Conversely, let $(Y,h)$ be any mean field equilibrium. The plans of the
types $j\ne a$ are the same $c^j$, so $Y^{-a}$ is the process above, and
consistency gives $h^a=(c^a)^p(Y^{-a})^{-\theta_a}$. Put $\tilde h=(h^a)^{1/p}$,
so that $c^a=(Y^{-a})^{\vartheta}\tilde h$. The utility identity makes
$\tilde h$ admissible for the ordinary problem with curvature $\rho$, and
its expenditure is at most $X_a$. If $\tilde h\ne\breve h$, uniqueness of
the ordinary optimum gives $pU_\rho(\breve h)>pU_\rho(\tilde h)=U_{\gamma_a}(h^a)$, where $U_\rho$ denotes the
utility functional with curvature $\rho$. Hold $Y$ fixed and let type $a$
deviate to $\check c^a=(c^a)^{1-p}(\bar c^a)^p$ with
$\bar c^a=(Y^{-a})^{\vartheta}\breve h$. H\"older's inequality for the measure
$\xi_t\dd t\,\mathbb P(\dd\omega)$ gives
$\E\int\xi_t\check c^a_t\dd t\le(\E\int\xi_tc^a_t\dd t)^{1-p}(\E\int\xi_t\bar c^a_t\dd t)^p\le X_a$,
and for $p=1$ the deviation is $\bar c^a$ itself. Completeness makes it
financeable, and since $Y=(c^a)^{\pi_a}Y^{-a}$ and $1-p=\pi_a\theta_a$, its
relative consumption is
$\check c^aY^{-\theta_a}=(\bar c^a)^p(Y^{-a})^{-\theta_a}=\breve h^{\,p}$, an
admissible plan with the correct reference and utility $pU_\rho(\breve h)$,
a strict improvement. This contradicts equilibrium, so $\tilde h=\breve h$
and all plans and the benchmark coincide. Finally, any positive multiplier that
supports $h^a$ at the benchmark $Y$ makes $\breve h$ the maximizer of the
Lagrangian of the ordinary problem at the corresponding price, by the index
identity, and the multiplier of the ordinary problem is unique by
Theorem~\ref{ec:thm:comparable}.
\end{proof}

The parameter region extends beyond weak interaction. For two types with
$\pi_a=3/4$, $\theta_a=4/5$, $\gamma_a=1$, $\theta_b=0$, $\gamma_b=2$,
$r=0.03$, $\kappa=0.2$, and $\delta=0.04$ one has
$\Theta_1+\Theta_2=1.2$, so that the contraction factor of
Theorem~\ref{thm:weakexist} would be $3/2$, whereas $p=0.4$, $\vartheta=2$,
$\rho=1$, $\kappa^{-a}=0.15$, $r^{-a}=0.03625$, $K^{-a}(1)=0.04$, and
$K_b(2)=0.04$, so the theorem covers all admissible adjustment
parameters, wealth levels, and references. The restriction to one benchmark
sensitive type is essential to the reduction. It does extend to a block
$\mathcal Z$ of benchmark sensitive types with identical
$(\gamma,\delta,\alpha,\beta,\theta)$ whose initial data are proportional,
$X_i=\varrho_iX_{\mathcal Z}$ and $h_{i,0-}=\varrho_ih_{\mathcal Z,0-}$ with
$\varrho_i>0$ for $i\in\mathcal Z$. Scaling a plan and its reference by
$\varrho_i$ scales its budget by $\varrho_i$ and its utility by
$\varrho_i^{1-\gamma}$, or shifts it by $(\log\varrho_i)/\delta$, so at any
benchmark the problems of the types in $\mathcal Z$ are equivalent. With
$\pi_{\mathcal Z}=\sum_{i\in\mathcal Z}\pi_i$ and
$Y^{-\mathcal Z}=(\prod_{i\in\mathcal Z}\varrho_i^{\pi_i})\prod_{j\notin\mathcal Z}(c^j)^{\pi_j}$,
the profile built from the theorem with $\pi_{\mathcal Z}$ in place of
$\pi_a$, whose constant in $Y^{-\mathcal Z}$ changes neither $g^{-\mathcal Z}$
nor $s^{-\mathcal Z}$, together with $c^i=\varrho_ic^{\mathcal Z}$, is an
equilibrium, and its benchmark lies in the bounded class. It is the only
equilibrium in that class, because there the price of every household in
the block is comparable to a geometric price, Theorem~\ref{ec:thm:comparable}
makes its budget constrained optimum unique, so that $h^i=\varrho_ih^{\mathcal Z}$
in every such equilibrium, and the consistency equation reduces to
$Y=(c^{\mathcal Z})^{\pi_{\mathcal Z}}Y^{-\mathcal Z}$, to which the uniqueness
argument of the theorem applies.

%% file: ec4f_infinite.tex
\subsection{Infinitely Many Types}\label{ec:infinitetypes}

Theorem~\ref{thm:weakexist} applies Brouwer's theorem to finitely many
multipliers. The restriction to finitely many types can be removed in two
directions: countably many parameter types under summability conditions,
and arbitrary distributions of initial wealth and inherited standards
under finitely many parameter regimes.

\begin{theorem}\label{ec:thm:countable}
Let the type set be $\mathbb N$ with probabilities $\pi_i>0$, let every
type satisfy the restrictions of Section~\ref{sec:setup} with $\kappa_i>0$,
and define $\Theta_1$, $\Theta_2$, $D$, $g$, and $s$ by the formulas of
Section~\ref{sec:exist} with infinite sums. Suppose
\begin{equation}\label{ec:eq:countablesum}
 \Theta_1+\Theta_2<1,\qquad
 \sum_i\frac{\pi_i}{\gamma_i}\big(|\nu_i|+\kappa_{a,i}+\kappa_{p,i}\big)<\infty,\qquad
 \sum_i\pi_i\Big(\frac{|a_i|+|b_i|}{\gamma_i}+|\log\hat K_i(\gamma_i)|+|\log X_i|\Big)<\infty,
\end{equation}
where $\hat\kappa_i>0$ and $\hat K_i(\gamma_i)>0$ in every artificial
market. Then for every sequence of references $h_{i,0-}>0$ a mean field
equilibrium exists in which every budget binds and
$\sum_i\pi_i|\mu_i|/\gamma_i<\infty$. Its benchmark satisfies
\eqref{eq:existbound} with the infinite sum definitions of $k(\mu)$ and
$R$. For fixed multipliers with $\sum_i\pi_i|\mu_i|/\gamma_i<\infty$ the
benchmark is unique in the bounded class and best response iteration
converges at the rate $q=\Theta_2/(1-\Theta_1)$.
\end{theorem}

\begin{proof}
All trend coefficients and $R=\sum_i\pi_id_i/(2\gamma_i(1-\Theta_1-\Theta_2))$
are finite by \eqref{ec:eq:countablesum}. Fix $\mu$ with
$\sum_i\pi_i|\mu_i|/\gamma_i<\infty$ and consider the benchmark map
$\mathcal T_\mu$ of Section~\ref{ec:proofexist}. If
$|y_t-k-gt-sB_t|\le M$, comparison with the two bounding geometric prices
gives
\[
 |\log h_i[y]_t|\le\frac1{\gamma_i}\Big(|a_i|+|b_i|+|\mu_i|+\theta_i(|k|+M)
 +|\nu_i+\theta_ig|t+|\kappa_{a,i}-\theta_is||B_t|+\kappa_{p,i}|W_t|\Big),
\]
and the sum over $i$ of the conditional expectations of the right side is
finite, uniformly on compact time intervals, almost surely, since
$\E[|W_t|\,|\,\F^B_t]=\E|W_t|$ and every term is summable by
\eqref{ec:eq:countablesum}. The series defining $\mathcal T_\mu y$ therefore
converges uniformly on compacts to an adapted c\`adl\`ag process, and the
argument of Step~1 applies to each summand. Summing the individual
estimates \eqref{ec:eq:loglipschitz} gives the contraction with factor $q$,
the summable index identity at $y^0=k(\mu)+gt+sB_t$ gives
$\|\mathcal T_\mu y^0-y^0\|_\infty\le(1-q)R$, and Banach's theorem gives the
unique fixed point $y^\mu$ in the ball of radius $R$ and the rate. Comparing
two multiplier sequences as in Step~2 gives
$\|y^\mu-y^{\tilde\mu}\|_\infty\le(1-\Theta_1-\Theta_2)^{-1}\sum_i\pi_i|\mu_i-\tilde\mu_i|/\gamma_i$.

For the budgets define $C_i^\pm$ and $I_i$ as in Step~3, which uses no
summation over types, so the envelopes \eqref{ec:eq:wealthenv} hold type by
type. By \eqref{ec:eq:countablesum}, $\sum_i\pi_i\sup_{x\in I_i}|x|<\infty$.
Let $\mathcal I=\prod_iI_i$ with the product topology, and for $f\in\mathcal I$
put $S(f)=\sum_j\pi_j(f_j-\log X_j)$ and
$\mu_i(f)=\gamma_i(f_i-\log X_i)-\theta_i(\gamma_i-1)S(f)$, which is
$\mathbf A^{-1}$ applied to $f-\log X$. Then
$\sum_i\pi_i\mu_i(f)/\gamma_i=DS(f)$, $k(\mu(f))=\bar m/D-S(f)$, and
$\sum_i\pi_i|\mu_i(f)|/\gamma_i$ is bounded over $\mathcal I$. Substituting
$f_i-\log X_i=\mu_i(f)/\gamma_i+\theta_i(1-1/\gamma_i)S(f)$ into
\eqref{ec:eq:wealthenv} shows that
$(\mathcal Hf)_i=\log\mathcal E_i(\mu(f))+f_i-\log X_i$ maps $\mathcal I$
into itself. The map is continuous in the product topology. If
$f^{(n)}\to f$ coordinatewise, dominated convergence gives
$S(f^{(n)})\to S(f)$ and $\sum_i\pi_i|\mu_i(f^{(n)})-\mu_i(f)|/\gamma_i\to0$,
the benchmarks then converge uniformly, each log relative consumption
converges uniformly by \eqref{ec:eq:loglipschitz}, and each expenditure
converges by dominated convergence on the infinite horizon. A fixed point
follows from Brouwer's theorem in finite dimensions. Fix $f^0\in\mathcal I$,
freeze the coordinates $i>N$ at $f^0_i$, and apply Brouwer's theorem to the
first $N$ coordinates of $\mathcal H$, which gives $f^{(N)}\in\mathcal I$
with $f^{(N)}_i=(\mathcal Hf^{(N)})_i$ for $i\le N$. A diagonal subsequence
converges coordinatewise to some $f^\ast$, since every $I_i$ is compact,
and coordinatewise continuity gives $f^\ast=\mathcal Hf^\ast$, that is
$\mathcal E_i(\mu(f^\ast))=X_i$ for every $i$. Individual optimality,
admissibility, index bounds, and financing are as in the finite case, and
the summability above justifies the conditional aggregation.
\end{proof}

If the preference and market parameters range over a compact set on which
$\inf\hat\kappa_i>0$ and $\inf\hat K_i(\gamma_i)>0$ for the trend
coefficients induced by the distribution, then continuity of the band
endpoints bounds $a_i$ and $b_i$, so every countable distribution of such
parameters with $\Theta_1+\Theta_2<1$ and $\sum_i\pi_i|\log X_i|<\infty$
satisfies \eqref{ec:eq:countablesum}, for arbitrary positive references.

\begin{theorem}\label{ec:thm:continuousinitial}
Let the preference and market parameters take finitely many values
$i=1,\dots,n$ with probabilities $\pi_i>0$ and satisfy the hypotheses of
Theorem~\ref{thm:weakexist}. Conditionally on each regime, let the initial
wealth $X$ and the reference $h_{0-}$ have any joint distribution,
independent of the Brownian motions, with
$\E[|\log X|+|\log h_{0-}|]<\infty$. Then a mean field equilibrium exists in
which every budget binds and the indices are bounded, and its benchmark
satisfies $|\log Y_t-k-gt-sB_t|\le R$ almost surely for some finite $k$,
with the $g$, $s$, and $R$ of Theorem~\ref{thm:weakexist}. The joint
distribution of the initial data may be nonatomic and unbounded.
\end{theorem}

\begin{proof}
Put $\bar h=\exp(\E\log h_{0-})$, let $w=X/h_{0-}$ be the wealth to reference
ratio, and let $\Lambda_i$ be the conditional law of $w$ in regime $i$,
whose absolute logarithm is integrable. Homogeneity gives
\begin{equation}\label{ec:eq:unitreference}
 h^\ast(P,h_{0-})=h_{0-}\,h^\ast(h_{0-}^{\gamma_i}P,1),
\end{equation}
by multiplying the objective by $h_{0-}^{1-\gamma_i}$ for power utility,
while for logarithmic utility the running utility changes by the constant
$\delta_i^{-1}\log h_{0-}$ and the adjustment cost is unchanged. We first
construct an equilibrium of the proxy population in which every household
has the reference $\bar h$ and the wealth $\bar hw$. Once its unit reference
responses $\mathring h_i(w)=h_i(w)/\bar h$ are known, the response of an
original household is $h_{0-}\mathring h_i(w)$, with expenditure
$h_{0-}w=X$, and the aggregate log relative consumption of both populations
agrees because $\E\log h_{0-}=\log\bar h$.

Choose nondecreasing finite valued quantizations
$\varrho_m:(0,\infty)\to(0,\infty)$ with $\varrho_m(w)\to w$ and
$|\log\varrho_m(w)|\le|\log w|+1$, for instance by truncating $\log w$ to
$[-m,m]$, rounding down to the grid $m^{-1}\mathbb Z$, and exponentiating.
Apply Theorem~\ref{thm:weakexist} to the finite population with the law of
$(i,\bar h\varrho_m(w),\bar h)$, and write $y^{(m)}$ for its log benchmark,
$k^{(m)}$ for its center, and $\mu^{(m)}_i(\varrho_m(w))$ for its log
multipliers. For a quantization value of zero probability, let the
multiplier be the unique one at which the budget $\bar h\varrho_m(w)$ binds
at the benchmark $y^{(m)}$, which exists by Theorem~\ref{ec:thm:comparable}
and changes nothing. Set $\mathring\mu_i^{(m)}(w)=\mu^{(m)}_i(\varrho_m(w))+\gamma_i\log\bar h$,
so that at the benchmark $y^{(m)}$ the unit reference response
$\mathring h^{(m)}_i(w)=h^\ast(e^{\mathring\mu_i^{(m)}(w)}e^{\delta_it}\xi^ie^{\theta_iy^{(m)}},1)$
has expenditure $\varrho_m(w)$. The uniqueness and strict monotonicity of
the budget multiplier in Theorem~\ref{ec:thm:comparable} make
$\mathring\mu_i^{(m)}(w)$ nonincreasing in $w$.

The envelope \eqref{ec:eq:wealthenv} supplies numbers
$\epsilon_i^{(m)}(w)\in[\log C_i^-,\log C_i^+]$ with
\[
 \log(\bar h\varrho_m(w))=-\frac{\mu^{(m)}_i(\varrho_m(w))}{\gamma_i}+\theta_i(1-\gamma_i^{-1})k^{(m)}+\epsilon_i^{(m)}(w).
\]
Averaging over regimes and $w$, and using $Dk^{(m)}=\bar m-\E[\mu^{(m)}/\gamma]$
and $D+\E[\theta(1-1/\gamma)]=1$, gives
\[
 k^{(m)}=\bar m+\E\log h_{0-}+\E\log\varrho_m(w)-\E\,\epsilon_i^{(m)}(w),
\]
so $(k^{(m)})$ is bounded. Returning to the individual envelope gives
constants $M_i$ independent of $m$ and $w$ with
$|\mathring\mu_i^{(m)}(w)+\gamma_i\log\varrho_m(w)|\le M_i$. Helly's selection
theorem for monotone functions on an exhausting sequence of compact
subintervals of $(0,\infty)$, a diagonal argument over the finitely many
regimes, and a further extraction at the countably many atoms of the $\Lambda_i$
in the discontinuity set of the limit give a subsequence along which
$\mathring\mu_i^{(m)}(w)\to\mathring\mu_i(w)$ for $\Lambda_i$ almost every
$w$, and $\int|\mathring\mu_i^{(m)}-\mathring\mu_i|\dd\Lambda_i\to0$ by
dominated convergence, with the dominating function
$M_i+\gamma_i(|\log w|+1)$. At the countably many exceptional $w$ define
$\mathring\mu_i(w)$ as any limit point, which leaves the integrals
unchanged.

Realize all these equilibria on one product space. Their benchmark
equations read
$y^{(m)}_t=(1-\Theta_1)^{-1}\big(\E\log h_{0-}+\sum_i\pi_i\int\E[\log\mathring h^{(m)}_i(w)_t\,|\,\F^B_t]\dd\Lambda_i(w)\big)$,
and the fixed multiplier estimate of Step~2, with the integral over $w$ in
place of the sum over types, gives
\[
 \|y^{(m)}-y^{(l)}\|_\infty\le\frac1{1-\Theta_1-\Theta_2}
 \sum_i\frac{\pi_i}{\gamma_i}\int|\mathring\mu_i^{(m)}(w)-\mathring\mu_i^{(l)}(w)|\dd\Lambda_i(w)
\]
along the subsequence. Hence $y^{(m)}$ converges in the uniform essential
path norm to an adapted c\`adl\`ag process $y$, and after a further
extraction $k^{(m)}\to k$ with $|y_t-k-gt-sB_t|\le R$. For $\Lambda_i$ almost every
$w$ define $\mathring h_i(w)=h^\ast(e^{\mathring\mu_i(w)}e^{\delta_it}\xi^ie^{\theta_iy},1)$,
and at the exceptional $w$ let $\mathring h_i(w)$ be the unique optimum
with expenditure $w$ at this benchmark. These responses have a version
that is jointly measurable in $(w,t,\omega)$ with adapted c\`adl\`ag
paths. To construct it, choose the optimal response at every rational log
multiplier on the canonical space, so that all c\`adl\`ag versions and all
estimates \eqref{ec:eq:loglipschitz} hold outside one null set. A real log
multiplier is the limit of rational valued Borel functions of $w$, and
\eqref{ec:eq:loglipschitz} makes the uniform limit in log consumption well
defined and equal to the optimum at the real multiplier. By
Theorem~\ref{ec:thm:comparable},
$\|\log\mathring h^{(m)}_i(w)-\log\mathring h_i(w)\|_\infty
\le\gamma_i^{-1}(|\mathring\mu_i^{(m)}(w)-\mathring\mu_i(w)|+\theta_i\|y^{(m)}-y\|_\infty)$, so
the expenditure of $\mathring h_i(w)$ is the limit of $\varrho_m(w)$, that
is $w$, by dominated convergence with the limiting integrable integrand,
and integrating the estimate over $\Lambda_i$ and conditioning on the
common path passes the benchmark equation to the limit. This is an equilibrium of
the proxy population.

Finally set $Y=e^y$, and for an original household with data $(i,X,h_{0-})$
take $h=h_{0-}\mathring h_i(X/h_{0-})$ and the multiplier
$e^{\mathring\mu_i(X/h_{0-})}h_{0-}^{-\gamma_i}$. Identity \eqref{ec:eq:unitreference}
gives optimality and admissibility, the budget is $h_{0-}w=X$, and
aggregation follows from $\E\log h_{0-}=\log\bar h$. The log moment bound
and the geometric envelopes justify all conditional expectations.
\end{proof}

The proof shows more. Two populations with the same regime probabilities,
the same conditional laws of $X/h_{0-}$ in each regime, and the same value
of $\E\log h_{0-}$ have the same set of equilibrium benchmarks in the
bounded class, because at any such benchmark the unit reference response
depends on the initial data only through the regime and the ratio
$X/h_{0-}$, and aggregation depends only on $\E\log h_{0-}$ and the laws of
the ratio. In particular, replacing every reference by $\bar h$ and every
wealth by $\bar hX/h_{0-}$ preserves this set.

%% file: ec4g_arithmetic.tex
\subsection{Arithmetic Benchmarks}\label{ec:arithmetic}

Remark~\ref{rem:scope} notes that a deviation does not enter the arithmetic
average $Y^{\rm ar}_t=N^{-1}\sum_jc^j_t$ as a power. Exact implementation
survives with this benchmark when consumption shares are constant, and it
fails in general. Markets are
complete for the common risk, preferences may be any functional of the
relative consumption path, and the same admissible relative paths are used
for the price taker and in the game. A price taking equilibrium for the
arithmetic benchmark is a profile in which each $c^i$ is optimal for a
price taker facing $Y^{\rm ar}$ with budget $X_i$.

\begin{proposition}\label{ec:prop:arithmeticpos}
Let $(c^1,\dots,c^N)$ be a price taking equilibrium for the arithmetic
benchmark in which every share $\varkappa_i=c^i_t/\sum_jc^j_t$ is a
deterministic constant, and let $0\le\theta_i<1$. Then the profile is a
Nash equilibrium of the game with the arithmetic benchmark.
\end{proposition}

\begin{proof}
Fix $i$, suppress its subscript, and let $\varrho=\tilde c/c$ for a
deviation $\tilde c$. The benchmark becomes $Y^{\rm ar}(1-\varkappa+\varkappa\varrho)$
and the relative consumption of the deviating household becomes
$h\phi_\varkappa(\varrho)$ with
$\phi_\varkappa(\varrho)=\varrho(1-\varkappa+\varkappa\varrho)^{-\theta}$ and
$h=c(Y^{\rm ar})^{-\theta}$. A price taker facing $Y^{\rm ar}$ reproduces this
relative path by buying $c\phi_\varkappa(\varrho)$. For $\varrho>0$,
\[
 \phi_\varkappa''(\varrho)=-\theta\varkappa(1-\varkappa+\varkappa\varrho)^{-\theta-2}
 \big\{2(1-\varkappa)+(1-\theta)\varkappa\varrho\big\}\le0 ,
\]
so the tangent at $\varrho=1$ gives
$\phi_\varkappa(\varrho)\le\theta\varkappa+(1-\theta\varkappa)\varrho$. With
$\E\int_0^\infty\xi_tc_t\dd t\le X$ and the deviation financeable, the
reproducing plan costs at most
$\theta\varkappa\E\int\xi_tc_t\dd t+(1-\theta\varkappa)\E\int\xi_t\tilde c_t\dd t\le X$.
It is admissible whenever the relative path of the deviation is, so price
taking optimality excludes a profitable deviation.
\end{proof}

Constant shares include identical households and proportional profiles. The
hypothesis cannot be dropped, as the following example with types without
loss aversion, which are admitted by Remark~\ref{rem:frictionless}, shows.

\begin{proposition}\label{ec:prop:arithmeticneg}
Let $r>0$, $\kappa>0$, let two households have logarithmic utility,
$\alpha_i=\beta_i=0$, distinct discount rates $\delta_1\ne\delta_2$,
relative concerns $0<\theta_1,\theta_2<1/2$, and wealth $X_1,X_2>0$. Then
the price taking equilibrium for the arithmetic benchmark is not a Nash
equilibrium of the two household game, although the weak interaction
condition \eqref{eq:weakcondition} holds.
\end{proposition}

\begin{proof}
With $\xi_t=\exp(-(r+\kappa^2/2)t-\kappa B_t)$ and any fixed benchmark,
logarithmic price taking utility differs from
$\E\int e^{-\delta_it}\log c_t\dd t$ by a constant, so the unique optimum of
household $i$ is $c^i_t=\delta_iX_ie^{-\delta_it}/\xi_t$, with
$\E\int\xi_tc^i_t\dd t=X_i$, and $Y^{\rm ar}=(c^1+c^2)/2$ completes the price
taking equilibrium. All discounted logarithmic utilities are finite. The
share $\varkappa_i(t)=\delta_iX_ie^{-\delta_it}/(\delta_1X_1e^{-\delta_1t}+\delta_2X_2e^{-\delta_2t})$
is deterministic, strictly monotone, and not constant. Fix $T>0$, put
$\vartheta(t)=1-\theta_i\varkappa_i(t)$,
$\bar\vartheta=\int_0^Te^{-\delta_it}\vartheta(t)\dd t/\int_0^Te^{-\delta_it}\dd t$,
and $f(t)=(\vartheta(t)-\bar\vartheta)\ind_{[0,T)}(t)$. For small $\epsilon>0$ the
deviation $\tilde c^i=c^i(1+\epsilon f)$ is positive, lies in the admissible
class of the type, and has the same budget, because $\xi_tc^i_t$ is
deterministic and $\int_0^Te^{-\delta_it}f(t)\dd t=0$. With the other
household fixed, its utility changes by
\[
 \Delta(\epsilon)=\int_0^Te^{-\delta_it}\big\{\log(1+\epsilon f(t))
 -\theta_i\log(1+\epsilon\varkappa_i(t)f(t))\big\}\dd t,
\]
and differentiation under the integral gives
\[
 \Delta'(0)=\int_0^Te^{-\delta_it}\vartheta(t)f(t)\dd t
 =\int_0^Te^{-\delta_it}(\vartheta(t)-\bar\vartheta)^2\dd t>0 ,
\]
so the deviation is profitable. Finally
$\sum_i\pi_i\theta_i(1+\gamma_i^{-1})=\theta_1+\theta_2<1$.
\end{proof}

The failure does not depend on $\kappa$ and is caused by the time varying
share of each household in the arithmetic benchmark. Whether a
counterexample exists when every type has $L>1$ is open.

%% file: ec3_additions.tex
\subsection{Certificates for Approximate Best Responses}\label{ec:approxcert}

Theorems~\ref{ec:thm:regret} and~\ref{thm:certificate} assume that
the candidate is an exact best response to some benchmark. Exact
optimality can be replaced by a feasible approximate adjoint, at the cost
of a complementarity slack that involves only the variation of the
candidate. Write
$\mathcal C(v)=\E\int_{[0,\infty)}e^{-\delta t}(\alpha\dd v^+_t+\beta\dd v^-_t)$
for the adjustment cost and
$\operatorname{Var}_\delta(v)=\E\int_{[0,\infty)}e^{-\delta t}\dd|v|_t$, and let
all paths carry the same reference.

\begin{theorem}\label{ec:thm:approxregret}
Let $h^0\in\mathcal A$ with $v^0=u(h^0)$, let $\tilde z$ be a bounded
positive progressively measurable process, and let
$\Upsilon_t=\E[\int_t^\infty e^{-\delta(s-t)}(1-\tilde z_s)\dd s\,|\,\F_t]$
satisfy $-\beta\le\Upsilon\le\alpha$ up to indistinguishability. Put
\[
 \mathcal K(\Upsilon,v^0)=\E\int_{[0,\infty)}e^{-\delta t}
 \big((\alpha-\Upsilon_t)\dd v^{0,+}_t+(\beta+\Upsilon_t)\dd v^{0,-}_t\big)\ge0 .
\]
For a positive price $P$ at which $h^0$ has finite expenditure let
$z^P=P\psi_p'(h^0)/u'(h^0)$. Then every $h\in\mathcal A$ with finite $\J_P(h)$
satisfies
\begin{equation}\label{ec:eq:approxregret}
 \J_P(h)-\J_P(h^0)\le\mathcal K(\Upsilon,v^0)
 +\E\int_0^\infty e^{-\delta t}\tilde z_t(h^0_t)^{1-\gamma}G_{\gamma,p}(z^P_t/\tilde z_t)\dd t ,
\end{equation}
with $G_{\gamma,p}$ as in Theorem~\ref{ec:thm:regret}, where the right side
may be infinite. If $P=\lambda e^{\delta t}\xi\varpi$ with $\lambda>0$ and
$w_0=\E\int\xi_t\varpi_t\psi_p(h^0_t)\dd t$, every budget feasible competitor
has a utility gain over $h^0$ of at most the right side plus
$\lambda(X_0-w_0)$, whether or not $w_0\le X_0$.
\end{theorem}

\begin{proof}
Since $\Upsilon\in[-\beta,\alpha]$, every admissible $v$ satisfies
$\mathcal C(v)\ge\E\int e^{-\delta t}\Upsilon_t\dd v_t$, and for the candidate
the gap between the two sides is $\mathcal K(\Upsilon,v^0)$. The Fubini
identity in the proof of Theorem~\ref{ec:thm:foc}, applied to $v-v^0$,
gives
$\mathcal C(v)-\mathcal C(v^0)\ge\E\int_0^\infty e^{-\delta t}(1-\tilde z_t)(v_t-v^0_t)\dd t-\mathcal K(\Upsilon,v^0)$,
where boundedness of $\Upsilon$ and $\tilde z$ and admissibility justify
every term. With $\Phi_p=\psi_p\circ u^{-1}$ it follows that
$\J_P(h)-\J_P(h^0)\le\mathcal K(\Upsilon,v^0)+\E\int_0^\infty e^{-\delta t}
\{\tilde z_t(v_t-v^0_t)-P_t(\Phi_p(v_t)-\Phi_p(v^0_t))\}\dd t$. At each
$(t,\omega)$ the expression in braces is strictly concave in $v_t$, and
writing the consumption as $sh^0_t$ its maximum is attained at
$s=(z^P_t/\tilde z_t)^{-1/(\gamma+1/p-1)}$ and equals
$\tilde z_t(h^0_t)^{1-\gamma}G_{\gamma,p}(z^P_t/\tilde z_t)$, as in the proof of
Theorem~\ref{ec:thm:regret}. The budget claim follows from
$U(h)-U(h^0)=\J_P(h)-\J_P(h^0)+\lambda(\E\int\xi\varpi\psi_p(h)\dd t-w_0)$
and the budget of $h$.
\end{proof}

The feasible process $\Upsilon$ can be built from an approximation of the
adjoint of the candidate, without knowing an exact best response.

\begin{corollary}\label{ec:cor:repair}
Let $z=P^0\psi_p'(h^0)/u'(h^0)$ be bounded and positive, let $\Psi$ be its
adjoint as in \eqref{ec:eq:index}, and let a bounded adapted c\`adl\`ag
process $\hat\Psi$ satisfy $\|\hat\Psi-\Psi\|_\infty\le\varepsilon_R$ and
$-\beta-\varepsilon_o\le\hat\Psi\le\alpha+\varepsilon_o$ for some
$\varepsilon_R,\varepsilon_o\ge0$. With $\bar\alpha=(\alpha+\beta)/2$,
$\tilde\alpha=(\alpha-\beta)/2$, $\eta=\varepsilon_o+\varepsilon_R$, and
$\varkappa=\bar\alpha/(\bar\alpha+\eta)$ put
$\Upsilon=\varkappa\Psi+(1-\varkappa)\tilde\alpha$ and
$\tilde z=\varkappa z+(1-\varkappa)(1-\delta\tilde\alpha)$. Then $\Upsilon$
and $\tilde z$ satisfy the hypotheses of Theorem~\ref{ec:thm:approxregret},
and with
$\hat{\mathcal K}=\E\int_{[0,\infty)}e^{-\delta t}((\alpha-\hat\Psi_t)^+\dd v^{0,+}_t+(\beta+\hat\Psi_t)^+\dd v^{0,-}_t)$,
\begin{equation}\label{ec:eq:repairbound}
 0\le\mathcal K(\Upsilon,v^0)\le\varkappa\hat{\mathcal K}
 +\varkappa(\varepsilon_R+\eta)\operatorname{Var}_\delta(v^0).
\end{equation}
\end{corollary}

\begin{proof}
The hypotheses give $-\beta-\eta\le\Psi\le\alpha+\eta$, and since
$\alpha-\tilde\alpha=\beta+\tilde\alpha=\bar\alpha$ and
$\varkappa(\bar\alpha+\eta)=\bar\alpha$, the affine map
$x\mapsto\varkappa x+(1-\varkappa)\tilde\alpha$ sends $[-\beta-\eta,\alpha+\eta]$
onto $[-\beta,\alpha]$. Since $\tilde\alpha<\alpha<1/\delta$, the constant
$1-\delta\tilde\alpha$ is positive, so $\tilde z$ is positive and bounded,
and linearity of conditional expectations gives
$\Upsilon_t=\E[\int_t^\infty e^{-\delta(s-t)}(1-\tilde z_s)\dd s\,|\,\F_t]$
because $(1-(1-\delta\tilde\alpha))/\delta=\tilde\alpha$. For the upward
variation,
$\alpha-\Upsilon=\varkappa(\alpha-\hat\Psi)+\varkappa(\hat\Psi-\Psi)+(1-\varkappa)\bar\alpha
\le\varkappa(\alpha-\hat\Psi)^++\varkappa\varepsilon_R+(1-\varkappa)\bar\alpha$,
with the same bound for $\beta+\Upsilon$, and $(1-\varkappa)\bar\alpha=\varkappa\eta$
gives \eqref{ec:eq:repairbound}.
\end{proof}

A direct approximate optimality argument would multiply the obstacle
violation by the variation of the competitor, which is not bounded on
$\mathcal A$, whereas here only the variation of the candidate enters. The
next lemma gives a sufficient residual estimate for $\varepsilon_R$ for a
continuous reconstruction of a numerical adjoint.

\begin{lemma}\label{ec:lem:residual}
Let $z$ be bounded and positive with adjoint $\Psi$, and let a bounded
adapted c\`adl\`ag process $\hat\Psi$ satisfy on $(0,\infty)$
\[
 \dd\hat\Psi_t=\big(\delta\hat\Psi_t-(1-z_t)\big)\dd t+\dd M_t+\dd R_t ,
\]
where $R$ is adapted and of finite variation on compacts and
$\int_0^te^{-\delta s}\dd M_s$ is a martingale on every finite interval. If
$\E[\int_{(\tau,\infty)}e^{-\delta(s-\tau)}\dd|R|_s\,|\,\F_\tau]\le\varepsilon_R$
almost surely for every finite stopping time $\tau$, then
$\|\hat\Psi-\Psi\|_\infty\le\varepsilon_R$. In particular
$\varepsilon_R=\varepsilon_0/\delta$ is valid when $\dd R_t=\varepsilon_t\dd t$
with $|\varepsilon_t|\le\varepsilon_0$.
\end{lemma}

\begin{proof}
Multiply the equation by $e^{-\delta t}$, integrate from a bounded stopping
time $\tau$ to $T$, and take conditional expectations. Boundedness of
$\hat\Psi$ removes the terminal term as $T\to\infty$, and the conditional
integrability of $R$ allows the passage to the limit in the finite
variation term, which gives
$\hat\Psi_\tau=\Psi_\tau-\E[\int_{(\tau,\infty)}e^{-\delta(s-\tau)}\dd R_s\,|\,\F_\tau]$.
The bound holds at every bounded stopping time and hence, by right
continuity, up to indistinguishability.
\end{proof}

For a reconstruction $\hat\Psi=\phi(\bm x)$ with a bounded $C^2$ function
$\phi$ of a reflected state $\dd\bm x=\bm m(\bm x)\dd t+\sigma(\bm x)\dd B+\sum_jr_j(\bm x)\dd l^j$,
with $z=z(\bm x)$ and a martingale part that satisfies the integrability
of Lemma~\ref{ec:lem:residual}, It\^o's formula gives
$\dd R=(\mathcal L\phi-\delta\phi+1-z)(\bm x)\dd t+\sum_j\nabla\phi(\bm x)\cdot r_j(\bm x)\dd l^j$,
so a bound $\varepsilon_0$ on the generator residual on the reached state
space, bounds $\bar b_j$ on $|\nabla\phi\cdot r_j|$ on the supports of the
regulators, and bounds $\bar A_j$ on the conditional discounted regulator
masses give $\varepsilon_R=\varepsilon_0/\delta+\sum_j\bar b_j\bar A_j$. The reflection terms
cannot be omitted, a nonsmooth reconstruction adds local time terms at the
interpolation interfaces, and a residual checked only at grid points or on
simulated paths does not provide these bounds.

If candidates $h^n$ evaluated at their own prices have indices in a fixed
compact subset of $(0,\infty)$, uniformly bounded $\operatorname{Var}_\delta(u(h^n))$
and $\E\int e^{-\delta t}(h^n_t)^{1-\gamma}\dd t$, and
$\hat{\mathcal K}_n+\varepsilon_{o,n}+\varepsilon_{R,n}\to0$, then the regret
bound \eqref{ec:eq:approxregret} with the repaired $\Upsilon$ and $\tilde z$
of Corollary~\ref{ec:cor:repair} tends to zero. Its price term is of order
$(\varepsilon_{o,n}+\varepsilon_{R,n})^2$, because
$G_{\gamma,p}(1)=G_{\gamma,p}'(1)=0$,
$G_{\gamma,p}''(M)=(p/\rho)M^{-1-1/\rho}$ with $\rho=1+p(\gamma-1)$, and
$\tilde z^n/z^n-1=O(\varepsilon_{o,n}+\varepsilon_{R,n})$. If the candidates are budget feasible and
their budget slack times the multiplier tends to zero, their utility regret
tends to zero as well.

In the setting of Theorem~\ref{thm:certificate}, with
$\Delta=\log(Y^N/Y^0)$, the exact optimality hypothesis for type $a$ can be
replaced in the same way. Construct $\tilde z^a$ and $\Upsilon^a$ for the
candidate $h^{a,0}$ at the price $P^{a,0}=\lambda_ae^{\delta_at}\xi(Y^0)^{\theta_a}$
with index $z^{a,0}$. H\"older's inequality in the proof of
Theorem~\ref{thm:certificate} makes every deviation in the finite game
affordable at the price taking budget, so applying
Theorem~\ref{ec:thm:approxregret} with $p=1$ and the price
$P^{a,0}e^{\theta_a\Delta}$, adding the budget term, and subtracting the
candidate's actual utility bounds the gain of every admissible financeable
unilateral deviation by
\begin{equation}\label{ec:eq:approxfinite}
 \mathcal K(\Upsilon^a,u(h^{a,0}))
 +\E\int_0^\infty e^{-\delta_at}\tilde z^a_t(h^{a,0}_t)^{1-\gamma_a}
 G_{\gamma_a}\Big(\frac{e^{\theta_a\Delta_t}z^{a,0}_t}{\tilde z^a_t}\Big)\dd t
 +U_a(h^{a,0})-U_a(h^{a,N})
 +\lambda_a\Big(X_a-\E\int_0^\infty\xi_tc^a_te^{\theta_a\Delta_t}\dd t\Big),
\end{equation}
provided the cross expenditure is finite. Exact complementarity and
$\tilde z^a=z^{a,0}$ recover Theorem~\ref{thm:certificate}. A profile whose
budgets are only approximately satisfied can be made feasible by scaling
$c^a$ by $\min\{1,X_a/\overline w_a\}$ for certified cost bounds
$\overline w_a\ge\E\int_0^\infty\xi_tc^a_t\dd t$, which changes the geometric
benchmark and the relative plans by constant factors, and
\eqref{ec:eq:approxfinite} is then evaluated for the scaled profile, with
the initial jumps measured from the original references.

Finally, the contraction in the proof of Theorem~\ref{thm:weakexist}
tolerates certified errors. Fix the multipliers, let $\mathcal T$ be the
benchmark map on the invariant ball with factor $q$ and fixed point
$y^\ast$, and let $y^0$ and all iterates lie in that ball with
$\|y^{k+1}-\mathcal Ty^k\|_\infty\le e_k$. The triangle inequality gives
$\|y^{k+1}-y^\ast\|_\infty\le e_k+q\|y^k-y^\ast\|_\infty$, so
$\|y^k-y^\ast\|_\infty\le q^k\|y^0-y^\ast\|_\infty+\sum_{j<k}q^{k-1-j}e_j$,
which tends to zero when $e_k\to0$. For a computed map
$\tilde{\mathcal T}$ with $\|\tilde{\mathcal T}y-\mathcal Ty\|_\infty\le e$
on the ball, any $y$ in the ball satisfies
$\|y-y^\ast\|_\infty\le(\|y-\tilde{\mathcal T}y\|_\infty+e)/(1-q)$, and
best responses with $\|\log\tilde h_i[y]-\log h_i[y]\|_\infty\le\epsilon_i$
give $e=(1-\Theta_1)^{-1}\sum_i\pi_i\epsilon_i$. On a fixed grid, the
discrete adjoint $\Psi=f+e^{-\delta\Delta t}\mathbf P\Psi$ of a fixed
policy with a stochastic matrix $\mathbf P$ satisfies
$\|\hat\Psi-\Psi\|_\infty\le\|\hat\Psi-f-e^{-\delta\Delta t}\mathbf P\hat\Psi\|_\infty/(1-e^{-\delta\Delta t})$
for any vector $\hat\Psi$, which certifies the algebraic error of a
linear solve but says nothing about the distance between discrete and
continuous policies.

\subsection{Stability and Variation Bounds for the Lattice}\label{ec:latticebounds}

Convergence of the scheme of Section~\ref{ec:lattice} as the time step
tends to zero is not proved. This subsection records what is proved. The
first result concerns the linear systems that Howard's algorithm solves,
the second bounds the expected total revision of the exact discrete
policies uniformly in the time step, which is what the interpolation error
requires, and the third gives a convergence rate for the adjoint of the
scheme on synchronized states. Throughout, $\Delta x=\kappa\sqrt{\Delta t}$, and
multilinear interpolation is used with all corners retained and without
clipping.

\begin{proposition}\label{ec:prop:stable}
For a fixed array of revisions let $\mathbf P$ be the interpolation
transition matrix of the scheme, a nonnegative matrix with row sums at most
one, and let $\mathbf D$ be the diagonal matrix with entry one on the rows
of inactive types and zero on the rows of revising types. The matrix
$\mathbf L=I-e^{-\delta\Delta t}\mathbf D\mathbf P$ of the linear step of
Howard's algorithm is invertible with $\mathbf L^{-1}\ge0$ and
$\|\mathbf L^{-1}\|_\infty\le(1-e^{-\delta\Delta t})^{-1}$. If
$\mathbf L\Psi=b$ with obstacle values bounded by $M_o$ on the revising rows
and right sides $\Delta tf$ with $\|f\|_\infty\le M_f$ on the inactive rows,
then $\|\Psi\|_\infty\le\max\{M_o,\Delta tM_f/(1-e^{-\delta\Delta t})\}$, and
an approximate solution $\hat\Psi$ satisfies
$\|\hat\Psi-\Psi\|_\infty\le\|\mathbf L\hat\Psi-b\|_\infty/(1-e^{-\delta\Delta t})$.
\end{proposition}

\begin{proof}
The nonnegative matrix $\mathbf D\mathbf P$ has supremum norm at most one,
so the Neumann series $\mathbf L^{-1}=\sum_{k\ge0}e^{-k\delta\Delta t}(\mathbf D\mathbf P)^k$
converges, is nonnegative, and has norm at most $(1-e^{-\delta\Delta t})^{-1}$.
With $M$ the stated bound, $\mathbf L(M\mathbf 1)$ equals $M\ge|b|$ on revising
rows and is at least $(1-e^{-\delta\Delta t})M\ge|b|$ on inactive rows, so
$\mathbf L(M\mathbf 1\pm\Psi)=\mathbf L(M\mathbf 1)\pm b\ge0$, and positivity
of the inverse gives $M\mathbf 1\pm\Psi\ge0$, that is $|\Psi|\le M$.
The residual estimate follows by subtraction.
\end{proof}

For the variation bound write $V=\mathrm{diag}(\gamma_i)+\theta\pi^\top/(1-\Theta_1)$
for the revision matrix of Section~\ref{ec:lattice}, let
$\bm\Pi=\mathrm{diag}(\pi_i/\theta_i)$ with $\theta_i\in(0,1)$, and
$\mathbf H=\bm\Pi V^{-1}$. For a time step $\Delta t$ put
$\delta^\Delta=(1-e^{-\delta\Delta t})/\Delta t\in(0,\delta)$,
$\underline y_i=\log(1-\alpha_i\delta^\Delta)$,
$\overline y_i=\log(1+\beta_i\delta^\Delta)$, and
$y^\circ_i=(\underline y_i+\overline y_i)/2$, which are well defined because
$\alpha_i<1/\delta$ and $\beta_i\ge0$, and satisfy $\overline y_i-\underline y_i\to\log L_i>0$
as $\Delta t\to0$. Let $\mathcal Q_\Delta(x)=\frac12(x-y^\circ)^\top\mathbf H(x-y^\circ)$.

\begin{theorem}\label{ec:thm:variation}
Let every type have $\theta_i\in(0,1)$ and $L_i>1$, and let the nodal
adjoints satisfy $-\beta_i\le\Psi^i\le\alpha_i$. Suppose that at every grid
state $x$ the revisions $g=g(x)$ solve the complementarity conditions of
Section~\ref{ec:lattice} exactly, with $y=x+Vg$ and the continuation
adjoint built from the common shock and multilinear interpolation. Realize
the transition as the common shock followed by independent unbiased
rounding of each coordinate to a neighboring grid coordinate, which has
exactly the transition probabilities of the scheme. Then for every
$T<\infty$ there is a constant $C_T$, independent of all sufficiently
small $\Delta t$, such that the chain $\bm x_k$ started at a grid point
$\bm x_0$ satisfies
\begin{equation}\label{ec:eq:variationbound}
 \E\sum_{k\Delta t\le T}\|g_k\|_1\le C_T(1+|\bm x_0|^2),\qquad
 \sup_{k\Delta t\le T}\E|\bm x_k|^2\le C_T(1+|\bm x_0|^2),\qquad
 \lim_{\varrho\to\infty}\sup_{\Delta t}\mathbb P\Big(\max_{k\Delta t\le T}|\bm x_k|>\varrho\Big)=0 ,
\end{equation}
and the rounding errors $\varepsilon_k$ are martingale differences with
$\E\max_{m\le T/\Delta t}|\sum_{k<m}\varepsilon_k|^2\le C_T\Delta x(1+|\bm x_0|^2)\to0$.
The same bounds hold for the chain stopped before the first corner
deletion or clipping.
\end{theorem}

\begin{proof}
The matrix $\bm\Pi V=\mathrm{diag}(\pi_i\gamma_i/\theta_i)+\pi\pi^\top/(1-\Theta_1)$
is symmetric positive definite, so
$\mathbf H=\bm\Pi(\bm\Pi V)^{-1}\bm\Pi$ is symmetric positive definite, with
$\mathbf HV=\bm\Pi$ and $V^\top\mathbf HV=\bm\Pi V$. Positive interpolation
weights preserve the obstacle bounds, so if $g_i>0$ then
$\alpha_i=C^i(y)\le(1-e^{y_i})\Delta t+e^{-\delta\Delta t}\alpha_i$, that is
$y_i\le\underline y_i$, and if $g_i<0$ then $y_i\ge\overline y_i$. The exact
quadratic identity
$\mathcal Q_\Delta(y)-\mathcal Q_\Delta(x)=(y-y^\circ)^\top\bm\Pi g-\frac12g^\top\bm\Pi Vg$
therefore gives
\begin{equation}\label{ec:eq:lyapunovrevision}
 \mathcal Q_\Delta(y)-\mathcal Q_\Delta(x)\le-c\|g\|_1-c_2|g|^2,\qquad
 c=\inf_{\Delta t\le\Delta_0}\min_i\frac{\pi_i(\overline y_i-\underline y_i)}{2\theta_i}>0,\quad
 c_2=\tfrac12\lambda_{\min}(\bm\Pi V)>0 ,
\end{equation}
because each term $(\pi_i/\theta_i)(y_i-y^\circ_i)g_i$ is at most
$-(\pi_i/\theta_i)(\overline y_i-\underline y_i)|g_i|/2$. The common shock
$\pm\Delta x\mathbf 1$, whose sign has mean $2p_+-1=\nu\Delta t/\Delta x$,
changes the expectation of $\mathcal Q_\Delta$ by
$\nu\Delta t\,\mathbf 1^\top\mathbf H(y-y^\circ)+\frac12\kappa^2\Delta t\,\mathbf 1^\top\mathbf H\mathbf 1
\le C\Delta t(1+\mathcal Q_\Delta(y))$. The rounding of coordinate $j$ has
conditional mean zero and variance $\Delta x^2f_j(1-f_j)$, where $f_j$ is
the fractional part of $y_j/\Delta x$, and the roundings of distinct
coordinates are conditionally independent. Given the state after the
shock, the rounding error $\varepsilon$ therefore has mean zero and second
moment at most
$\Delta x\sum_j\mathrm{dist}((Vg)_j,\Delta x\mathbb Z)\le\Delta x\|Vg\|_1$,
since $x$ lies on the grid. Its contribution to $\mathcal Q_\Delta$ is
therefore at most $C'\Delta x\|g\|_1$, which for small $\Delta x$ is absorbed
by half of the term $c\|g\|_1$. The one step inequality
\[
 \E[\mathcal Q_\Delta(\bm x_{k+1})\,|\,\F_k]
 \le(1+C\Delta t)\mathcal Q_\Delta(\bm x_k)+C\Delta t-\tfrac c2\|g_k\|_1-c_2|g_k|^2
\]
and the discrete Gronwall inequality bound
$\sup_{k\Delta t\le T}\E\mathcal Q_\Delta(\bm x_k)$. Summing then gives the
first two bounds in \eqref{ec:eq:variationbound}, because $\mathcal Q_\Delta$
is uniformly equivalent to $|x|^2$ up to constants, and
$(1+C\Delta t)^{-k}(\mathcal Q_\Delta(\bm x_k)+1)$ is a nonnegative
supermartingale whose maximal inequality gives the third. The rounding
errors are martingale differences, so Doob's inequality and the first bound
give the last claim. Stopping preserves every estimate.
\end{proof}

The variation bound is what the control of the interpolation noise
requires, and it does not use the weak interaction condition.
What remains for convergence of the coupled scheme is the identification of
weak limits of the state paths and adjoints with a continuous time
equilibrium, which requires the limiting reflection geometry to be well
posed, and the control of the nonlinear policy iteration and of the outer
budget equations on a fixed grid. On synchronized states, where the scheme
reduces to the one dimensional chain of Section~\ref{ec:diagnostics}, the
convergence rate can be proved.

\begin{proposition}\label{ec:prop:scalarrate}
Let $[\za,\zb]$ and $H$ be the band and the function of
Lemma~\ref{lem:band} for the parameters of the scheme, let $a=\log\za$ and
$b=\log\zb$, and let $\tilde H(x)=H(e^x)$ be the adjoint in log index
coordinates, extended by $\alpha$ below $a$ and by $-\beta$ above $b$. Let
$p_+$ be as in Section~\ref{ec:lattice}. Fix $\underline x<a<b<\overline x$, nodes $x_j=\underline x+j\Delta x$ for
$0\le j\le m$ with $m=\lceil(\overline x-\underline x)/\Delta x\rceil$, and consider
the double obstacle system
\[
 \Psi_0=\alpha,\qquad \Psi_m=-\beta,\qquad
 \Psi_j=\min\big\{\alpha,\max\{-\beta,\Delta t(1-e^{x_j})+e^{-\delta\Delta t}(p_+\Psi_{j+1}+(1-p_+)\Psi_{j-1})\}\big\},
 \quad 0<j<m .
\]
It has exactly one solution, and there is a constant $C$, independent of
all sufficiently small $\Delta t$, such that
$\max_j|\Psi_j-\tilde H(x_j)|\le C\Delta x=C\kappa\sqrt{\Delta t}$. The
piecewise linear interpolant converges uniformly on $[\underline x,\overline x]$
at the same rate.
\end{proposition}

\begin{proof}
The right side defines a contraction with constant $e^{-\delta\Delta t}$ on
the interior values, because clipping is nonexpansive and the weights are
nonnegative and sum to one, so the solution is unique. Let $r_j$ be the
residual of $\tilde H$ in the $j$th equation. Smooth fit makes $\tilde H$
continuously differentiable with a bounded second derivative on each side
of $a$ and of $b$, and on $(a,b)$ it solves
$\varsigma\tilde H''+\nu\tilde H'-\delta\tilde H+1-e^x=0$. At a node whose
neighbors lie in $(a,b)$, Taylor expansion with $2p_+-1=\nu\Delta x/\kappa^2$
and $1-e^{-\delta\Delta t}=\delta\Delta t+O(\Delta t^2)$ cancels the terms of
order $\Delta t$ by the equation and gives $|r_j|\le C_1\Delta x^4$. Clipping
does not enlarge the residual, because $\tilde H(x_j)$ lies between the
obstacles. At the at most six nodes within $\Delta x$ of $a$ or $b$ the
bounded second derivative gives $|r_j|\le C_2\Delta x^2$. At every other
node of the stopping regions the strict inequalities $\za<1-\delta\alpha$
and $\zb>1+\delta\beta$ of Lemma~\ref{lem:band}, that is $1-e^a-\delta\alpha>0$ and
$1-e^b+\delta\beta<0$, make the clipping return the constant obstacle
exactly for small $\Delta x$, so $r_j=0$. Let $\mathbf P$ be the
substochastic matrix of the nearest neighbor walk killed at the two
endpoints and $e_j=|\Psi_j-\tilde H(x_j)|$. Nonexpansiveness gives
$e\le e^{-\delta\Delta t}\mathbf Pe+r$, hence
$e\le\sum_{n\ge0}e^{-n\delta\Delta t}\mathbf P^nr$. A killed transition
probability is at most that of the free walk, and Fourier inversion with
the squared modulus $1-4p_+(1-p_+)\sin^2t$ of the characteristic function
gives
$\sup_{j,k}(\mathbf P^n)_{jk}\le C_3/\sqrt{n+1}$ uniformly for
$p_+\in[1/4,3/4]$, so that
$\sup_{j,k}\sum_ne^{-n\delta\Delta t}(\mathbf P^n)_{jk}\le C_4/\sqrt{1-e^{-\delta\Delta t}}\le C_5/\Delta x$.
Using the row sum bound for the $O(\Delta x^4)$ part of the residual and the
last estimate at each of the at most six nodes with an $O(\Delta x^2)$
residual gives
$\|e\|_\infty\le C_1\Delta x^4/(1-e^{-\delta\Delta t})+6C_2C_5\Delta x\le C\Delta x$.
Interpolation of the nodal error is nonexpansive and the interpolation
error of $\tilde H$ is $O(\Delta x^2)$.
\end{proof}

The proof gives the rate $\sqrt{\Delta t}$, and the errors observed on the
one dimensional chain of Section~\ref{ec:diagnostics} fall faster, so the
estimate need not be sharp. On the diagonal of the three dimensional lattice the saved solution
at $\Delta t=1/64$ agrees with the one dimensional chain to within $2\times10^{-13}$, so
the proposition explains the adjoint errors on synchronized states in
Table~\ref{tab:async}. This agreement is a check at one time step and does
not prove convergence of the coupled scheme or of the portfolios. An audit
of that solution, reported with the code,
evaluates the full residual of every adjoint row, including the revising
rows, and the obstacle and support residuals, all of which are below
$10^{-13}$, as is the residual $8.6\times10^{-14}$ of the cost equations.
Evaluation with outward rounding bounds the discounted probability of
reaching a row with deleted interpolation corners by $8\times10^{-12}$ for
the stored transition matrix after exact normalization. These checks
concern the discrete problem and do not certify the discretization error.